\documentclass[acmsmall]{acmart}
\usepackage{graphicx} % Required for inserting images

\usepackage{macro}
\usepackage{amsmath}

\setcopyright{cc}
\setcctype{by}
\acmDOI{10.1145/3763102}
\acmYear{2025}
\acmJournal{PACMPL}
\acmVolume{9}
\acmNumber{OOPSLA2}
\acmArticle{324}
\acmMonth{10}
\received{2025-03-25}
\received[accepted]{2025-08-12}



\begin{document}

\title{Active Learning for Neurosymbolic Program Synthesis}

\author{Celeste Barnaby}
\orcid{0000-0001-7688-6133}
\affiliation{%
  \institution{University of Texas at Austin}
  \city{Austin}
  \country{USA}
}
\email{celestebarnaby@utexas.edu}

\author{Qiaochu Chen}
\orcid{0000-0003-4680-5157}
\affiliation{%
  \institution{New York University}
  \city{New York}
  \country{USA}
}
\email{qc1127@cs.nyu.edu}

\author{Ramya Ramalingam}
\orcid{0009-0007-6175-6919}
\affiliation{%
  \institution{University of Pennsylvania}
  \city{Philadelphia}
  \country{USA}
}
\email{ramya23@seas.upenn.edu}

\author{Osbert Bastani}
\orcid{0000-0001-9990-7566}
\affiliation{%
  \institution{University of Pennsylvania}
  \city{Philadelphia}
  \country{USA}
}
\email{obastani@seas.upenn.edu}

\author{Işıl Dillig}
\orcid{0000-0001-8006-1230}
\affiliation{%
  \institution{University of Texas at Austin}
  \city{Austin}
  \country{USA}
}
\email{isil@cs.utexas.edu}

\begin{abstract}
%Neurosymbolic program synthesis describes a class of techniques for learning programs that combine symbolic constructs with neural networks. Such techniques have gained recent popularity due to their data efficiency and interpretability (over purely neural approaches) and  because of their flexibility (over purely symbolic methods).  Despite these advantages, establishing user trust in synthesized programs remains challenging.  Prior techniques for interactive program synthesis have shown how to bolster trust in the purely symbolic setting; however, when applied to the \emph{neurosymbolic} context, they often return  an erroneous solution. This paper addresses this issue by proposing a novel active learning framework for neurosymbolic programs. Our approach is based upon a new evaluation strategy called \emph{constrained conformal evaluation (CCE)}, which accounts for  neural mispredictions and user-provided feedback. Our proposed method iteratively makes CCE more precise until all remaining programs are guaranteed to be observationally equivalent. We have implemented this method in a tool called \toolname and experimentally evaluated it on two neurosymbolic domains. Our results demonstrate that \toolname identifies the ground truth program in under 5 rounds of user interaction on average. 

The goal of \emph{active learning for program synthesis} is to synthesize the desired program by asking targeted questions that minimize user interaction. While prior work has explored active learning in the purely symbolic setting, such techniques are inadequate for the increasingly popular paradigm of \emph{neurosymbolic program synthesis}, where the synthesized program incorporates neural components. When applied to the neurosymbolic setting, such techniques can ---and, in practice, do --- return an unintended program due to mispredictions of neural components. 
This paper proposes a new active learning technique that can handle the unique challenges posed by neural network mispredictions.  Our approach is based upon a new evaluation strategy called \emph{constrained conformal evaluation (CCE)}, which accounts for  neural mispredictions while taking into account user-provided feedback. Our proposed method iteratively makes CCE more precise until all remaining programs are guaranteed to be observationally equivalent. We have implemented this method in a tool called \toolname and experimentally evaluated it on three neurosymbolic domains. Our results demonstrate that \toolname identifies the ground truth program for 98\% of the benchmarks, requiring under 5 rounds of user interaction on average. In contrast, prior techniques for active learning are only able to converge to the ground truth program for at most 65\% of the benchmarks.

\end{abstract}

\maketitle

\section{Introduction}\label{sec:intro}

% some stuff about neurosymbolic learning
\emph{Neurosymbolic learning} refers to techniques that combine neural networks and symbolic reasoning to learn new concepts in an interpretable and data-efficient manner.  A particular instance of this framework is \emph{neurosymbolic program synthesis}, which learns programmatic representations that incorporate neural networks~\cite{chaudhuri2021neurosymbolic}. Typically, these programmatic representations are expressed in a neurosymbolic domain-specific language (DSL) consisting of standard programming constructs (e.g., conditionals, combinators etc.) as well as pre-trained neural networks that convert high-dimensional data (e.g., images) to a more structured symbolic representation. In recent years, neurosymbolic program synthesis has found numerous applications in  image processing and search~\cite{barnaby2024photoscout, imageeye, ellis2018learning}, data extraction~\cite{chen2023data, cheng2022binding, verbruggen2021semantic}, and question answering~\cite{gupta2023visual, suris2023vipergpt, chen2021web}.

% challenge: how to trust the results

% prior work on question selection and its shortcomings in the neurosymbolic setting

Although neurosymbolic representations are generally more interpretable than purely neural approaches, a fundamental challenge is ensuring that the learned program is the intended one.  
While the user can inspect the synthesized program to evaluate its correctness, this step can pose a significant hurdle for  end-users who do not possess the necessary programming expertise.
Prior work on program synthesis~\cite{samplesy, learnsy, activelearning1, activelearning2, activelearning3} has aimed to address this issue through \emph{active learning} techniques that interactively query the user to clarify ambiguities. The key idea behind these techniques is to select questions that will minimize the number of user interactions, while ensuring that the correct program is eventually returned by the synthesizer.

% \begin{figure}
%   \begin{center}
%   \vspace{-0.2in}
%     \includegraphics[width=0.99\textwidth,trim={0 25cm 27cm 0cm},clip]{figures/imageeye_example.pdf}
%   \end{center}
%   \vspace{-0.13in}
%   \caption{An input-output example for batch image editing.}\label{fig:img}
%   \vspace{-0.13in}
% \end{figure}
However, when applied to the  neurosymbolic setting, these techniques may return an unintended program. 
{To see why, consider the image editing task from prior work ~\cite{imageeye}, where the goal is to learn a neurosymbolic program that generalizes from a small set of input-output examples to edit a large collection of images. For illustration, suppose the user wants to extract individual photos of people holding baseball bats from a dataset of images, each containing multiple individuals. This task may be achieved using the program in Figure ~\ref{eq:prog}, written in a neurosymbolic DSL.}
{Here, the \textbf{\textsf{Is}} operator is implemented using a neural object detector. The program identifies all \texttt{person} objects in an image that are holding \texttt{baseball\_bat} objects, and applies a crop operation to those individuals. The notion of “holding” is approximated by the spatial relationship \textsf{NextTo}.  To synthesize this program, active learning techniques iteratively query the user for labeled examples (e.g., providing the expected output for a given image) until the system is confident in the inferred logic. Returning to our example, an active learner might ask the user to label the image in Figure~\ref{fig:img} by providing the three cropped photos in Figure~\ref{fig:user-output}. However, for this particular image, a state-of-the-art neural classifier ~\cite{rekognition} fails to detect the baseball bat held by the person in the middle. As a result, when the correct program is executed on the input image, its output does not match the user-provided examples, leading the system to incorrectly reject it as an invalid hypothesis despite it being the desired ground-truth program.}

%In order to synthesize this program, active learning techniques ask the user to label data (e.g., provide the desired output image for a given image) until the system is confident that it has found the intended program. Going back to our example, an active learner may ask the user to label the image shown in Figure~\ref{fig:img} by providing the three cropped photos in Figure ~\ref{fig:user-output}. However, for this particular image, a state-of-the-art neural network for object classification ~\cite{rekognition} fails to detect the baseball bat held by the person in the middle. Thus, when the program from Figure~\ref{eq:prog} is executed on this input image, it will fail to produce the desired edit, as it ``thinks'' that the center person is not next to a baseball bat. As a result, existing active learning approaches will \emph{rule out} this program, even though it is in fact the \emph{desired} ground-truth program.  }

\begin{figure}
    % \small
\begin{minipage}{\textwidth}
\begin{align*}
    % \textsf{Apply}(\textsf{Crop}, \lambda x. \text{ let} \  &y = \textbf{\textsf{Segment}}(x) \text{ in }  \\ 
    % \textsf{Filter}(y, &\textsf{Filter}(y,  \textbf{\textsf{HasAttribute}}( \texttt{person})) \times  \\ &\textsf{Filter}(y,  \textbf{\textsf{HasAttribute}}( \texttt{backpack})), \textsf{HasRelation}(\texttt{contains})))
\{ \textsf{Find}(\textbf{\textsf{Is}}(\textsf{Object}(\texttt{baseball\_bat})), \textbf{\textsf{Is}}(\textsf{Object}(\texttt{person})), \textsf{NextTo}) \rightarrow \textsf{Crop} \}
    \end{align*}
\vspace{-.5cm}
\subcaption{A program in a neurosymbolic image editing DSL.}
\label{eq:prog}
\end{minipage}
\medskip 
\begin{minipage}{0.44\textwidth}
    \includegraphics[width=.7\textwidth]{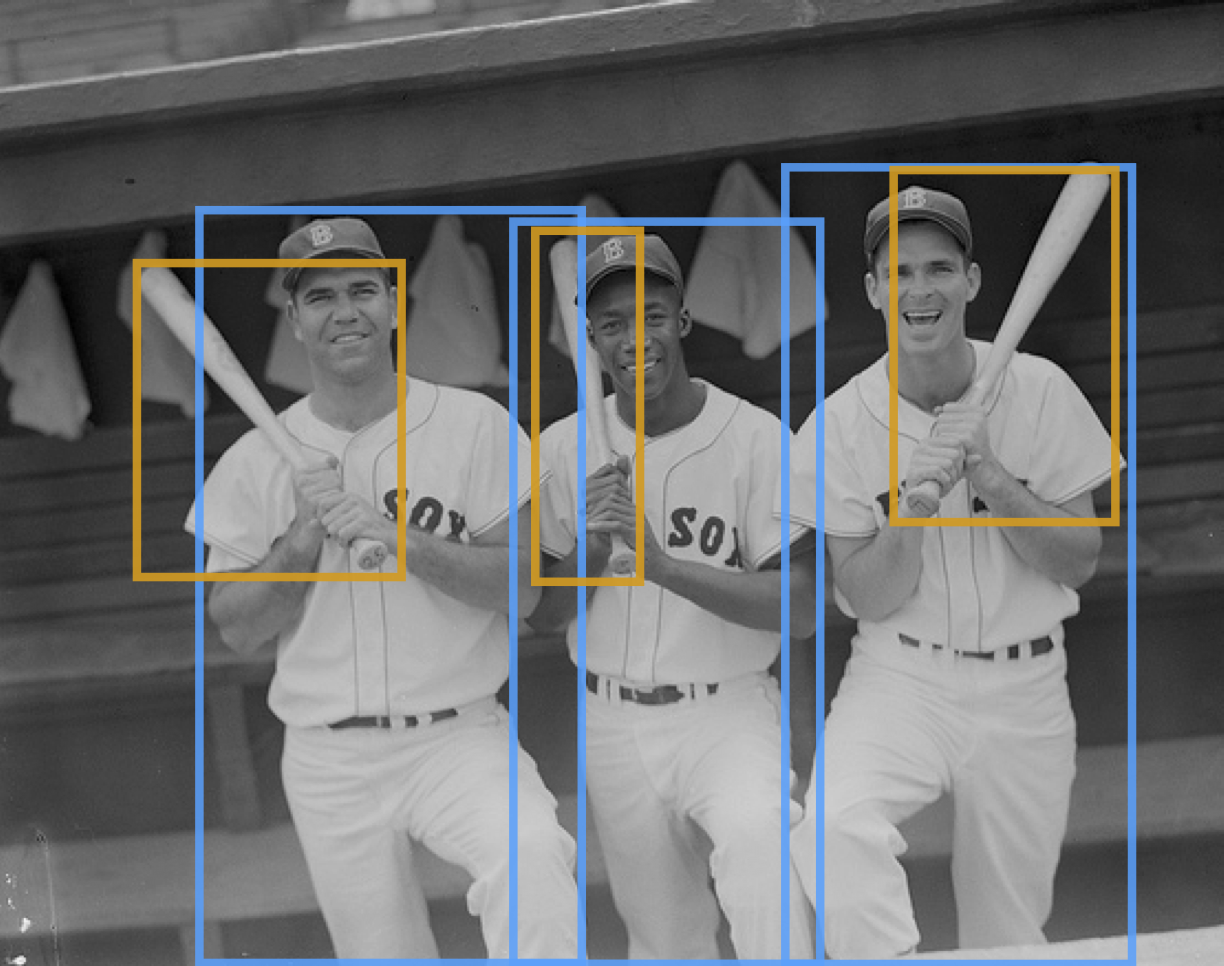}
     \subcaption{Ground truth object detections.}
     \label{fig:img}
\end{minipage}
\begin{minipage}{0.44\textwidth}
    \includegraphics[width=.9\textwidth]{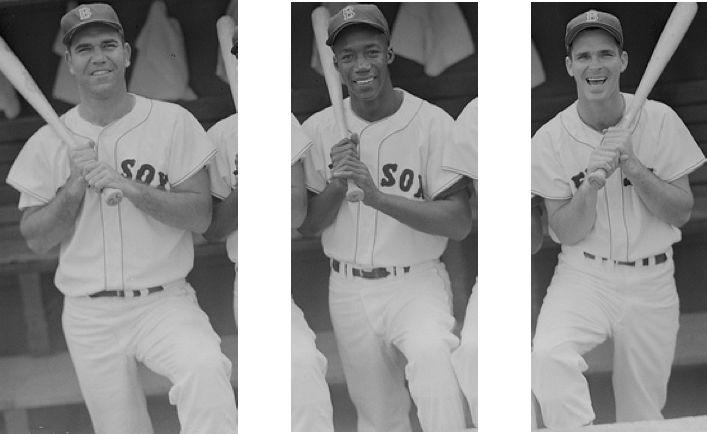}
    \subcaption{Output provided by the user.}
    \label{fig:user-output}
\end{minipage}
 \medskip
\begin{minipage}{0.44\textwidth}
    \includegraphics[width=.7\textwidth]{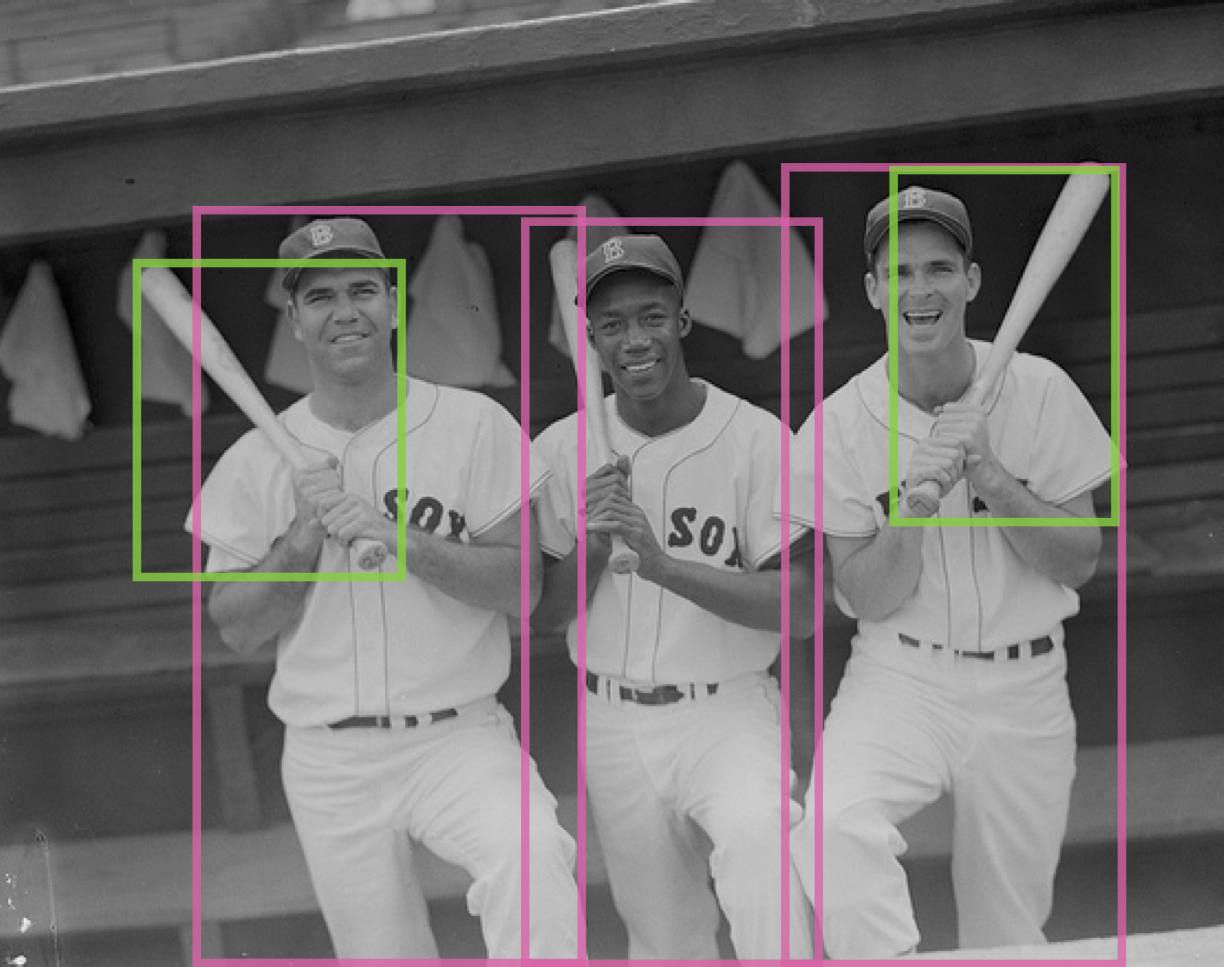}
    \subcaption{Objects detected by the neural network.}
\end{minipage}
\begin{minipage}{0.44\textwidth}
    \includegraphics[width=.72\textwidth]{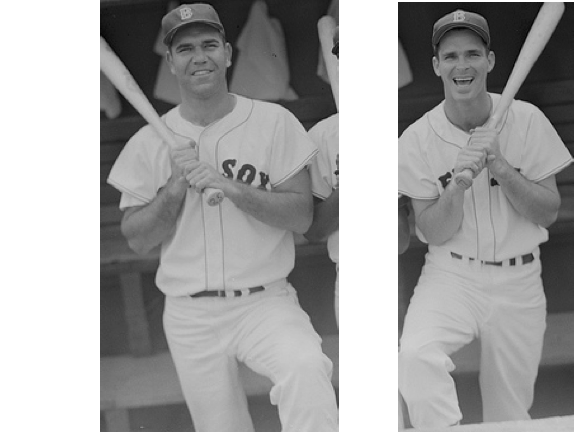}
    \subcaption{Output of the above program.}
\end{minipage}
\vspace{-.2in}
    \caption{(a) A program for cropping people who hold baseball bats. (b) The ground truth object detections, with people in blue and baseball bats in orange. (c) An example output provided by the user. (d) The objects detected by the neural network. (e) The output of the program in (a) when executed on the image. }
    \label{fig:abs_img_example}
    \vspace{-.2in}
\end{figure}

Motivated by this problem, this paper proposes a new active learning framework targeting neurosymbolic programs. Our proposed method deals with possible mispredictions of neural components using \emph{conformal prediction}~\cite{angelopoulos2023conformal, balasubramanian2014conformal}, a principled technique for providing reliable measures of uncertainty for machine learning models. Rather than producing a single prediction, conformal prediction yields a \emph{set} of predictions such that the true label is contained within this prediction set with very high probability. In the context of neurosymbolic programming, we can use these prediction sets to define \emph{conformal semantics} for programs such that every program returns an \emph{output set} that is very likely to contain the ground truth label~\cite{ramalingam2024uncertainty}.  A program $\prog$ is considered to be a possible solution to the learning task if $\prog$ is consistent with the specification under the conformal semantics. The key advantage of this approach is that  it significantly reduces the risk of discarding the correct program despite neural mispredictions.

%A key observation underlying our approach is that we can use \emph{conformal prediction}~\cite{angelopoulos2023conformal, balasubramanian2014conformal} to prevent the desired neurosymbolic program from being pruned from the hypothesis space. At a high level, conformal prediction is a  technique for providing reliable measures of uncertainty for machine learning models' predictions. Rather than producing a single prediction, conformal prediction yields a \emph{set} of predictions such that the true label is contained within this prediction set with very high probability. In the context of neurosymbolic programming, we can use these prediction sets to define \emph{conformal semantics} for programs such that every program returns an \emph{output set} that is very likely to contain the ground truth label~\cite{ramalingam2024uncertainty}. 

When using conformal semantics, active learning becomes even more critical because the inherent flexibility of returning a \emph{set} of predictions leads to an explosion in the number of potential solutions. A neurosymbolic synthesis problem can easily admit an intractable number of programs consistent with the conformal semantics, making it impractical for users to manually evaluate all possibilities. The role of active learning is to guide the user through a sequence of targeted questions that gradually refine the hypothesis space of programs by narrowing down the conformal prediction sets. Each round of interaction refines the program's conformal semantics, either by enhancing the task specification or by confirming the ground-truth label for a neural prediction. Thus, as user interaction progresses, program evaluation needs to account for both the results of conformal prediction as well as any user feeback provided thus far. We refer to this type of program evaluation strategy as \emph{constrained conformal evaluation (CCE)}.

%Based on this discussion, we consider a program $\prog$ to be a solution to the learning task if $\prog$ is consistent with the specification under the conformal semantics.
%However,  a synthesis problem  often admits intractably many solutions under this definition, so it is not feasible to ask a user to comb through such an enormous space of solutions. Thus, the goal of active learning 
%is to ask the user a sequence of questions that make the conformal semantics more precise 
%until all remaining programs are ``indistinguishable.''  Each round of user interaction refines the hypothesis space either by strengthening the specification of the learning task \emph{or} by supplying the ground-truth label for a neural prediction. Thus, as user interaction progresses, program evaluation needs to account for both the results of conformal prediction as well as any user feeback provided thus far. We refer to this type of program evaluation strategy as \emph{constrained conformal evaluation (CCE)}. 

\begin{wrapfigure}{r}{0.47\textwidth}
    % \centering
    \begin{center}
    \vspace{-0.35in}
    \includegraphics[scale=0.15,trim={13cm 13cm 10cm 10cm},clip]{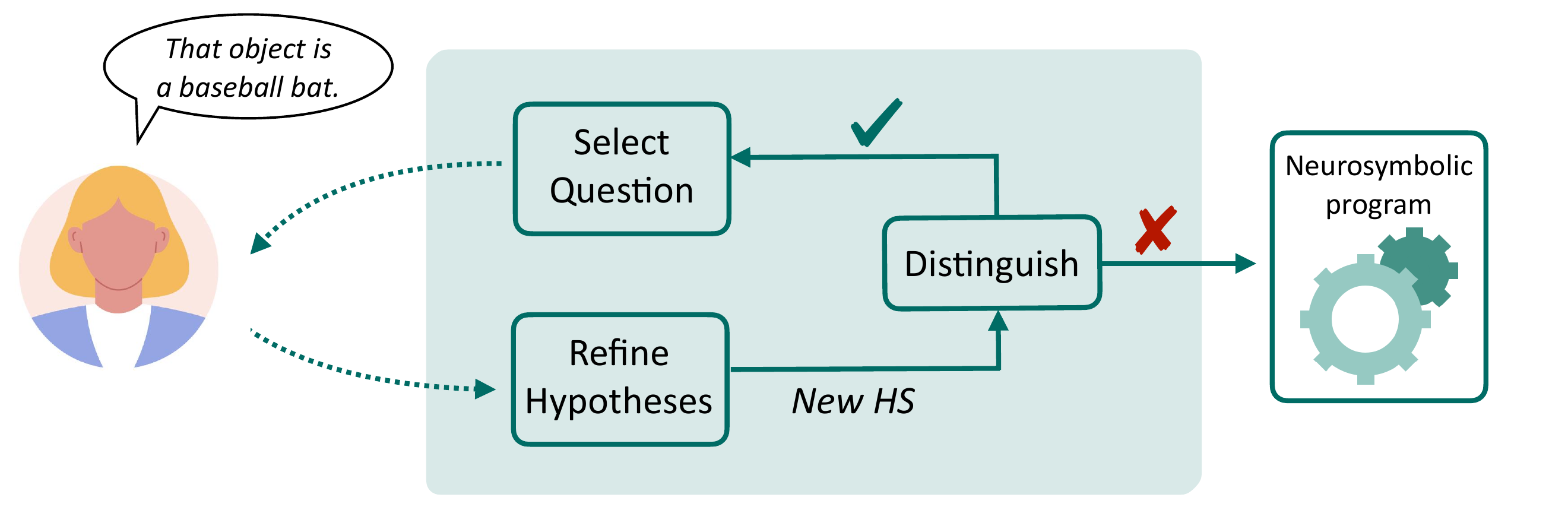}
    \end{center}
    \vspace{-0.15in}
    \caption{Overview of our approach.}
    \label{fig:overview}
    \vspace{-0.2in}
\end{wrapfigure}

As illustrated in Figure~\ref{fig:overview}, each round of user interaction chooses a question to ask to the user, with the goal of shrinking the hypothesis space as much as possible. In each round of user interaction, programs are eliminated from the hypothesis space either because the user provides a new input-output example or because the CCE results have become more precise. Upon refinement of the hypothesis 
space, a \emph{distinguishability} check is performed to identify pairs of programs that \emph{could} have different CCE results if we continued to query the user. If such programs exist, user interaction continues and active learning identifies a new question. Active learning terminates when all remaining programs in the hypothesis space are observationally equivalent.
%to the desired program under the {ground truth semantics}. 

A key challenge in realizing the active learning framework from Figure~\ref{fig:overview} is that its key components (namely, Select Question, Refine Hypotheses, and Distinguish) require performing CCE many times, which can be very expensive. In particular, under conformal semantics,  expressions evaluate to \emph{sets of values}, so CCE can require an exponential number of computations with respect to the cardinality of these sets. Our method addresses this problem using two key ideas. First, it uses bidirectional abstract interpretation to perform CCE in a practical way. At a high level, the idea is to use the output examples provided by the user to infer abstract values for program sub-expressions. These abstract values can then be used to filter infeasible conformal prediction results, leading to smaller sets. Second, when selecting what questions to ask the user, our method uses a bounded form of conformal evaluation to over-approximate the pruning power of each question, resulting in a much more practical question selection algorithm.

We have implemented the proposed approach in a tool called \toolname and evaluated it on 112 benchmarks spanning
{three neurosymbolic domains: (1) batch image editing, (2) visual arithmetic, and (3) image search for visual concept learning.}
%two neurosymbolic domains:  a list-processing DSL that performs computations over digit images, and an image-editing DSL that applies manipulations to specific objects within an image~\cite{imageeye}. 
A key highlight of our evaluation is that \toolname can successfully identify the desired program for $98\%$ of the benchmarks, whereas prior active learning techniques that do not use conformal semantics fail to produce the desired program for around 35\% of the benchmarks. Another key highlight of our evaluation is that \toolname converges to the desired program in under 5 rounds of interaction on average. %whereas a random question selection baseline requires almost 94 rounds of user interaction. 
Additionally, several ablation studies confirm the effectiveness of our proposed algorithmic optimizations.

To summarize, this paper makes the following key contributions:
\begin{itemize}[leftmargin=*]
    \item We define the \emph{neurosymbolic active learning} problem and propose the first algorithm for solving~it. 
    \item We introduce \emph{constrained conformal evaluation (CCE)} as a new type of program semantics that takes into account both user feedback and prediction sets obtained through conformal prediction. Furthermore, we present a practical CCE technique that uses bidirectional abstract interpretation. 

    \item We implement our approach in a new tool called \toolname and instantiate it on three neurosymbolic domains. Our experiments on 112 tasks show that \toolname can identify the desired program for 98\% of the benchmarks, requiring an average of 4.9 rounds of user interaction.  
\end{itemize}

% our solution, maybe together with a diagram

% something about results and 
\section{Overview}\label{sec:overview}

{In this section, we provide a high-level overview of our approach using a  simple  example.
%Inspired by applications of program synthesis in image editing~\cite{imageeye}, visual question answering~\cite{gupta2023visual}, and image search~\cite{barnaby2024photoscout}, we will consider a synthesis task that requires performing a computation over  a list of hand-written digit images. 
Specifically, suppose that a user is examining a hand-written ledger and wants to count the number of transactions greater than 5 dollars. This task may be automated using the program $\prog_1$ in Figure~\ref{fig:hypspace}, which uses standard combinators such as \texttt{map}, \texttt{fold}, and \texttt{filter}, along with a neural perception component, \mnistpred, which converts an image to a digit. Consider the input list $\inp$ in Figure~\ref{fig:io}. Note that the true label for $x_2$ in $\inp$ is ambiguous; if the user believes the label is 0, then the desired output is 2. However, if \mnistpred\ classifies $x_2$ as 8, then $\prog_1(\inp)$ will return 3, disqualifying it as a solution. In this case, the synthesizer might fail to return any program, or it may produce an erroneous program that doesn't align with the user's expectations.
}

%In this section, we give a high-level overview of our approach with the aid of a simple example. Motivated by recent applications of program synthesis in image editing~\cite{imageeye}, visual question answering~\cite{gupta2023visual}, and image search~\cite{barnaby2024photoscout}, we consider a program synthesis task involving images, which necessitates the use of neural networks for perception. In particular, suppose we are given a list of images that correspond to hand-written digits, and we wish to perform some computation over them using a list manipulation DSL from prior work~\cite{feser2015synthesizing}.  In addition to the standard list combinators \texttt{map}, \texttt{fold} and \texttt{filter}, this neurosymbolic DSL also includes a neural {perception} component \mnistpred\ that converts the input image to a digit. However, since perception is inherently imperfect, \mnistpred\ can sometimes output the wrong digit. For example, given the input image from Figure~\ref{fig:digit}, it may erroneously output 2 instead of the ground truth value 4. 

\begin{figure}[!t]
% \vspace*{-.7cm}
\begin{minipage}{.35\textwidth}
     \centering
    \includegraphics[scale=0.3, trim={0 33cm 47cm 0},clip]{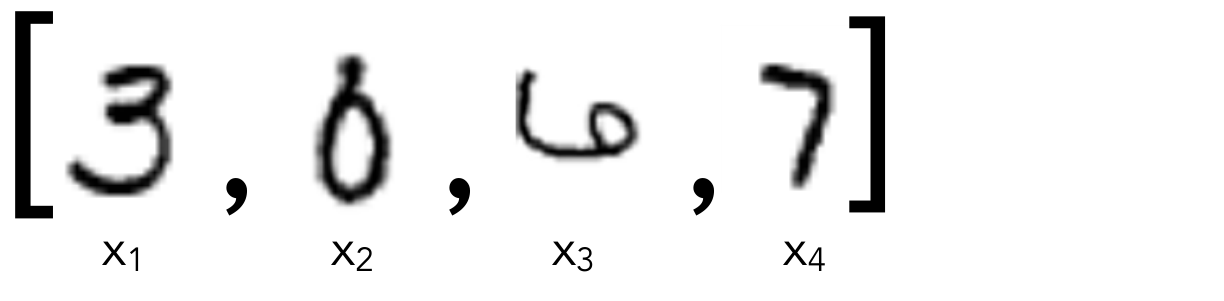}
    \caption{List of hand-drawn digits.}
    \label{fig:io}
\end{minipage}%
\begin{minipage}{.65\textwidth}
 \centering
 \vspace{-.1in}
\footnotesize
\begin{align*}
\prog_1 :=& \ \lambda  \ l . \ \texttt{fold} \  \texttt{inc} \ 0  \ (\texttt{filter} \ (\lambda x. \ x > 5 ) (\texttt{map} \ \mnistpred \ l))  \\ 
\prog_2 :=& \ \lambda  \ l . \ \texttt{fold}  \ \texttt{inc} \ 0  \ (\texttt{filter} \ (\lambda x. \ 6 < x < 9 ) (\texttt{map} \ \mnistpred \ l)) \\ 
\prog_3 :=& \ \lambda  \ l . \ \texttt{fold} \  \texttt{inc} \ 0  \ (\texttt{filter} \ (\lambda x. \ x < 4 ) (\texttt{map} \ \mnistpred \ l)) \\ 
\prog_4 :=& \ \lambda  \ l . \ \texttt{fold}  \ \texttt{inc}  \ 0  \ (\texttt{filter} \ (\lambda x. \ 2 < x < 8 ) (\texttt{map} \ \mnistpred \ l))
\end{align*}
\vspace{-0.2in}
\caption{Hypothesis space $\progs$.}
\label{fig:hypspace}
\end{minipage}%
\vspace{-.2in}
\end{figure}

%Now, consider the synthesis task wherein the user demonstrates their desired computation through a set of input-output examples. 

%Unfortunately, when performing synthesis, we encounter a roadblock if \mnistpred\ incorrectly classifies one of the digits. For example, if $x_2$ is classified as an 8 instead of 0, then $\prog_1(\inp)$ will be 3. Thus, $\prog_1$ will not be considered a solution to the user's synthesis problem, even though this program clearly matches their intent. In this case, the synthesizer may fail to return any program, or it may return a program that does not match the user's intent.

%Instead, the synthesizer will erroneously return the following program: 
%\[
%\prog_2 := \ \lambda  \ l . \ \texttt{fold}  \ (\texttt{filter} \ (\lambda x. \ x > 6 ) (\texttt{map} \ l \ \mnistpred)) \ 0 \  \texttt{sum} \\ 
%\]

\vspace{-0.05in}
\paragraph{\textbf{Solution: conformal semantics.}} As discussed in Section~\ref{sec:intro}, a viable solution is to use \emph{conformal prediction (CP)}~\cite{angelopoulos2023conformal} to ensure the ground truth label is included in the result with high probability. CP generates prediction sets for  unlabeled data by calculating nonconformity scores based on previously labeled data, determining which labels to include to maintain a specified confidence level. In our example, given the input list from Figure~\ref{fig:io}, CP would produce the prediction sets [\{3\}, \{0, 8\}, \{6, 9\}, \{7\}], indicating the model's uncertainty about the second and third digits. This uncertainty impacts the output of $\prog_1$ from Figure~\ref{fig:hypspace}: if the second digit is $0$, $\prog_1$ returns $2$; if it is $8$, it returns $3$. Thus, under \emph{conformal semantics}, the output of $\prog_1$ is the set $\{2, 3\}$. Since this set includes the user-provided label $2$, $\prog_1$ remains consistent with the input-output example and is not pruned. %from the hypothesis space.

\vspace{-0.05in}
\paragraph{\textbf{Need for active learning.}} 
While the intended program is no longer pruned from the solution space, a new problem arises: \emph{many} other programs are also consistent with the input-output examples under conformal semantics. 
 In our running example, any other program whose output set contains $2$ (e.g., $\prog_2$, $\prog_3$,  $\prog_4$ in Figure~\ref{fig:hypspace}) is also a potential solution. 
This issue already exists in standard Programming-By-Example (PBE) scenarios, where input-output examples do not fully capture the program's intended behavior. To deal with this issue, prior work has proposed \emph{active learning} techniques that help users resolve such ambiguity through user interaction~\cite{activelearning1,activelearning2}. However, these techniques  only deal with  ambiguity in the specification. In the neurosymbolic setting,  another significant source of ambiguity involves the ground truth labels of neural components. Existing techniques for active learning fail to handle this second class of ambiguities and may prune the intended program. In the remainder of this section, we explain how our proposed method interacts with the user to correctly resolve both types of ambiguities.

\vspace{-0.05in}
\paragraph{\textbf{Checking ambiguities.}} {We first need a method to determine whether there are any remaining ambiguities in the solution space.} Intuitively, the solution space $\progs$ is \emph{ambiguous} if it contains a pair of \emph{distinguishable} programs $(\prog, \prog')$, meaning they could produce different outputs under the \emph{ground truth semantics}. For example, in the hypothesis space $\progs$ from Figure~\ref{fig:hypspace}, the pair $(\prog_1, \prog_2)$ is distinguishable because $\prog_1$'s output set on $\inp$ is $\{2, 3\}$, while $\prog_2$'s output set is $\{1, 2\}$. Thus, it is possible that $\prog_1(\inp) = 2$ and $\prog_2(\inp) = 1$, meaning they differ under the ground truth.  $\prog_1$ and $\prog_4$ are also distinguishable even though both produce $\{2, 3\}$ under conformal semantics, as  $\prog_1$ could produce $2$ under the ground truth semantics, while  $\prog_4$ could produce $3$.

%Based on the above discussion, we  need a way of determining whether there are any remaining ambiguities. We say that the solution space $\progs$ is \emph{ambiguous} if it contains a pair of \emph{distinguishable} programs $(\prog_1, \prog_2)$, meaning that they \emph{could} output different values under the ground truth semantics.  As an example, consider the  hypothesis space $\progs$ shown in Figure~\ref{fig:hypspace}.
%Here, the pair $(\prog_1, \prog_2)$ is distinguishable because $\prog_1$'s output set on $\inp$ is $\{2, 3\}$ and $\prog_2$'s output set on $\inp$ is $\{1, 2\}$, so it is possible that $P_1(I)=2$ and $P_2(I) = 1$, meaning that they differ under the ground truth semantics. Similarly, $P_1$ and $P_4$ are also distinguishable even though both produce the output set $\{2,3\}$ under the conformal semantics: this is because $\texttt{\textbf{toDigit}}$ could have a ground truth semantics that would cause $P_1$'s ground truth output to be $2$ and $P_4$'s ground truth output to be $3$. 

\vspace{-0.05in}
\paragraph{\textbf{Question selection.}} Since the solution space in our example is ambiguous, we query the user. User queries can either request a new input-output example for the target function or ask the user to label an image. Our approach selects a question $q$ based on its \emph{pruning power}—the fraction of programs that will be pruned from the hypothesis space for the worst possible answer to $q$. For instance, consider a question $q_1$ asking the user to label $x_2$ from  Figure~\ref{fig:io}. If the user labels $x_2$ as $0$, $\prog_2$ will be pruned from the program space, while if the user labels $x_2$ as $8$, $\prog_1$ and $\prog_3$ will be pruned. Hence, $q_1$ has a pruning power of $\mathsf{min}(0.25, 0.5) = 0.25$. In contrast, consider question $q_2$ asking the user to label $x_3$. If the answer is $9$, then all programs can still output $2$, meaning that $q_2$ has a pruning power of 0.  Thus, $q_1$ is the preferable question.

%Since the solution is ambiguous in our example, we need to question the user to resolve this ambiguity. User queries can either require the user to provide a new input-output example for the target function or to label one of the images. Our approach chooses a question $q$ based on its \emph{pruning power} -- that is, the fraction of programs that will be pruned from the hypothesis space for the worst answer to $q$. As an example, consider a question $q_1$ that asks the user to label list element $x_2$ from Figure~\ref{fig:io}. If the user were to label $x_2$ with $0$, then every program would still be able to output $2$, and we would hence \emph{not} prune \emph{any} programs. Thus, $q_1$'s pruning power is $0$. On the other hand, consider an alternative question $q_2$ that asks  the user to label $x_3$. This question has two possible answers (namely $6$ and $9$) and either of these answers would cause one program to be pruned from the hypothesis space. Hence, the pruning power of $q_2$ is $0.25$, and our technique would therefore favor $q_2$ over $q_1$.

\vspace{-0.05in}
\paragraph{\textbf{Hypothesis space refinement.}} After the user responds to $q_1$ with $0$,  we need to determine which programs in the hypothesis space are still viable. To do this, we evaluate each program $\prog \in \progs$ under the conformal semantics while incorporating all user feedback—a process called \emph{constrained conformal evaluation (CCE)}. For instance, before the user answers $q_1$, the result of CCE on $\prog_2$ for $\inp$ is $\{1, 2\}$, as $x_2$ could be between $6$ and $9$. After the user labels $x_2$ as $0$, the result becomes $\{1\}$. Since the output no longer includes $2$, $\prog_2$ is inconsistent with the input-output example and is pruned from the hypothesis space. {At this point, however, there is still ambiguity because $P_1, P_3,$ and $ P_4$ can all output $2$ but they are not equivalent. Hence, our method repeats this process of checking for ambiguities and generating new questions until all remaining programs are equivalent. }

%\osbert{Should we say this continues until we only have a single program left?}

%Now, suppose that we proceed with question $q_2$, and the user responds with   $6$. At this point, we need to determine which programs in the hypothesis space are still consistent with the user feedback.  To do so, we need to evaluate every program $\prog \in \progs$ under the conformal semantics, but taking into account all user feedback provided so far -- we refer to this as \emph{constrained conformal evaluation (CCE)}. For example,  the result of performing CCE on $\prog_4$ for $\inp$ is $\{2,3\}$ before the user answers $q_2$ (since $x_3$ \emph{may or may not} be between $4$ and $8$), but it becomes the singleton set $\{3\}$ after the user labels $x_3$ as $6$. Because the output no longer contains $2$, this program is inconsistent with the input-output examples and is hence pruned from the hypothesis space. 

\vspace{-0.05in}
\paragraph{\textbf{Challenges.}} {As discussed above, our procedure involves three main components: checking for ambiguities, question selection, and hypothesis space refinement. Among these  components, checking for ambiguities often takes negligible time because we just need to find a single pair of programs that are distinguishable (see Section ~\ref{sec:runtime}). Meanwhile, hypothesis space refinement and question selection take considerable time due to the need to perform CCE on all pairs of programs and inputs. Unfortunately, performing CCE is computationally quite expensive -- for instance, computing the conformal output of $\prog_1(\inp)$ requires evaluating $\prog_1$ on four possible combinations: $\{3\} \times \{0,8\} \times \{6,9\} \times \{7\}$. While this is manageable for small examples, it becomes impractical for larger inputs or prediction sets. For example, with list of 5 images each having a prediction set of size 4, conformal evaluation would need to consider 1024 possibilities. Our algorithm deals with this challenge through two key algorithmic optimizations, one primarily targeting hypothesis space refinement and the other targeting question selection. Next, we give an overview of these two key optimizations.}

\paragraph{\textbf{Practical CCE for hypothesis space refinement}} 
{To make  hypothesis space refinement practical, we leverage a key observation: \emph{the primary use of CCE is to check whether or not a program is consistent with the input-output examples.} Our method uses this observation to make CCE more scalable. By reasoning backwards from the desired output, we can infer constraints on subexpressions and thereby to reduce the size of the prediction sets during conformal evaluation. }
%\emph{necessary} conditions on the outputs of symbolic operators and uses them to reduce the size of the prediction sets during conformal evaluation. 

\begin{wrapfigure}{r}{0.35\textwidth}
\vspace{-0.2in}
\begin{minipage}{.34\textwidth}
    \centering
    \includegraphics[scale=0.17, trim={0, 22.5cm, 47cm, 0},clip]{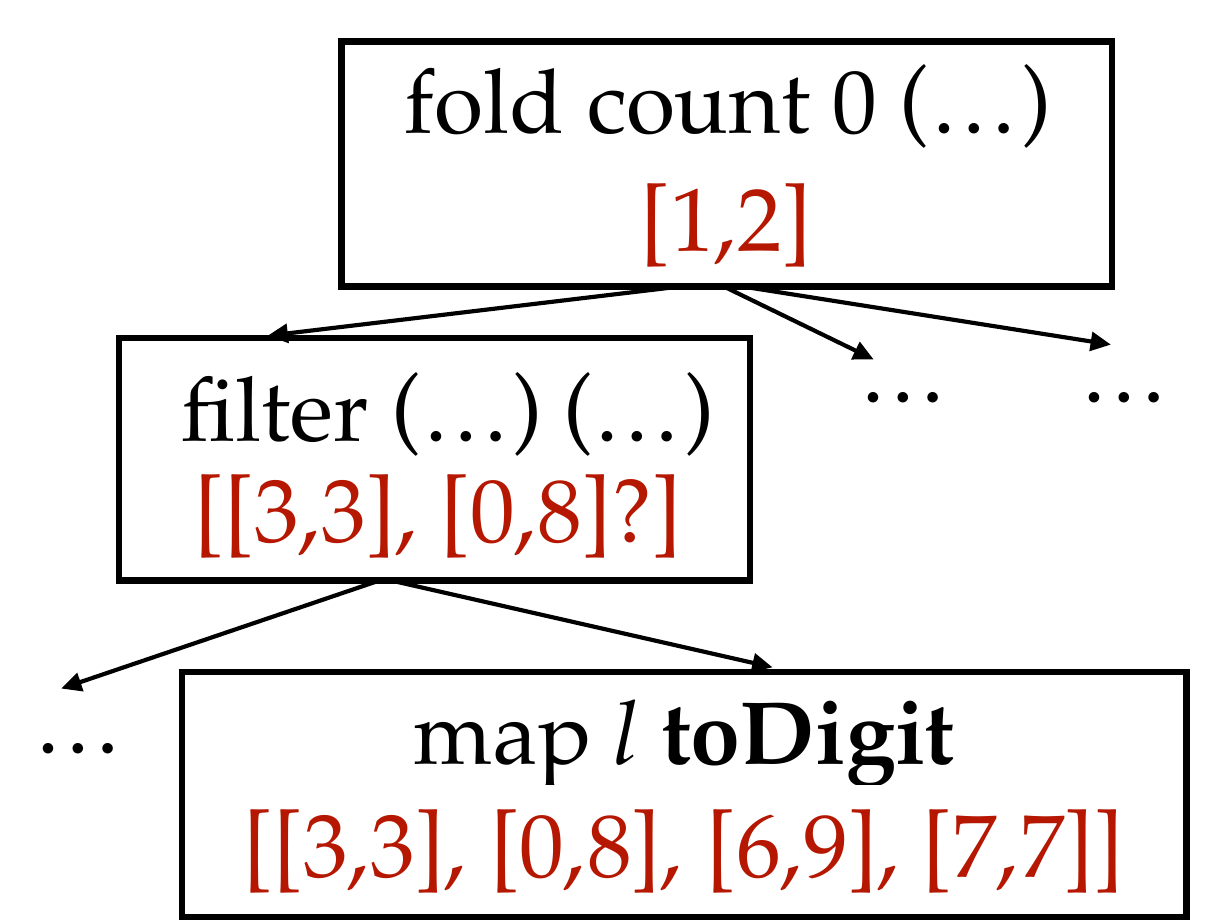}
    \vspace{-0.1in}
    \caption{The annotated AST of $\prog_3$ after forward abstract interpretation.}
     \label{fig:ast1}
\end{minipage} 
\vspace{-0.1in}
\end{wrapfigure}

As an example, consider checking whether  $\prog_3$ is consistent with the user's IO example. To perform CCE more efficiently, our method first performs abstract interpretation in the forward direction to compute abstract values for each sub-expression, as shown in the annotated    abstract syntax tree (AST) in Figure~\ref{fig:ast1}.  Here, each  node is annotated by its abstract value $[\alpha_1, \ldots, \alpha_n]$, where each $\alpha_i$ is either an interval or an \emph{optional} interval $[l, u]?$, indicating the value may or may not exist at that position in the list. 
The leaf node annotations reflect the conformal prediction  results (e.g., the set $\{0, 8\}$ is abstracted as $[0, 8]$), and  abstract values  of internal nodes are obtained by applying abstract transformers (e.g., resulting in the abstract value $[1, 2]$ for the root).

\begin{wrapfigure}{r}{0.35\textwidth}
\vspace{-0.2in}
\begin{minipage}{.34\textwidth}
    \centering
    \includegraphics[scale=0.17, trim={0, 22.5cm, 46cm, 0},clip]{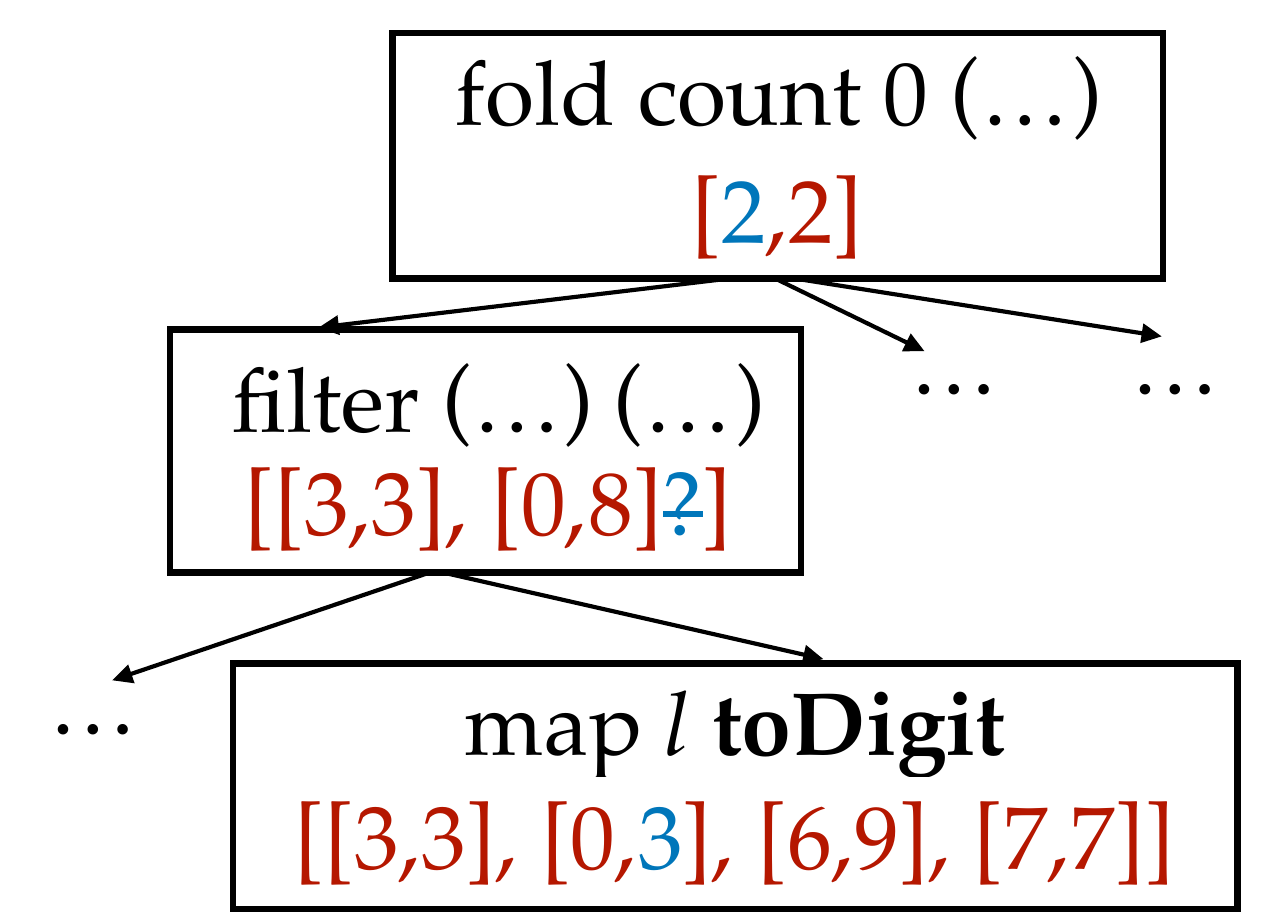}
    \vspace{-0.1in}
    \caption{The annotated AST of $\prog_3$ after backward abstract interpretation.}
    \vspace{-0.12in}
     \label{fig:ast2}
\end{minipage}%
\end{wrapfigure}

Next, our method performs backward reasoning from the output  to tighten these annotations. For instance, in order for $\prog_3$ to output $2$, there must be exactly $2$ items  less than $4$ in the input. Since $x_3$ and $x_4$ are certainly $\geq4$, we can infer $x_2$ must not have been filtered; hence $x_2 < 4$.
As a result, we obtain the strengthened annotations shown in Figure~\ref{fig:ast2}. Intuitively, these annotations correspond to necessary conditions that each element  must satisfy in order for the program to be consistent with the input-output examples. Thus, when performing CCE, we can  filter evaluation results that do not satisfy the inferred necessary condition. For example, when evaluating $\prog_3$ on $\inp$,  we need only consider the $2$ possible ground truth values in $\{3\} \times \{0\} \times \{6,9\} \times \{7\}$. In practice, this approach significantly reduces the sizes of output sets when performing CCE.
%and is crucial for the practicality of our approach. 

\paragraph{\textbf{Minimizing CCE for question selection.}}  {Recall that, during question selection, we need to compute the pruning power of each question, which naïvely requires performing CCE on all (program, input) pairs. To make question selection more practical, we leverage the following key observations: \emph{(1) Any evaluation strategy that under-approximates the CCE result can be used to derive an upper bound on the pruning power of a question. (2) If the upper bound for a question $q$ is lower than that of the true pruning power of a previously encountered question, we can discard it and avoid performing CCE for $q$. }  Based on these observations, our method first uses a more lightweight method, called \emph{bounded conformal evaluation (BCE)}, to derive an upper-bound on the pruning power  of every question and uses them to  reduce the set of  questions for which CCE must be performed.

For instance, consider again question $q_2$ asking the user to label the ground truth value of $x_3$.  To compute an upper bound on $q_2$'s pruning power in a lightweight manner, we perform a bounded form of CCE  where we restrict intermediate results to sets of  size at most $k$ through sampling. For instance, setting $k=1$ in this example and sampling the value $0$ from the prediction set of $x_2$, we determine that the only program that could be pruned by $q_2$ is $P_4$, which gives an upper bound of $0.25$ on the pruning power of $q_2$. Now, if we have previously computed the true pruning power of $q_1$ as $0.25$, we need not consider $q_2$, since it cannot possibly be a strictly better question than $q_1$. As we show experimentally, this  strategy greatly reduces user interaction time.}

\section{Problem Formulation}\label{sec:probform}

% In this section, we first introduce neurosymbolic programs and conformal semantics, and then formulate the active learning problem addressed in the next section.
\subsection{Neurosymbolic Programs}

We consider a {family} of \emph{neurosymbolic} programming languages defined by the \emph{meta-grammar} in Figure~\ref{fig:metalanguage}, which is parametrized over symbolic functions $\symfuncs$ and neural networks $\nns$.  Every program defines a function called $\synthfunc$ in a functional language with expressions, let statements, and sequencing. Expressions include conditionals of the form $\ite{E}{E}{E}$ as well as compositions of symbolic and neural components. In particular, $\symfunc(E_1, \ldots, E_n)$ denotes the application of a symbolic function $\symfunc \in \symfuncs$ to expressions $E_1, \ldots, E_n$. On the other hand,  $\nn(x)$ performs forward propagation on input $x$ using a neural network $\nn \in \nns$. 

\begin{wrapfigure}{r}{0.49\textwidth}
\vspace{-0.2in}
\small
\begin{minipage}{.48\textwidth}
\[
\begin{array}{lll}
    P & := &\mathsf{def} \ \synthfunc(x) =  S  \\ 
    S & := & E \ | \ \texttt{let } x = E \ | \ S; S \\ 
    E & := & c \ | \ x \ | \ \symfunc(E, \ldots, E) \ | \  \nn(x) \ | \ \ite{E}{E}{E} 
\end{array}
\]
\begin{align*}
\symfunc \in \mathsf{Symbolic \ functions} \quad \quad \nn \in \mathsf{Neural \ networks}
\end{align*}
\end{minipage}
\vspace{-0.1in}
\caption{Meta-grammar for DSL. }
\vspace*{-0.1in}
\label{fig:metalanguage}
\end{wrapfigure}

While the distinction between neural and symbolic components may seem minor in terms of syntax, it creates a significant difference in semantics. Symbolic functions have precisely defined semantics and always produce \emph{accurate} results. In contrast, neural components $\nn$ \emph{approximate} an oracle function $\hat{\nn}$ and  may produce inaccurate results for some inputs (i.e., there may exist an input $x$ such that $\nn(x) \neq \hat{\nn}(x)$). To capture this difference, we define two types of semantics for neurosymbolic DSLs: \emph{evaluation} semantics $\stansem{P}$ and \emph{ground truth} semantics $\gtsem{P}$. These differ only in their treatment of neural components: $\stansem{\nn}(x)$ represents the output of a forward pass through the neural network $\nn$, while $\gtsem{\nn}(x)$ consults an external oracle $\hat{\nn}$ to return the \emph{true label} for $x$.

%While the distinction between neural and symbolic components may seem minor in terms of syntax, there is a significant difference  in terms of  semantics: Symbolic functions have precisely defined semantics and thus always produce \emph{accurate} results. On the other hand, neural components $\nn$ \emph{approximate} some \emph{unknown} computation $\hat{\nn}$, so they may produce inaccurate results on some inputs (i.e., there may exist an input $x$ for which $\nn(x) \neq \hat{\nn}(x)$).   To capture this distinction, we differentiate between the \emph{evaluation} and \emph{ground truth} semantics of neurosymbolic DSLs, denoted as  $\stansem{P}$ and $\gtsem{P}$ respectively. 
 % These two semantics only differ with respect to the neural components. Specifically, given an input $x$, $\stansem{\nn}(x)$ yields the result  of performing forward propagation through the layers of a neural network $\nn$ starting with input $x$. In contrast, $\gtsem{\nn}(x)$ consults an external oracle $\hat{\nn}$ to obtain the ground truth label for $x$.  

\subsection{Neurosymbolic Learning and Conformal Semantics}\label{sec:conformal}

This paper is concerned with \emph{neurosymbolic learning} from examples. That is, given a set  $\inout = \{(\inp, \out)\} \subseteq \inps \times \outs$ (where $\inps$ is an input space and $\outs$ is an output space), we wish to learn a program $\prog$ in a  neurosymbolic DSL such that $\prog$ is consistent with $\inout$ under the \emph{ground truth semantics}, i.e.:
{\small 
\begin{equation}\label{eq:NSLE-gt}
\textstyle
\forall (\inp, \out) \in \inout. \ \gtsem{\prog}(\inp) = \out
\end{equation}
}
However, since the ground truth semantics requires consulting an external oracle (e.g., a human), we cannot easily check whether a program $\prog$ constitutes a valid solution to the learning problem. {As a consequence, running a traditional synthesizer with the ground truth semantics would require the user to provide an enormous number of ground truth labels.}

A natural strategy for minimizing this workload is to leverage active learning. 
%\osbert{Maybe we should cite the existing active learning papers here.}
To do so, we need  to identify  programs that are consistent with the input examples $\inout$ and ground truth labels for perception. In this section, we focus on the former and discuss the latter later. A na\"{i}ve approach is to approximate the ground truth semantics via the evaluation semantics, and apply an existing active learning strategy for PBE.
That is, we could formulate the learning problem as that of finding a program $\prog$ that is consistent with $\inout$ under the evaluation semantics, such that
{\small 
\begin{equation}\label{eq:NSLE-eval}
\forall (\inp, \out) \in \inout. \ \stansem{\prog}(\inp) = \out
\end{equation}
}
However, as discussed in previous sections, this alternative formulation  has a significant drawback: because the evaluation semantics may be incorrect, a solution to Equation~\ref{eq:NSLE-gt} may \emph{not} be a solution to Equation~\ref{eq:NSLE-eval}. This means  that the alternative problem formulation would inherently allow incorrect solutions, while ruling out the desired solution. 
%the ground truth program $\prog$ may not produce the provided output for one of the inputs. As a result, this formulation may rule out the ground truth program from being returned as a solution.
%that the correct program would never be returned as a solution.

Recent work has proposed to use \emph{conformal prediction} to address this problem~\cite{ramalingam2024uncertainty}. At a high level, conformal prediction ~\cite{angelopoulos2023conformal, shafer2008tutorial} replaces each neural component $\nn:\mathcal{X}\to\mathcal{Y}$ with a \emph{conformal predictor} $\tilde{\nn}:\mathcal{X}\to2^{\mathcal{Y}}$ that predicts sets of labels $\confout \subseteq\mathcal{Y}$.  Intuitively, $\tilde{\nn}(x)$ represents the set of all plausible labels for input $x$. More precisely, conformal prediction aims to ensure that the ground truth label is contained in the prediction set.
This property is known as \emph{coverage},  and prior work in the conformal prediction literature~\cite{ramalingam2024uncertainty} can provide coverage with bounded probability under the standard i.i.d. assumption. 
In the remainder of this paper, we assume access to a conformal predictor that can guarantee coverage, and we utilize such a predictor to define the \emph{conformal semantics} of any neurosymbolic DSL (formalized in the next subsection). In more detail, given an input $\inp$, the conformal semantics $\confsem{\prog}(\inp)$ yields a \emph{set} $\confout$ of outputs such that $\gtsem{\prog}(\inp) \in \confsem{\prog}(\inp)$. 
In other words, the conformal semantics generates a set of outputs that  includes the ground truth, and we define our learning problem in terms of these conformal semantics. \change{Section ~\ref{sec:impl} provides details about the conformal prediction technique used in our implementation.}

\begin{definition}{\bf (Neurosymbolic learning from examples (NSLE))}\label{def:NSLE} Given a set of input-output examples $\inout$, the goal of neurosymbolic learning is to find a set of programs $\progs$ satisfying:
{
\begin{equation}\label{eq:nsle}
\prog \in \progs \Longleftrightarrow  \forall (\inp, \out) \in \inout. \ \out \in \confsem{\prog}(\inp)
\end{equation}
}
\end{definition}

% \osbert{Given that we can query the user to refine prediction sets, can't we actually solve the original problem for the ground truth semantics? Isn't our algorithm guaranteed to converge to the true program under the ground truth semantics? Conformal semantics just accelerates the process.}
In other words, the goal of NSLE is to identify all programs consistent with the examples under the \emph{conformal} semantics. Each such program is referred to as a \emph{solution} to the NSLE problem. {While our formulation of NSLE is guaranteed to include the desired solution, it may also admit \emph{many} others. As a result, users must manually identify the correct program from a potentially large set of candidates—a process that can be burdensome, especially without expertise in the underlying DSL. To mitigate this, we define the \emph{active learning for NSLE} problem, which aims to efficiently guide the user toward the intended solution.}
% That is, the goal of NSLE is to find all programs that are consistent with the  examples under the \emph{conformal} semantics. We refer to each program in this set as a \emph{solution} to the NSLE problem.  \changed{Please note our NSLE problem formulation in the above definition always admits the desired solution, but it may also admit \emph{many} other solutions. This requires the user to sift through all potential solutions to find the correct one, which can be impractical if the solution space is large or the user lacks expertise in the underlying DSL. Thus, we instead define the \emph{active learning for NLSE} problem to address this shortcoming.}
%Note that there are two advantages of this problem formulation: First, it obviates the need for consulting an oracle to evaluate each program, as would be required by Equation~\ref{eq:NSLE-gt}. Second, it never rules out the desired program from the solution space, as would be done by Equation~\ref{eq:NSLE-eval}.

\subsection{Constrained Conformal Semantics}

At a high level, the goal of {active learning} is to {refine the conformal semantics} until the ground truth program is found. Intuitively, the approach starts with the vanilla conformal semantics, but each round of user interaction makes the conformal semantics more precise, making the output sets smaller while still containing the ground truth.  {We refer to the refined semantics that incorporates user feedback as the \emph{constrained conformal semantics}, and we formally define \emph{user feedback} to precisely characterize this alternative semantics:}

\begin{definition}{\bf (User feedback)}
Let $\inps$ be a finite input space\footnote{As standard in prior work, we assume a finite set of inputs because, in the PBE setting, the synthesized program is used to automate a task on a large but fixed set of inputs.  }  on which programs will be executed. User feedback $\fb$ is a set of \emph{labels} $(f, \inp, \out)$ where $f \in \nns \cup \{ \synthfunc \}$, $\inp \in \inps$, and $\out = \gtsem{f}(\inp)$. 
\end{definition}

% \celeste{To clarify: I added an operation $\oper \in \{=, \neq\}$ to each tuple in the feedback. Inequalities are necessary for backwards evaluation, where we refine the feedback by  eliminating possible values from prediction sets.}
In this definition,  a label can take two forms: If $f = \synthfunc$, the label corresponds to an input-output example for the function to be synthesized. On the other hand, for $f \in \nns$, the label specifies the ground truth semantics of neural component $f$ for a specific input. Given user feedback $\fb$, we use the notation $\fbf_\fb$ to denote the formula representation of $\fb$, i.e.,
\[
\textstyle
\fbf_\fb = \bigwedge_{(f, I, O) \in \fb} f(I) = O
\]
%We say that $\fbf_{\fb}$ \emph{models} $f(I) = O$, denoted $\fbf_\fb \models f(I) = O$, iff $(f, I, O) \in \fb$. 
For notational convenience, we often leave $\fb$ implicit and denote user feedback simply by $\fbf$. We assume that the user always gives feedback consistent with the ground truth.

\begin{figure}
    \footnotesize
    \begin{mathpar}    
    % \inferrule*[Left=Neural]{\bindings(x) = \inp \in \inps \ \ \ \forall o\in \mathcal{\outs}.(f_n, \inp, o) \not\in \assump \ \ \ \ \out = \tilde{\nn}(I) }
    % {\bindings \vdash \nn(x) \constrainedeval \out} \\ 
    \inferrule*[Left=Neural]{\bindings(x) = \inp \ \ \ \ \ \confoutint = \{(\out, \nn(\inp) = \out)  \ | \ \out \in \tilde{\nn}(I) \wedge \textsf{SAT}(\nn(\inp)= O \wedge \fbf) \}}
    {\bindings \vdash \nn(x) \constrainedeval \confoutint} \\ 
    % \inferrule*[Left=Neural2]{\bindings(x) = \inp \ \ \ \ \ \confout = \tilde{\nn}(I) \ \ \ \ \ (\nn, \inp, \out) \in \assump}
    % {\bindings \vdash \nn(x) \constrainedeval \confout \cap \{\out\}} \\
    \inferrule*[Left=Symb]{\bindings \vdash E_1 \constrainedeval \confoutint_1 \quad \ldots  \quad \bindings \vdash E_n \constrainedeval \confoutint_n, \\ \confoutint = \{(\stansem{\symfunc}(\out_1, \ldots, \out_n), \bigwedge_{i} \varphi_i) \ | ((\out_1, \varphi_1), \ldots (\out_n, \varphi_n)) \in \confoutint_1 \times \ldots \times \confoutint_n  \wedge \textsf{SAT}(\fbf \land \bigwedge_{i} \varphi_i)\} }{\bindings \vdash \symfunc(E_1, \ldots, E_n) \constrainedeval \confoutint} \\
    % \inferrule*[Left=IfTrue]{\bindings \vdash E_1 \constrainedeval \{\texttt{true} \} \ \ \ \bindings \vdash E_2 \constrainedeval \confout }{\bindings \vdash E_1 ? E_2 : E_3 \constrainedeval \confout } \ \ \ \ \  \ \ \ \ \  \ \ \ \ \  
    % \inferrule*[Left=IfFalse]{\bindings \vdash E_1 \constrainedeval \{\texttt{false} \} \ \ \ \bindings \vdash E_3 \constrainedeval \confout }{\bindings \vdash E_1 ? E_2 : E_3 \constrainedeval \confout } \\ 
    % \inferrule*[Left=IfEither]{\bindings \vdash E_1 \constrainedeval \{\texttt{true}, \texttt{false} \} \ \ \ \bindings \vdash E_2 \constrainedeval \confout \ \ \  \bindings \vdash E_3 \constrainedeval \confout'}{\bindings \vdash E_1 ? E_2 : E_3 \constrainedeval \confout \cup \confout' } \\ 
   % \inferrule*[Left=If]{\bindings \vdash E_1 \constrainedeval \confoutint_1 \ \ \ \bindings \vdash E_2 \constrainedeval \confoutint_2 \ \ \  \bindings \vdash E_3 \constrainedeval \confoutint_3 \\ \confoutint = \{(\stansem{\out_1 ? \out_2 : \out_3}, \bigwedge_{i} \varphi_i) \ | \ ((\out_1, \varphi_1), (\out_2, \varphi_2), (\out_3, \varphi_3)) \in \confoutint_1 \times \confoutint_2 \times \confoutint_3 \wedge \textsf{SAT}(\bigwedge_{i} \varphi_i)  \}}{\bindings \vdash E_1 ? E_2 : E_3 \constrainedeval \confout  } \\ 
    \inferrule*[Left=Assign]{\shortstack{$\bindings \vdash E \constrainedeval \confoutint$ \\ $\bindings[x \mapsto \confoutint] \vdash S \constrainedeval \confoutint'$}}{\bindings \vdash \texttt{let } x=E;S \constrainedeval \confoutint'}  \ \ \ \ \ \ \ \ \ \ \ \ \ \ \ \
   \quad  \inferrule*[Left=Lookup]{ \ }{\bindings \vdash x \constrainedeval \bindings(x)} \quad \ \ \ \ \ \ \ \ \ \ \ \ \ \
    % \inferrule*[Left=App]{\bindings,x\mapsto I \vdash S \constrainedeval O \ \ \ \forall o \in \outs. (\gtprog, \inp, o) \not\in \assump }{\bindings \vdash (\lambda x.S)I \constrainedeval O} \\ 
    \inferrule*[Left=Prog]{\shortstack{$\bindings(x) = \inp \ \ \ \ \  \bindings \vdash S \constrainedeval \confoutint$ \\ $ \confout = \{\out \ | \ (\out, \varphi) \in \confoutint \}$} }{\bindings \vdash \textsf{def } \synthfunc(x) = S  \constrainedeval \confout } \\
    % \inferrule*[Left=Prog2]{\bindings(x) = \inp \ \ \ \ \  \bindings \vdash S \constrainedeval \confout \ \ \ \ \ (\synthfunc, \inp, \out) \in \assump}{\bindings \vdash \textsf{def } \synthfunc(x) = S  \constrainedeval \confout \cap \{\out\} } 
    \end{mathpar}
    \vspace{-0.35in}
    \caption{Constrained conformal semantics. Conditionals are omitted because they are a special case of  {\sc Symb}.}
    \label{fig:constrainedsem}
    \vspace{-0.1in}
\end{figure}

Next, we define \emph{constrained conformal semantics} to evaluate programs while incorporating user feedback. Given a program $\prog$ and input $\inp$, constrained conformal evaluation uses the user-provided label for evaluating an expression $E$ when available, and performs standard conformal evaluation otherwise. Formally, Figure~\ref{fig:constrainedsem} presents the constrained conformal semantics using judgments of the form $ \bindings \vdash S \constrainedeval \confout $, indicating that $S$ evaluates to the output set $\confout$ under the valuation $\bindings$ and user feedback $\fbf$ (represented as a formula). 
%The top-level rule {\sc Prog} produces a set $\confout$ of output values, while the other rules generate pairs of output values and their constraints. 
%We later illustrate the need for constraints through an example.

%presents the constrained conformal semantics for our meta-language using judgments of the form$\bindings \vdash S \constrainedeval \confout$, meaning that $S$ evaluates to output set $\confout$ under valuation $\bindings$ and user feedback $\fbf$ (represented as a formula). The top-level rule labeled {\sc Prog} produces a set $\confout$ of output values, whereas the remaining rules produce a set $\confoutint$ of pairs of output values and their corresponding constraints. We later illustrate the need for constraints through an example.  
The first rule, {\sc Neural}, shows how to evaluate a neural expression $\nn(x)$ under constrained conformal semantics. It first looks up the value $\inp$ for $x$ under $\bindings$ and performs conformal prediction to obtain the set of possible outputs $\confout$. The result is then filtered to retain only outputs consistent with the user feedback, i.e., those for which $\nn(I) = O \land \fbf$ is logically satisfiable.\footnote{While we formalize this consistency check as the satisfiability of a logical formula, this check does not require an  SMT solver. In particular, these formulas can be viewed as a set of input-output examples for $\nn \in \nns$. \change{Thus, this consistency check may be implemented by querying whether an input $\inp$ maps to output $\out$ in a hash map designated for $\nn$.}} \change{As discussed in Section ~\ref{sec:conformal}, conformal prediction guarantees that $\confout$ will contain the user feedback with bounded, user-controllable probability.}  

The second rule, {\sc Symb}, handles symbolic expressions by evaluating each nested expression $E_i$ under constrained conformal semantics and applying $\symfunc$ to each combination of results, filtering out combinations with unsatisfiable constraints. 
The final rule, {\sc Prog}, shows how to evaluate the entire program under user feedback. This rule evaluates the function body using the other rules to obtain the result $\confoutint$, and then discards the constraints to produce the final output set.

\begin{definition}{\bf (Constrained conformal evaluation)} Let $P$ be a program, $I$ an input, and  let $\fbf$ denote user feedback. We say that $P$ conformally evaluates to $\confout$ on $I$ under user feedback $\fbf$, denoted $\confsem{P}_\fbf(I) = \confout$, if and only if $[ x \mapsto I] \vdash  \prog \constrainedeval \confout$.
\end{definition}

Note that a program $P$ may produce an output set inconsistent with a user-provided label. Therefore, we define a notion of consistency with user feedback:

\begin{definition}{\bf (Consistency with user feedback)}
%Let $\prog$ be a program and $\inps$ be the input space. 
We say that a program $\prog$ is \emph{consistent} with user feedback $\fbf_\fb$, denoted $\prog \models \fbf_\fb$, iff $\forall \ (\synthfunc, I, O) \in \fb. \ \ \out \in \constrainedsemm{\prog}(\inp).$

\end{definition}
In other words, a program $\prog$ is consistent with user feedback if, for all input-output examples $(I, O)$ provided by the user, $\prog$ can return $O$ on $I$ under the conformal semantics. 

\begin{example}
    Let $l = [x_1, x_2, x_3]$ be an input list of hand-written digits for which the  ground truth label is $[2, 4, 9]$ and whose conformal prediction result is $[\{2, 7\}, \{4\}, \{8,9\}]$. Consider the following pairs of user feedback  and programs:
{\small % or \scriptsize for even smaller text
\begin{align*}
    \assump_1 &= (\synthfunc(\inp) = 15) \wedge (\mnistpred(x_1) = 2) 
    &\prog_1 = \lambda x. \texttt{fold} \ (\texttt{map} \ x \ \mnistpred ) \ 0 \ \texttt{sum} \\  
    \assump_2 &= (\synthfunc(\inp) = 15) \wedge (\mnistpred(x_3) = 9)  
    &\prog_2 = \lambda x. \texttt{fold} \ (\texttt{map} \ x \ \mnistpred ) \ 1 \ \texttt{sum} 
\end{align*}
}

Here, for $\prog_1$, we have $\prog_1 \models \assump_1$ and $\prog_1 \models \assump_2$. However, for $\prog_2$, we have $\prog_2 \models \assump_1$ but  $\prog_2 \not\models \assump_2$. This is because $\prog_1$ and $\prog_2$ evaluate to the following values under the constrained conformal semantics:
{ 
    \begin{align*}
   \confsem{\prog_1}_{\assump_1}(l) = \{14, 15\} \quad \confsem{\prog_1}_{\assump_2}(l) = \{15, 20\}   \quad 
    \confsem{\prog_2}_{\assump_1}(l) = \{15, 16\} \quad \confsem{\prog_2}_{\assump_2}(l) = \{16, 21 \}      
    \end{align*}
    }
   % Hence, assuming no other input contains $x_3$ or $x_1$, $\prog_1, \prog_2 \models \assump_1$, $\prog_1 \models \assump_2$, and $\prog_2 \not\models \assump_2$.
\end{example}

\subsection{Active Learning Problem}\label{sec:active_learning_problem}

In this section, we formally define the \emph{active learning} problem addressed in this paper.

\begin{definition}{\bf (Hypothesis space)}\label{def:hs}
Let $\progs$ be the set of all possible programs, and let $\assump$ be the current user feedback. The hypothesis space defined by $\assump$, denoted $\hs{\progs}{\assump}$, is the set $\{ \prog  \in \progs \ | \ P \models \assump \} $.
\end{definition}

\begin{definition}{\bf (Question)}
    A question $\question$ is a tuple $(f, \inp)$ where $f \in \nns \cup \{\synthfunc\}$ and $\inp \in \inps$. We say that an answer $a$ to a question $(f, \inp)$ is valid   iff $a = \gtsem{f}(\inp)$. 
\end{definition}
{Note that there are two types of questions: those involving $\synthfunc$ ask the user to provide a new input-output specification of the function being synthesized. In contrast, those involving neural components $f \in \nns$ ask the user to provide the ground truth label for some perception task. } 
Given a question $q = (f, \inp)$ with a \emph{valid} answer $a$, we use $\mathsf{Feedback}(q)$ to denote the feedback $f(\inp) = a$, and $\mathsf{Possible}(q)$ to represent the set of possible user feedback formulas for $q$. This terminology naturally extends to a question space $\questions$. 
%We now define what it means for a pair of neurosymbolic programs to be \emph{indistinguishable}:

\begin{definition} {\bf (Indistinguishable Programs)}\label{def:indis}
Two neurosymbolic programs $\prog_1$ and $\prog_2$ are \emph{indistinguishable} under feedback $\assump$, denoted $\prog_1 \indist_\assump \prog_2$, if 
for all inputs $\inp$ in the input space and for any set of questions $Q$ not already answered by $\fbf$, we have:
%\begin{enumerate}
%\item $\confsem{\prog_1}_{\assump}(\inp) = \confsem{\prog_2}_{\assump}(\inp) $
%For any set of questions $Q$ and for any $\assump' \in \mathsf{Possible}(Q)$, we have:
%, where $q_i = (f_i, \inp_i)$  and any possible answer $a_i$ to $q_i$, we have:  
%\begin{align*}
 %\confsem{\prog_1}_{\assump  \land (\underset{i}{\bigwedge} f_i(\inp_i) = a_i)}(\inp) = \confsem{\prog_2}_{\assump  \land (\underset{i}{\bigwedge} f_i(\inp_i) = a_i)}(\inp)
% \confsem{\prog_1}_{\assump  \land \assump'}(\inp) = \confsem{\prog_2}_{\assump  \land \assump'}(\inp)
%\end{align*}
{ 
\begin{equation}\label{eq:dis}
 %\confsem{\prog_1}_{\assump  \land (\underset{i}{\bigwedge} f_i(\inp_i) = a_i)}(\inp) = \confsem{\prog_2}_{\assump  \land (\underset{i}{\bigwedge} f_i(\inp_i) = a_i)}(\inp)
 \forall \assump' \in \mathsf{Possible}(Q).\ \ \confsem{\prog_1}_{\assump  \land \assump'}(\inp) = \confsem{\prog_2}_{\assump  \land \assump'}(\inp)
\end{equation}
}
%\end{enumerate}

\end{definition}
% \osbert{We need Proposition 3.7 here to ensure that if we have indistinguishability for a single $(q,a)$ pair, then we have indistinguishability for all assumptions $\assump'$:
% \begin{align*}
% \eval{\assump \cup \assump'}{\inp}{\prog_1}{\confsem{\prog_2}(\inp)}.
% \end{align*}
% }
In other words, programs $P_1, P_2$ are indistinguishable under $\assump$ if we cannot distinguish their behavior by asking the user more questions. Note that Definition~\ref{def:indis} \emph{implies} $\confsem{\prog_1}_{\assump}(\inp) = \confsem{\prog_2}_{\assump}(\inp)$;  however, $\confsem{\prog_1}_{\assump}(\inp) = \confsem{\prog_2}_{\assump}(\inp)$ does \emph{not} imply that the programs are indistinguishable. In particular, two programs that evaluate to the same output set under $\assump$ \emph{could} evaluate to different sets as we strengthen $\assump$ by getting additional feedback from the user.

%Condition (1) in the above definition corresponds to the standard notion of \emph{observational equivalence} (i.e., the two programs have the same input-output behavior). While condition (1) alone is sufficient for indistinguishability in the standard setting, it is not in the conformal semantics setting. Since conformal semantics \emph{overapproximate} the ground truth semantics, condition (1) only guarantees that there exists a possibility that $\prog_1$ and $\prog_2$ are observational equivalent under conformal semantics. To be able to conclude that $\prog_1$ and $\prog_2$ are indeed observational equivalent, we require $\prog_1$ and $\prog_2$ to have the same input-output behavior with respect to \emph{all} possible refinements of the conformal semantics, which is exactly stated in condition (2).     
% we need an additional condition (2) in the neurosymbolic setting.
\begin{example}
    Consider two programs $\prog_1 = \lambda x. \nn^1(x)$, $ \prog_2 = 
\lambda x. \nn^2(x) + 1$, input space $\{ \inp\}$, and user feedback $\varnothing$. Suppose further that $\tilde{\nn^1}(\inp) = \{ 1, 2 \}$ and  $\tilde{\nn^2}(\inp) = \{ 0, 1\}$. Here, we have $\confsem{\prog_1}_{\assump}(\inp) = \confsem{\prog_2}_{\assump}(\inp) = \{ 1, 2 \}$; but these programs are  \emph{distinguishable}. 
    In particular, consider the question $(\nn^1, \inp)$ and possible answer $2$. Then, we have $\confsem{P_1}_{\nn^1(x) = 2}(\inp) = \{2 \} $, but  $\confsem{P_2}_{\nn^1(x) = 2}(\inp) = \{1, 2 \} $, so they violate Equation~\ref{eq:dis}. Our definition of indistinguishability requires the programs to remain observationally equivalent under any answers to future questions. 
\end{example}

% \begin{proposition}\label{prop:indist} 
% Consider user feedback $\assump$ and $\assump'$ such that  $\assump' \Rightarrow \assump$. Then, for any  programs $\prog_1$, $\prog_2$ where $\prog_1 \indist_\assump \prog_2$, we have $\prog_1 \indist_{\assump'} \prog_2$. 
% \end{proposition}

% This proposition states that if two programs are indistinguishable under $\assump$, then we cannot make them distinguishable by adding more facts to $\assump$. Thus, the goal of interactive synthesis is to keep asking the user questions until all programs in the hypothesis space become indistinguishable, at which point querying the user with more questions is unnecessary.
%If two programs $P_1, P_2$ are distinguishable under feedback $\fbf$, we use the notation $P_1 \not \indist P_2 $. 
%The following lemma states a sufficient condition for concluding that two programs are distinguishable. 

% \begin{lemma}\footnote{Proofs of all lemmas and theorems are provided in the.}\label{lemma:sufficient}
% Let $\prog_1, \prog_2$ be a pair of programs such that $\constrainedsem{\prog_1}(I) \neq \constrainedsem{\prog_2}(I)$ for some $\inp \in \inps$, where $\inps$ denotes the input space. Then, $\prog_1 \not \indist \prog_2$. 
% \end{lemma}

%Next, we state the following lemma that says that two indistinguishable program must be observationally equivalent under the ground truth semantics: 

\begin{lemma}\label{lemma:indis}\label{lemma:sufficient}
Let $\prog_1, \prog_2$ be a pair of programs such that $\prog_1 \indist_\assump \prog_2$. Then, assuming $\assump$ represents accurate user feedback, we have $\forall I \in \inps. \ \gtsem{\prog_1}(I) = \gtsem{\prog_2}(I)$ where $\inps$ denotes the input space.\footnote{Proofs of all lemmas and theorems are provided in 
Appendix ~\ref{sec:proofs}.
% the Appendix of the extended version of the paper \cite{extended}.
}
\end{lemma}

%In other words, if two programs are indistinguishable according according to Definition~\ref{def:indis}, they must be observationally equivalent on the input space under the ground truth semantics. Thus,  the goal of active learning is to query the user until all programs in the hypothesis space become indistinguishable. 

This lemma is important because it states that active learning can terminate when all programs in the hypothesis space are indistinguishable from one another. 
%Thus, we formulate the active learning problem as follows: 

\begin{definition}{\bf (Active learning for NSLE).} Let $\questions$ be the space of all possible queries and let $\gtprog$ be a ground truth program.
%such that $\forall \inp_i \in \inps. \gtsem{\gtprog}(\inp_i) = \out_i$.  
The goal of the active learning problem is to ask the user a series of questions $q_0, \ldots, q_n$ in $\questions$ such that (1) for $\fbf = \bigwedge_{i=0}^n \mathsf{Feedback}(q_i)$, we have 
$\forall P \in (\hs{\progs}{\assump}). \ P \indist_\assump \gtprog$; (2)  the number of questions asked to the user is minimized.
\end{definition}

The first criterion states that when active learning terminates, all remaining programs  are observationally equivalent on inputs $\inps$ under the ground truth semantics. The second condition focuses on optimality, aiming to minimize the number of questions the user must answer. However, achieving a globally optimal solution is NP-hard~\cite{samplesy} without knowing the answer in advance. Like previous work on active learning~\cite{samplesy}, we focus on \emph{worst-case} rounds of interaction. Our revised problem aims to construct a sequence of questions where each is \emph{locally} optimal, formalized in terms of pruning power:

\begin{definition} {\bf (Pruning Power)}.\label{def:pp}
Let $\progs$ be the  hypothesis space, and let   $\question = (f, \inp)$ be a question  with possible answers $a_1, \dots, a_n$. Then, the \emph{pruning power} of $\question$  modulo hypothesis space $\progs$ is: 
\begin{align*}
\textstyle
    \pp{q, \progs} = \frac{|\progs| - \underset{i}{\text{max}} | (\hs{\progs}{(f(\inp) = a_i)})|}{|\progs|}.
\end{align*}
\end{definition}

In other words, the pruning power of a question $\question$ is the fraction of programs pruned from the hypothesis space by the \emph{worst} answer to $\question$.
%With this definition in place, we now state our revised problem statement:

\begin{definition}{\bf (Revised problem statement)}\label{def:revised} Let $\questions$ denote all  questions and let $\gtprog$ be a ground truth program.
%such that $\forall \inp_i \in \inps. \gtsem{\gtprog}(\inp_i) = \out_i$.  
Active learning  aims to ask the user a series of questions $q_0, \ldots, q_n$ in $\questions$ such that: 

\begin{enumerate}
    \item For $\fbf = \bigwedge_{i=0}^n \mathsf{Feedback}(q_i)$, we have: $\forall P \in (\hs{\progs}{\assump}). \ P \indist_\assump \gtprog$
\item For $\fbf^i = \bigwedge_{j=0}^{i-1} \mathsf{Feedback}(q_j)$ and $\questions^i = \questions \backslash \cup_{j=0}^{i-1} \{q_j\}$, we have
$
q_i =    
        \underset{\question \in \questions^i
        }{\text{argmax }}\pp{q, \hs{\progs}{\fbf^i}}
$
\end{enumerate}
\end{definition}

As before, the first criterion states that all remaining programs in the hypothesis space are indistinguishable. The second condition states a \emph{local optimality} requirement, namely that the next question asked to the user should have maximal pruning power among all remaining questions. 

\paragraph{\bf \emph{Remark}.}  { A reasonable alternative to considering the worst-case answer is to reason about the \emph{expected} answer by leveraging prediction probabilities of neural components. While we also considered this alternative, reasoning about expected rounds of interaction does not work well empirically, as we evaluate in Section~\ref{sec:rounds}. Intuitively,  neural networks can sometimes be confidently wrong about their predictions; hence, reasoning about worst-case behavior turns out to be a better objective in practice.}

\section{Interactive Neurosymbolic Synthesis Algorithm}\label{sec:synth}

% In this section, we present our interactive synthesis algorithm that aims to solve the problem presented in Definition~\ref{def:revised}. 

\subsection{Top-Level Algorithm}\label{sec:top-level}

\begin{figure}[!t]
\vspace*{-0.5cm}
\begin{minipage}[t]{0.55\textwidth}
% \begin{subfigure}
    \centering
    \scriptsize
    \begin{algorithm}[H]
    \small
    \begin{algorithmic}[1]
    \Procedure{ActiveLearning}{$\progs, \questions, \inps, \inout$}
    \Statex\Input{Hypothesis space $\progs$, 
 question space $\questions$,  input space $\questions$  and  IO examples $\inout$.}
    \Statex\Output{A program $\prog \in \progs$}
    \State $\fbf \assign \{\synthfunc(\inp) = \out \ | \ (\inp,\out)\in \inout \}$
    \While{true}
    \State $\progs \assign \textsc{RefineHS}(\progs, \fbf, \inout)$
    \State $\prog' \assign \textsf{Sample}(\progs)$
    \If{ $\neg \textsc{Distinguish}(\prog', \progs \setminus \prog', \fbf, \inps, \questions)$}
    \State {\bf break;}
    \EndIf
    \State $\question \assign \text{\sc{SelectQuestion}}(\progs,\questions)$
    \State $\out \assign \text{\sf{QueryUser}}(\question)$
    \State $\fbf \assign \fbf \wedge f(\inp) = \out; \ \ \questions \assign \questions \setminus (f, \inp)$
    \If{$f = \synthfunc$}
    $Z \gets Z \cup \{ (I, O) \} $
    \EndIf
    \EndWhile
   \State \Return $\prog \in \progs$
    \EndProcedure
    \end{algorithmic}
    \end{algorithm}
    \vspace{-0.5in}
    \caption{Active learning algorithm. }
    \label{fig:toplevel}
% \end{subfigure}
\end{minipage}%
\begin{minipage}[t]{0.45\textwidth}
% \begin{subfigure}
    \centering
    \small
    \begin{algorithm}[H]
    \small
    \begin{algorithmic}[1]
    \Procedure{RefineHS}{$\progs, \fbf, \inout$}
    \Statex\Input{$\progs$ is a hypothesis space, $\fbf$ is a user feedback formula, and $\inout$ is the IO examples.}
    \Statex\Output{A new hypothesis space $\progs'$}
    \State $\progs' \assign \progs$
    \ForAll{$\prog \in \progs$}
    \ForAll{$(\inp, \out) \in \inout$}
    \State $\confout \assign \underline{\textsc{CCE}}(\prog, \fbf, \inp)$
    \If{$\out \not\in \confout$}
    \State $\progs' \assign \progs' \setminus \prog$ 
    \State {\bf break;} 
    \EndIf
    \EndFor
    \EndFor 
    \State \Return $\progs'$
    \EndProcedure
    \end{algorithmic}
    \end{algorithm}
    \vspace{-0.1in}
    \caption{Hypothesis space refinement algorithm. }
    \label{fig:refinehs}
% \end{subfigure}
\end{minipage}%
\vspace{-.15in}
\end{figure}

Our top-level active learning algorithm is presented in Figure~\ref{fig:toplevel}. This procedure takes four inputs, namely (1) $\progs$, which is the hypothesis space and is assumed to contain the ground truth program $\prog^*$, (2) the space $\questions$ of all possible questions, (3) the input space $\inps$, and (4) a set of initial input-output examples provided by the user. The return value of {\sc ActiveLearning} is a program that is  observationally equivalent to $\prog^*$ on the input space under  the ground truth semantics.

The {\sc ActiveLearning} procedure gradually grows the user feedback $\fbf$ and refines the hypothesis space $\progs$ until all pairs of programs in $\progs$ are indistinguishable according to Definition~\ref{def:indis}.  Initially, the user feedback $\fbf$ contains the input-output examples provided by the user (line 2). In each iteration, the algorithm first computes the new hypothesis space by calling {\sc RefineHS} (shown in Figure~\ref{fig:refinehs}), which computes $\hs{\progs}{\fbf}$ from Definition~\ref{def:hs} with some optimizations discussed in the next subsection. Next, at line 5, it samples a program $\prog'$ from the hypothesis space and then calls the {\sc Distinguish} procedure (presented in Figure~\ref{fig:dist}) to check whether there exists any program in $\progs$ that is distinguishable from $\prog'$ according to Definition~\ref{def:hs}.  If no distinguishable programs exist, the procedure terminates. Otherwise, it selects a question using  {\sc SelectQuestion}  (explained in Section~\ref{sec:select}) to maximize pruning power, and the user's answer is added to the feedback.

%If all programs in $\progs$ are indistinguishable from $\prog'$, the active learning procedure terminates. 
%Otherwise, it needs to keep querying the user to refine the hypothesis space. Thus,  line 8 in Figure~\ref{fig:toplevel} invokes the {\sc SelectQuestion} procedure (explained in Section~\ref{sec:select}) to find a question with maximal pruning power. The user answers this question, and their answer is added to the user feedback. 

The top-level algorithm relies on three auxiliary procedures: {\sc RefineHS}, {\sc Distinguish}, and {\sc SelectQuestion}. Since the first two are straightforward, we describe them here, deferring the discussion of {\sc SelectQuestion} to Section~\ref{sec:select}. As shown in Figure~\ref{fig:refinehs},  {\sc RefineHS}   iterates over all programs in the hypothesis space, checking if each $\prog \in \progs$ satisfies the feedback $\Phi$ using the \textsc{CCE} procedure for constrained conformal evaluation. If any input-output example $(\inp, \out) \in \inout$ is not included in the \textsc{CCE} result, $\prog$ is pruned from the hypothesis space. The {\sc Distinguish} procedure is presented in Figure~\ref{fig:dist} and checks if there is a program $\prog \in \progs$ that is distinguishable from $\prog'$. It first examines whether any input $\inp \in \inps$ causes \textsc{CCE} to produce different results for $\prog$ and $\prog'$. If so,  the algorithm returns true. Otherwise, it considers all possible question-answer pairs $(q, a)$ and checks if any program $\prog \in \progs$ could be distinguished from $\prog'$ based on the user's response.

%and checks whether there exists a program $\prog \in \progs$ that is distinguishable from $\prog'$. To do so, it first checks whether there exists an input $\inp \in \inps$ and $\prog \in \progs$ for which \textsc{CCE} yields different results for $\prog$ and $\prog'$. If so, then $\prog$ and $\prog'$ are definitely distinguishable according to Lemma~\ref{lemma:sufficient}, so the algorithm returns true. Otherwise, it considers all possible question and answer pairs $(q, a)$ and checks whether any program $\prog \in \progs$ could be distinguished from $\prog'$ if the user were to provide $q= a$ as feedback. 

%\paragraph{\bf Remark.} In the worst case, $\textsc{Distinguish}$ requires $\mathcal{O}(|\progs| \times |\inps| \times |\questions| \times |\textsf{Answers}(\questions)|)$ calls to the \textsc{CCE} procedure.  
%However, in practice, $\textsc{Distinguish}$ is very fast, as the procedure only runs until it finds \emph{a single pair} of indistinguishable programs.  When the hypothesis space is initially very large, there are many programs that are distinguishable, so it’s very easy to find a single pair of distinguishable programs. After each round of user interaction, the hypothesis space shrinks substantially, so by the time it becomes hard to find a pair of distinguishable programs, there are very few programs left in the hypothesis space. Furthermore, after each round of user interaction, \textsc{CCE} becomes easier due to the constraints imposed by the user’s answers thus far.

\begin{figure}
    \centering
    % \small
    \begin{algorithm}[H]
    \small
    \begin{algorithmic}[1]
    \Procedure{Distinguish}{$\prog', \progs, \fbf, \inps, \questions$}
    \Statex\Input{Program $\prog'$, hypothesis space $\progs$, user feedback $\fbf$,  input space $\inps$, question space $\questions$.}
    \Statex\Output{A boolean indicating whether there is any program in $\progs$ distinguishable from $\prog'$.}

    %\ForAll{$\inp, (\prog_1, \prog_2) \in \inps \times \progs_{pairs}$}
    %\If{$\stansem{\prog_1}(\inp) \neq \stansem{\prog_2}(\inp) $}  \Return true
    %\EndIf 
    %\EndFor

    \ForAll{$\inp, \prog \in \inps \times \progs $}
    \If{$\underline{\textsc{CCE}}(\prog, \inp, \fbf) \neq \underline{\textsc{CCE}}(\prog', \inp, \fbf) $}
     \Return true
    \EndIf 
    \EndFor
    \ForAll{$(\question, \answer) \in \questions \times \textsf{Answers}(\question)$}
 %   \State $\fbf' \assign \fbf \wedge (\question = \answer)$
    \If{$\textsc{Distinguish}(\prog', \progs, \fbf \wedge (\question = \answer), \inps, \questions \setminus q)$}  \Return true
    \EndIf
    \EndFor
    \State \Return false
    \EndProcedure
    \end{algorithmic}
    \end{algorithm}
    \vspace{-0.5in}
    \caption{Procedure for checking distinguishability.}
    \label{fig:dist}
    \vspace{-.2cm}
\end{figure}

\begin{lemma}\label{lemma:dist1}
Let $\progs$ be a hypothesis space and let $\prog'$ be a program randomly sampled from $\progs$. For any user feedback $\fbf$, input space $\inps$, and question space $\questions$, we have
$
\forall \prog_1, \prog_2 \in \progs. \ \prog_1 \indist_\fbf \prog_2
$ if and only if  {\sc Distinguish}($\prog', \progs \backslash \prog', \fbf, \inps, \questions$) returns false.
\end{lemma}

\begin{theorem}
Let $\prog^*$ be the ground truth program, and let $\prog$ be the program returned by the {\sc ActiveLearning} procedure. Then, assuming $\prog^* \in \progs$, we have 
$
\forall \inp \in \inps. \ \gtsem{\prog}(\inp) = \gtsem{\prog^*}(\inp)
$.
\end{theorem}
\paragraph{\bf \emph{Remark}.} In the worst case, \textsc{Distinguish} requires $\mathcal{O}(|\progs| \times |\inps| \times |\questions| \times |\textsf{Answers}(\questions)|)$ \textsc{CCE} calls. However, in practice, \textsc{Distinguish} is efficient because it stops as soon as it finds a single program that is distinguishable from $\prog'$. Early in the active learning loop, when the hypothesis space is large, it is easy to find such a program, so this subroutine runs very quickly. After each round of user interaction, the hypothesis space shrinks, and although finding a distinguishable program becomes harder, fewer programs remain -- hence, the running time continues to be fast. 

\subsection{Practical Constrained Conformal Evaluation}\label{sec:confeval}

%\todo{Emphasize that forward alone doesn't make sense}

Several aspects of the {\sc ActiveLearning} procedure require evaluating \emph{many} programs on \emph{many} inputs using constrained conformal semantics.
In particular, in Figures~\ref{fig:refinehs} and~\ref{fig:dist}, constrained conformal evaluation is performed using the call to procedure ${\textrm{CCE}}$, whose invocations are \underline{underlined}. 
Unfortunately, when the prediction sets are large, performing constrained conformal evaluation can be very costly.
In this section, we present a CCE algorithm that leverages user-provided input-output examples to make constrained conformal evaluation more efficient in practice.

%algorithm for performing constrained conformal evaluation. This algorithm takes three inputs, namely a program $\prog$ to be evaluated (represented in terms of its abstract syntax tree (AST)), the user feedback $\fbf$, and an input $\inp$. The result of CCE is a set $\confout$ such that $\confout = \confsem{P}_\fbf(I)$. 

To understand the key insight behind the \textsc{CCE} procedure, consider the conformal evaluation of a symbolic function $\symfunc$: As shown in the  {\sc Symb} rule of Figure~\ref{fig:constrainedsem}, an expression $\symfunc(E_1, \ldots, E_n)$ is evaluated by first evaluating each $E_i$ and then taking the cross product of the resulting sets $\confoutint_i$. If each $\confoutint_i$ has cardinality $m$, the expression can produce up to $m^n$ outputs, causing an exponential blowup.
However, in practice, we find that many of these outputs are \emph{infeasible} given the user feedback provided so far. As an example, consider the following function:
\begin{align*}
\small
    \lambda  \ x_1, x_2, x_3 .  (\mnistpred(x_1) + \mnistpred(x_2)) \times \mnistpred(x_3) 
\end{align*}
Suppose the prediction set for $x_3$ includes $\{2, 4, 7\}$, but the expected output for input $\inp$ is $9$. Since multiplying by an even number produces an even result, the labels 2 and 4 for $x_3$ must be spurious. Thus, even without explicit user feedback, we can refine $x_3$'s prediction set using this information. This motivates our \textsc{CCE} procedure based on bidirectional abstract interpretation.

%and suppose that the prediction set for $x_3$ includes $\{2,4,7\}$ for some input $\inp$ whose corresponding output is specified as $9$ (i.e., $(\inp, 9) \in \inout$). Since anything multiplied by an even number produces an even number, we know that the possible labels 2 and 4 for $x_3$ must be spurious. Thus, we can use this information to refine the prediction set for $x_3$ even though the user has not provided \emph{any} labels for $\mnistpred$. Motivated by this observation, our method uses bidirectional abstract interpretation  to reduce the size of the prediction sets. 

\begin{figure}[!t]
\vspace{-.2in}
\begin{minipage}[t]{0.55\textwidth}   
% \begin{subfigure}
    \centering
     \small
    \begin{algorithm}[H]
    \small
    \begin{algorithmic}[1]
    \Procedure{CCE}{$\prog, \fbf, \inp$}
    \Statex\Input{$\prog$ is a program, $\fbf$ is  feedback, and $\inp$ is an input.}
    \Statex\Output{An output set $\confout$}
   % \State $\tree \assign \textsf{AST}(\prog)$
    \State $\spec \assign \textsc{ForwardAI}(\prog, \fbf, \inp)$
    \State $\spec \assign \textsc{BackwardAI}(\prog, \theta)$
    \State $\confoutint \assign \textsc{EvalConsistent}(\prog, \spec, \fbf)$
    \State \Return $\{\out \ | \ (\out, \varphi) \in \confoutint \} $ 
    \EndProcedure
    \end{algorithmic}
    \end{algorithm}
    \vspace*{-1cm}
    \caption{Constrained conformal evaluation algorithm.}
    \label{fig:cce}
% \end{subfigure}
\vspace*{-0.3cm}
% \begin{subfigure}
    \centering
    % \small
    \begin{algorithm}[H]
    \small
    \begin{algorithmic}[1]
    \Procedure{EvalConsistent}{$\prog, \theta, \inp$}
    \Statex\Input{$\prog$ is an AST, $\theta$ is a constraint \\ specification, and $\inp$ is an input.}
    \Statex\Output{Output set $\confout$}

    \State $n \assign \textsf{Root}(\prog);  \quad f \assign \textsf{Label}(n); \quad \confoutint \assign \varnothing$
    \If{$\textsf{Leaf}(n)$} \Return $\confsem{\prog}_{\spec[n]}(\inp)$
    \EndIf
    \ForAll{$n_i \in \textsf{Children}(n)$}
    \State $\confoutint_i \assign \textsc{EvalConsistent}(n_i, \spec, \inp)$
    \EndFor
   % \State $\confout \assign \varnothing$
    \ForAll{$((\out_1, \varphi_i), \ldots, (\out_k, \varphi_k)) \in \confoutint_1 \times \ldots \times \confoutint_k $}
    %\State $\out \assign \stansem{f}(\out_1, \ldots, \out_k)$
    \If{$\textsf{SAT}(\underset{i}{\bigwedge} \varphi_i) \wedge \stansem{f}(\out_1, \ldots, \out_k) \in \gamma(\spec[n])$}
    %\textsf{SAT}(f(\inp) = \out \wedge \spec[n])$}
    \State $\confoutint \assign \confoutint \cup \{(\stansem{f}(\out_1, \ldots, \out_k), \underset{i}{\bigwedge} \varphi_i)\}$
    \EndIf 
    \EndFor 
    \State \Return $\confoutint$
    \EndProcedure
    \end{algorithmic}
    \end{algorithm}
    \vspace*{-1cm}
    \caption{Conformal evaluation using abstract values.}
    \label{fig:evalconsistent}
% \end{subfigure}
\end{minipage}%
\begin{minipage}[t]{0.45\textwidth}
% \begin{subfigure}
    \centering
    % \small
    \begin{algorithm}[H]
    \small
    \begin{algorithmic}[1]
    \Procedure{ForwardAI}{$\prog, \fbf, \inp$}
    \Statex\Input{$\prog$ is an AST, $\fbf$ is  feedback,  $\inp$ is input.}
    \Statex\Output{Mapping $\spec$ from AST nodes to abstract values}

    %\State $\spec \assign [n \mapsto \text{true} \ | \ n \in \textsf{Nodes}(\prog) ]$
    \ForAll{$n \in \textsf{Nodes}(\prog)$}
    % \If{$n = \textsf{Root}(\tree) \wedge \textsf{Label}(\fbf, \synthfunc, \inp) = \textsf{Some}(\out)$}
    % \State $\spec[n] \assign (n = \out)$
    \If{$\fbf \models \textsf{Label}(n)(\inp) = \out$}
    \State $\spec[n] \assign \alpha(\out)$
    \Else \  $\spec[n] \assign \abssem{\textsf{SubProg}(\prog, n)}(\inp) $
    \EndIf 
    \EndFor
    \State \Return $\spec$ 
    \EndProcedure
    \end{algorithmic}
    \end{algorithm}
    \vspace*{-1cm}
    \caption{Forward abstract interpretation.}
    \label{fig:initspec}
% \end{subfigure}
\vspace*{-0.3cm}
% \begin{subfigure}
    \centering
    % \small
    \begin{algorithm}[H]
    \small
    \begin{algorithmic}[1]
    \Procedure{BackwardAI}{$\prog, \theta$}
    \Statex\Input{$\prog$ is an AST, $\theta$ is a mapping from AST nodes to abstract values.}
    \Statex\Output{A strengthened version of $\spec$}

    \State $n \assign \textsf{Root}(\prog)$
    \State $f \assign \textsf{Label}(n)$
    \ForAll {$n_0, \ldots, n_m \in \textsf{Children}(n) $}
    % \State $\chi \assign \backsem{f}(\spec[n], \spec[n_0], \ldots, \spec[n_{i-1}], \spec[n_{i+1}], \ldots, \spec[n_m])$
    \State $\chi \assign \backsem{f}{i}(\spec[n], \spec[n_0], \ldots, \spec[n_m])$
    \State $\spec[n_i] \assign \spec[n_i] \sqcap  \chi$ 
    \State $\spec \assign \textsc{BackwardAI}(n_i, \theta)$
    \EndFor
    \State \Return $\spec$

    \EndProcedure
    \end{algorithmic}
    \end{algorithm}
    \vspace*{-1cm}
    \caption{Backward abstract interpretation.}
    \label{fig:strengthen}
% \end{subfigure}
\end{minipage}%
\vspace*{-0.4cm}
\end{figure}

As shown in Figure~\ref{fig:cce}, \textsc{CCE} takes three inputs, namely a program $\prog$ to be evaluated (represented in terms of its AST), the user feedback $\fbf$, and an input $\inp$. The result of \textsc{CCE} is a set $\confout$ such that $\confout = \confsem{P}_\fbf(I)$.  The algorithm consists of three high-level steps. The first two steps compute a mapping $\spec$ from each node $n$ of the program's AST to an abstract value. The third step, called {\sc EvalConsistent} (shown in Figure~\ref{fig:evalconsistent}), uses the computed mapping $\spec$ to refine the prediction sets, with the goal of preventing an exponential blowup in their size as they are propagated through the AST.  In particular, after evaluating all the children of a node $n$ (lines 5-6 in Figure~\ref{fig:evalconsistent}), {\sc EvalConsistent}  filters evaluation  results that are inconsistent with $\spec[n]$ (lines 8--9). Because the filtering operation helps keep the evaluation results of subexpressions in check, {\sc EvalConsistent} allows for more practical propagation of these output sets.   

The \textsc{CCE} procedure computes mapping $\spec$ using bidirectional abstract interpretation. The forward direction, summarized in procedure {\sc ForwardAI} (Figure~\ref{fig:initspec}), initializes $\spec$ using the user feedback $\fbf$ and forward abstract interpretation on input $\inp$. Note that, even though we wish to evaluate the program on a \emph{concrete} input $\inp$, abstract interpretation is useful for summarizing prediction sets into a \emph{single} abstract value, thereby avoiding potential exponential blowup. In the {\sc ForwardAI} procedure, we assume access to an abstract semantics $\abssem{\cdot}$ (see line 6) that can be used to compute the abstract output of a program $P$ on a concrete input $I$. 
%Since the abstract domain and transformers are necessarily domain-specific, we discuss the abstract semantics for the two neurosymbolic languages used in our evaluation in Section ~\ref{sec:domains}.

Given mapping $\spec$ initialized via forward abstract interpretation, the \textsc{CCE} procedure invokes {\sc BackwardAI} to \emph{refine} the abstract values  for each node. Starting with the abstract value for the root, {\sc BackwardAI} (see Figure~\ref{fig:strengthen}) uses abstract transformers in the inverse direction to compute the abstract values for the children of a construct $f$. As shown on line 5, we use the notation $\backsem{\cdot}{i}$ to denote an \emph{inverse abstract transformer} that computes the abstract value for the $i$'th operand of $f$, given its output abstract value and the previous abstract values of its  children. 
%Please note that computing the inverse abstract semantics would not be feasible without invoking {\sc ForwardAI} first, as we would not know anything about the abstract values of the other children.  

\begin{example}

Consider the following program in the hypothesis space:
{\small 
\[
\textstyle
\vspace{-0.05in}
\prog = \lambda x_1, x_2, x_3. \tightunderbrace{\mnistpred(x_3)}{a} + \overbrace{(\tightunderbrace{\mnistpred(x_1)}{b} 
\times \tightunderbrace{\mnistpred(x_2)}{c})}^{d} 
\]
}
Suppose  the user provides a new input-output example $\fbf \equiv \synthfunc(\inp_1, \inp_2, \inp_3) = 16$ where each $\inp_j$ is an image of a digit, and suppose we have the following prediction sets:
{\small 
\[
\confsem{\mnistpred}(\inp_1) = \{1, 2 \} \quad 
\confsem{\mnistpred}(\inp_2) = \{2, 4 \} \quad 
\confsem{\mnistpred}(\inp_3) = \{6, 7, 8 \} \quad 
\]
}
%We need to perform \textsc{CCE} to determine if program $\prog$ is consistent with the new input-output example when refining the hypothesis space. Using the conformal evaluation semantics from Figure~\ref{fig:constrainedsem}, we can determine that $\confsem{\prog}_\fbf(\inp_1, \inp_2, \inp_3)= \{ 8, 9, 10, 11, 12, 14, 15, 16 \}$. Since $16 \in \confsem{\prog}_\fbf(x_1, x_2, x_3)$, this program should be in the refined hypothesis space. However, determining that $\prog \models \fbf$ requires computing a fairly large output set. 

%We now illustrate how the \textsc{CCE} procedure checks whether a program is consistent with user-provided examples.
Using the {\sc ForwardAI} procedure, we first compute the following mapping $\spec$ in the \emph{interval abstract domain}~\cite{cousot77}: $a \mapsto [6,8]$, $b \mapsto [1,2]$, $c \mapsto [2,4]$, $d \mapsto [2,8]$, and $\prog \mapsto [16,16]$. Here, $a$ maps to $[6,8]$ because the prediction set $\{6,7,8\}$ for $\inp_3$ is abstracted as $[6,8]$. The abstract value for $\prog$ is $[16,16]$ due to the user-provided output of 16, and $d$ is computed using the forward abstract transformer for multiplication.
Next, {\sc BackwardAI} refines the abstract values. Starting with $\prog$ and using the inverse semantics for $+$, we infer $\spec[a] = [8,14] \sqcap [6,8] = [8,8]$, and then $\spec[d] = [8,8] \sqcap [2,8] = [8,8]$. Using the inverse transformer for $\times$, we compute $\spec[b] = [2,4] \sqcap [1,2] = [2,2]$ and $\spec[c] = [4,4] \sqcap [2,4] = [4,4]$. Finally, {\sc EvalConsistent} evaluates each AST node. For $a$, the only consistent value in $\{6,7,8\}$ is $8$, and similarly, $2$ and $4$ are obtained for $b$ and $c$. The result of {\sc EvalConsistent} is the singleton set ${16}$, allowing us to conclude $\prog \models \fbf$ without exponential blowup during conformal evaluation.

\end{example}

\vspace{-0.1in}
\paragraph{\textbf{Remark 1.}} Our \textsc{CCE} strategy fundamentally relies on \emph{both}  forward \emph{and} backward  reasoning. Without  backward reasoning, {\sc EvalConsistent} would not be able to filter out any values. Conversely, backward reasoning is not possible without forward reasoning, as inverse transformers require the abstract values of the arguments in addition to the output.

\vspace{-0.05in}
\paragraph{\textbf{Remark 2.}} 
Given an input $\inp$ to which $n$ neural constructs may be applied, the worst case complexity of $\textsc{CCE}$ is $\mathcal{O}(n^m)$, where  $m$ is the maximum prediction set size. Thus, while our $\textsc{CCE}$ algorithm does not change the \emph{theoretical} worst-case complexity of constrained conformal evaluation, it makes a significant difference in practice: {our experiments in Section~\ref{sec:eval} show that the na\"{i}ve conformal evaluation strategy grows exponentially with respect to prediction set size, whereas our proposed method does not.}

\subsection{Question Selection}\label{sec:select}

In this section, we describe the {\sc SelectQuestion} procedure that identifies a question with maximal pruning power. A naive approach would consider all possible questions, compute their pruning powers, and return the best one. Unfortunately, computing the pruning power of a question can be very expensive because it requires  refining the hypothesis space and performing constrained conformal evaluation for \emph{every} remaining program. Thus, even with our optimized \textsc{CCE} procedure, this method can be very expensive.

The key idea of our question selection algorithm is to \emph{overapproximate} the pruning power of a question through a new evaluation strategy called \emph{bounded conformal evaluation (\textsc{BCE})}. The difference between \textsc{BCE} and \textsc{CCE} is similar to the difference between \emph{beam search} and \emph{exhaustive search}. Like beam search limits the beam length to $k$, \textsc{BCE} restricts output set cardinalities to a fixed $k$, which acts as a hyper-parameter for our algorithm.  \textsc{BCE} approximates the constrained conformal evaluation rules   from Figure~\ref{fig:constrainedsem}
 by modifying the {\sc Neural} and {\sc Symb} rules to choose $k$ elements from $\Delta$ whenever $|\Delta| > k$. 
Since \textsc{BCE} is a simple modification of these rules, we omit its formal presentation, using $\confsem{P}_\fbf^k(I)$ to denote the result of \textsc{BCE} with hyper-parameter $k$. Unlike \textsc{CCE}, \textsc{BCE} avoids exponential blowup, as all output sets are bounded by $k$.

\begin{figure}[!t]
    \centering
    \small
    \begin{algorithm}[H]
    \begin{algorithmic}[1]
    \Procedure{SelectQuestion}{$\progs, \questions, \fbf$}
    \Statex\Input{$\progs$ is a hypothesis space, $\questions$ is a question space, and $\fbf$ is a feedback formula.}
    \Statex\Output{Question $\question^* \in \questions$ with highest pruning power}

    \vspace{0.05in}
    \State $\question^* \assign \bot \ ; \ \beta^* \assign 0$ \algcmt{$q^*$ denotes the best question so far, and $\beta^*$ is its pruning power}
    \State $\worklist \assign \textsf{SortedSet}( )$ \algcmt{Worklist sorted in decreasing order of pruning power}
   \vspace{0.04in}
    \ForAll{$(f, \inp) \in \questions$} \algcmt{Worklist initialization}
    \State $\eta \assign \text{max}_{\answer_i} \ | \{ \prog \in \progs \ | \ \prog \models_k (\fbf \land f(\inp) = a_i) \} |$ \algcmt{Refine HS using \textsc{BCE}}
    \State $\beta \assign {(|\progs| - \eta)}/{|\progs|}$ \algcmt{$\beta$ overapproximates the pruning power of $(f, \inp)$}
    \State $\worklist \assign \worklist \cup ((f, \inp), \beta)$
    \EndFor 

    \While{$\worklist \neq \varnothing$}
    \State $(\question, \beta) \assign \worklist\textsf{.first}()$ \algcmt{Process question with best overapproximate pruning power}
    \If{$\beta^* \geq \beta$} \text{break;} \algcmt{Found question with best pruning power}
    \EndIf 
    \State $\eta \assign \text{max}_{\answer_i}  |\textsc{RefineHS}(\progs, \question, \fbf \land f(\inp) = a_i)|$
    \State $\beta' \assign  {(|\progs| - \eta)}/{\progs}$ \algcmt{Compute actual pruning power of $q$}
    \If{$\beta' > \beta^*$} 
     $\question^* \assign \question \ ; \ \beta^* \assign \beta'$ \algcmt{Found question with higher pruning power}
    \EndIf
    \EndWhile 
    \State \Return $\question^*$
    % \State $\worklist\text{.append}( \ ( \text{\sc{GetPruningPower}}(\progs_{\assump}, q, \assump, \tau), q) \ )$
    % \EndFor
    % \State $\worklist\text{.sort}(\text{reverse}=\text{\emph{true}})$
    % \State $q_{opt}, \beta_{opt} \assign \bot, 0$
    % \While{\emph{true}}
    % \State $\question' \assign \worklist\text{.pop}()\text{.second}()$
    % \State $\beta' \assign \text{\sc{GetPruningPower}}(\progs_{\assump}, \question', \assump, 1)$
    % \If{$\beta' > \beta_{opt}$}
    % \State $\question_{opt} \assign q'$
    % \State $\beta_{opt} \assign \beta'$
    % \If{$\worklist\text{.length} == 0 \ \vee \ \beta_{opt} > \worklist\text{.peek}()\text{.first}()$}
    % \State \Return $\question_{opt}$
    % \EndIf 
    % \EndIf 
    % \EndWhile
    \EndProcedure
    \end{algorithmic}
    \end{algorithm}
    \vspace{-0.4in}
    \caption{Question selection algorithm. }
    \vspace{-0.1in}
    \label{fig:qselect}
\end{figure}

\begin{definition}{\bf (\textsc{BCE}-consistency)} Let $\prog$ be a program and $\fbf_\fb$ be user feedback. We say that $P$ is \textsc{BCE}-consistent with $\fbf_\fb$, denoted $\prog \models_k \fbf_\fb$, iff
$ 
    \forall (\synthfunc, \inp, \out) \in \fb. \ \out \in \confsem{P}_{\fbf_\fb}^k(\inp)
$.
\end{definition}

Our question selection algorithm utilizes the observation that \textsc{BCE}-consistency can be used to over-approximate the pruning power of a question, as stated in the following theorem:

\begin{theorem}
Let $q = (f, I)$ be a question and $\progs$ be a program space consistent with $\fbf$. Then:
{\small 
\begin{align*}
    \pp{q, \progs} \leq  \frac{|\progs| - \underset{i}{\text{max}}  \ | \{ \prog \in \progs \ | \ \prog \models_k \fbf \land f(I) = a_i \}  |}{|\progs|}
\end{align*}
}
\end{theorem}

We now discuss the {\sc SelectQuestion} algorithm summarized in Figure~\ref{fig:qselect}. This algorithm takes as input a hypothesis space $\progs$ consistent with feedback $\fbf$ and a set of questions $\questions$, and returns a question $q^* \in \questions$ such that $\forall q \in \questions, \ \pp{q^*, \progs} \geq \pp{q, \progs}$. To compute $q^*$, the algorithm maintains a worklist $\worklist$ of pairs $(q, \beta)$, where $q$ is a question and $\beta$ is an upper bound on its pruning power. The worklist is initialized in lines 4–7 using \textsc{BCE}-consistency (line 5). The second loop (lines 8–13) processes elements in \emph{decreasing} order of over-approximated pruning power, updating the best question $q^*$ and highest pruning power $\beta^*$ encountered. Each iteration removes the element with the highest possible pruning power $\beta$ under \textsc{BCE} semantics. If $\beta^*$ exceeds $\beta$, the search terminates, as $\beta$ is an upper bound for all remaining questions. Otherwise, the algorithm computes the \emph{actual} pruning power $\beta'$ of $q$ by calling {\sc RefineHS}. If $\beta'$ is larger than $\beta^*$, $q^*$ and $\beta^*$ are updated.

%This algorithm takes as input a hypothesis space $\progs$ consistent with feedback $\fbf$ and a space of questions $\questions$ and returns a question $q^* \in \questions$ such that $\forall q \in \questions. \  \pp{q^*, \progs} \geq \pp{q, \progs}.$ To compute such a question $q^*$, the algorithm maintains a worklist $\worklist$ of pairs $(q, \beta)$ where $q$ is a question and $\beta$ is an upper bound on the pruning power of $q$. The worklist is initialized in lines 4--7  using the notion of \textsc{BCE}-consistency (see line 5). The second loop in lines 8--13 of the algorithm processes elements in the worklist in \emph{decreasing} order of their over-approximated pruning power and may update the best question $q^*$ and the highest pruning power $\beta^*$ encountered so far. Specifically, each iteration of the loop removes the element with the \emph{highest} possible pruning power $\beta$ under \textsc{BCE} semantics. If $\beta^*$ is larger than $\beta$, we can terminate the search as $\beta$ is an upper bound on the pruning power of all remaining questions in $\worklist$. Otherwise, the algorithm computes the \emph{actual} pruning power $\beta'$ of $q$ by calling the {\sc RefineHS} procedure defined earlier. If  $\beta'$ is larger than $\beta^*$, then $q$ is the best question encountered so far; thus, both $q^*$ and $\beta^*$ are updated.

\begin{theorem}
The {\sc SelectQuestion} procedure returns a question $q^* \in \questions$ with the highest pruning power -- i.e.,
$
\forall q \in \questions. \ \pp{q, \progs} \leq \pp{q^*, \progs}
$.
\end{theorem}

\noindent
{\bf \emph{Remark 1.}} Any evaluation strategy that over-approximates pruning power (i.e., under-approximates $\hs{\progs}{\fbf}$) could be used instead of \textsc{BCE} without affecting the optimality of {\sc SelectQuestion}. 
%As described  in Section~\ref{sec:impl}, one of our application domains uses a variation of \textsc{BCE} where we sample a proportion of the elements in the prediction set, up to some bound $k$, instead of using the same $k$ value throughout.  

\vspace{0.05in}
\noindent
{\bf \emph{Remark 2.}}
Since $\textsc{SelectQuestion}$ procedure performs $\textsc{RefineHS}$ at most once for every (question, answer) pair, it can require up to  $\mathcal{O}(|\progs| \times |\inps| \times |\questions| \times |\textsf{Answers}(\questions)|)$ calls to \textsc{CCE}. While the theoretical worst-case complexity of {\sc Distinguish} is the same as that of  $\textsc{SelectQuestion}$, empirically, we find that  $\textsc{SelectQuestion}$ dominates the runtime of $\textsc{ActiveLearning}$ (see Section ~\ref{sec:runtime}). 
\section{Implementation}\label{sec:impl}
We have implemented the proposed active learning technique as a new tool called \toolname{} written in Python. \toolname{} can be customized for different application domains by (1) instantiating the meta-DSL from Figure~\ref{fig:metalanguage} with a specific DSL,  (2) defining a suitable abstract domain, along with appropriate forward and backward abstract transformers, and (3) supplying a synthesis engine that can be used to construct the initial hypothesis space for that domain.

\vspace{-0.05in}
\paragraph{{\bf \emph{Sampling.}}}  
%While our implementation closely follows our formal presentation, it performs one important additional optimization:  
%In practice,  the hypothesis space may initially contain tens of thousands of programs; thus, computing the pruning power of each question on the entire hypothesis space can be infeasibly expensive. 
Following prior work~\cite{samplesy}, we  compute the pruning power of each question on a sampled subset of the hypothesis space. As shown in ~\cite{samplesy}, this sampling procedure assures that the selected query has pruning power $\varepsilon$-close to that of the optimal query with bounded probability.

\vspace{-0.05in}
\paragraph{{\bf \emph{BCE generalization.}}} Recall that our presentation of the question selection algorithm uses BCE to over-approximate the pruning power of a question. 
% However, as mentioned in Remark 1 of Section~\ref{sec:select}, the question selection algorithm can utilize any technique that under-approximates the hypothesis space. 
\toolname{} implements a generalized version of BCE  that samples $\mathsf{min}(k, n \times k')$ elements from the prediction set, where $k$ is a positive integer, $k'$ is a real-valued percentage in the range $(0, 1)$, and $n$ is the size of the  prediction set. This generalized BCE  strategy allows  sampling  some percentage $k'$ of  the  prediction set up to some limit $k$. All hyper-parameters used in our implementation, including $k$ and $k'$,  are chosen for each domain based on a validation set. 
%BCE proportionally sampled 25\% of the elements in the prediction set. These settings yielded the best performance of $\textsc{SelectQuestion}$ in our experiments.
%This evaluation strategy has a hyperparameter $k$ specifying the fixed beam length. In the MNIST domain, $k$ was set to 1. In the image editing domain, BCE proportionally sampled 25\% of the elements in the prediction set. These settings yielded the best performance of $\textsc{SelectQuestion}$ in our experiments.
\paragraph{\bf \emph{Computing prediction sets for neural functions}}
\change{Our evaluation technique utilizes conformal prediction, a technique for augmenting machine learning methods with prediction sets that guarantee coverage of the true labels with some user-specified probability. Given any model $f: \mathcal{X} \to \mathcal{Y}$, the goal is to learn a model $\tilde{f}: X \to 2^{\mathcal{Y}}$ that produces sets of labels that include the true label with high probability. Critically, this guarantee should be achieved regardless of how well or poorly the initial model $f$ performs. We assume the existence of a \textit{calibration dataset} $Z \in \mathcal{D}^n$ with which to train $\tilde{f}$, and a \textit{non-conformity score} $s_{f} : \mathcal{X} \times \mathcal{Y} \to \mathbb{R}$, which is a heuristic measure of how ``close'' the prediction $f(x)$ is to any label $y$. Intuitively, a smaller score $s_f(x,y)$ means that $y$ has a higher likelihood of being the true label of $x$. The basic idea of conformal prediction is to consider the collection of prediction-set functions $\{\tilde{f}_{\tau}\}_{\tau \in \mathbb{R}}$ parameterized by a thresholding score-value $\tau$:
\begin{equation*}
\tilde{f}_{\tau}(x) = \{y \in \mathcal{Y}: s_{f}(x, y) \leq \tau \}
\end{equation*}
The calibration dataset $Z$ gives an empirical distribution over the non-conformity scores induced by $\mathcal{D}$, and using this, a thresholding value $\tilde{\tau}$ can be chosen to achieve a guarantee of the  form: 
\begin{equation*}
\mathbb{P}_{Z \sim \mathcal{D}^{n}, (x,y) \sim \mathcal{D}}[y \in \tilde{f}_{\tilde{\tau}}(x)] \geq 1 - \delta
\end{equation*}
Here, $\mathcal{D}$ is the fixed distribution from which data is drawn, and $\delta \in (0, 1]$ is the miscoverage rate. Although conformal prediction makes no assumptions on how $s_f$ is defined, a well-constructed non-conformity score often results in smaller prediction sets.}
\section{Evaluation}\label{sec:eval}

In this section, we describe the results of our experimental evaluation, which aims to answer the following research questions: 

\begin{itemize}[leftmargin=*]

\item \textbf{RQ1:} Can our active learning approach identify the ground truth program, and how do the results compare against prior work on active learning?

\item \textbf{RQ2:} How many rounds of user interaction are required, and how does this compare against alternative question selection approaches?
 
\item \textbf{RQ3:} How important are our key algorithmic ingredients for the runtime and scalability of our approach?
\item \textbf{RQ4:} Which components of \toolname dominate its runtime in practice?
\end{itemize} 

\subsection{Application Domains and Benchmarks}\label{sec:setup}

To answer our research questions, we instantiate \toolname for three application domains from prior work: batch image editing, image search, and visual arithmetic. In what follows, we provide a brief description of each of these domains and their DSLs. The neural components of each DSL are denoted in \textbf{bold}. We refer the  interested reader to 
Appendix ~\ref{sec:domains2}
%the Appendix of the extended version of the paper \cite{extended} 
for more details about the abstract domains and corresponding abstract transformers.

\begin{wrapfigure}{r}{0.4\textwidth}
    \centering
    \small
    \vspace{-.5cm}
    \begin{minipage}{.4\textwidth}
    \[
    \begin{array}{r l}
        \prog := & \{ \extractor \shortrightarrow A, \cdots, \extractor \shortrightarrow A\} \\
         A := & {\sf Blur} \ | \ {\sf Brighten} \ | \ {\sf Crop} \ | \ \cdots \\
         \extractor := & {\sf All} \ | \ \textsf{\textbf{Is}}(\varphi) \ | \ {\sf Complement}(E) \\ 
         | &{\sf Union}_N(E_1, \cdots, E_N) \\ 
         | &{\sf Intersect}_N(E_1, \cdots, E_N) \\
         | & {\sf Find}(\extractor, \varphi, f) \ | \ {\sf Filter}(E, \varphi)  \\
          \varphi := & {\bf {\sf \textbf{\textsf{Object}}}}(O) \ | \ {\bf {\sf \textbf{\textsf{Smiling}}}} \ | \cdots \ \cdots \\
         f := & {\sf GetLeft} \ | \  {\sf GetRight} \ | \  {\sf GetAbove} \ | \ \cdots
    \end{array}
    \]
    \end{minipage}
    \vspace{-.3cm}
    \caption{\imageedit DSL. }
    \label{fig:imageeditdsl}
\end{wrapfigure}

\subsubsection{Batch Image Editing} Our first application domain,  \imageedit, consists of the neurosymbolic image editing DSL (see Figure ~\ref{fig:imageeditdsl}) and batch image editing tasks considered in recent work~\cite{imageeye}. This DSL uses neural components (e.g. image segmentation, object classification, facial recognition, etc.) to construct a symbolic representation of an image, and applies manipulations (e.g. crop, blur, etc.) to \emph{parts} of the image that match some criterion. For this application, we use the same abstract domain proposed in ~\cite{imageeye} but additionally define backward abstract transformers for each DSL construct. We evaluate on the \emph{same tasks} from ~\cite{imageeye}, excluding 13 text-related benchmarks due to the lack of conformal predictors for the underlying neural networks. The input spaces for these tasks consist of the same image datasets considered in ~\cite{imageeye} (e.g. an album of 359 wedding photos collected from Flickr). 
%\change{40.8\% of these images contain at least one neural misprediction, such as an undetected object or a misclassified attribute.}  
Note that neural components in this domain perform binary classification; hence, the maximum size of each prediction set for this domain is 2.

\change{Notably, our experimental setup differs from that of ~\cite{imageeye}, wherein the reported results rely on repeated synthesis attempts guided by an external oracle. This oracle checks whether the synthesized program is observationally equivalent to the ground-truth program and provides additional examples as needed. Our evaluation setting does not assume access to such an oracle.}

\begin{wrapfigure}{r}{0.35\textwidth}
     \centering
     \small
     % \vspace{-.5cm}
     \begin{minipage}[b]{0.35\textwidth}
         \centering
  % \large
    \[
    \begin{array}{r l}
        E  := &  r(t_1, \ldots, t_n) \\
         | & \ E \rightarrow E  \ | \ E \wedge E \ | \ E \vee E \ \\
         | &\ \neg E \ |  \ \exists x. E \ | \ \forall x. E \\

        r  := & \textbf{\textsf{HasType}} \ | \ \textbf{\textsf{HasEmotion}} \\
           | & \textbf{\textsf{HasProperty}} \ | \ \textsf{HasRelation}  \\ 
                t := & x \ | \ c \\
    \end{array}
    \]
    \vspace{-.4cm}
         \caption{\imagesearch DSL.}
         \label{fig:imagesearchdsl}
     \end{minipage}
\end{wrapfigure}

\subsubsection{Image Search and Visual Concept Discovery} Our second domain, \imagesearch, involves tasks that require finding all images in a corpus that exhibit a certain property (e.g., containing both a dog and a cat) or that involve a ``visual concept'' (e.g., guitar player, still life photo). Such tasks have been considered both in the neurosymbolic synthesis literature~\cite{barnaby2024photoscout} as well as machine learning and computer vision literature ~\cite{visualconcepts,Sun_2015_ICCV}. We consider the DSL (see Figure ~\ref{fig:imagesearchdsl}) and tasks from  ~\cite{barnaby2024photoscout} as well as a set of visual concept discovery tasks~\cite{visualconcepts}.  The input space for tasks in this domain consists of the same real-world images as the \imageedit domain. Furthermore, the abstract domain is the same as \imageedit, with transformers adapted to the image search DSL. In this domain, neural networks also perform binary classification; hence the maximum prediction set size is also 2.

\subsubsection{Visual Arithmetic}

\begin{wrapfigure}{r}{.25\textwidth}
 \vspace{-0.4in}
    \centering
    \small
    \begin{minipage}{.25\textwidth}
    \begin{align*}
        \prog := & \ \lambda \ l . \ E \\ E :=& \ l \ | \ c \in \mathbb{N} \ | \ \texttt{fold} \ f \ c \ E  \\ &| \ \texttt{map} \ g \ E \ | \ \texttt{filter} \ h \ E \\ 
        f :=&  \ \texttt{sum} \ | \ \texttt{max} \\ &| \ \texttt{product} \ | \ \texttt{inc} \\ g :=& \texttt{curry} \ f \ c \ | \ \texttt{\textbf{toDigit}} \\ h :=& \lambda x . x \lhd c 
    \end{align*}
    \vspace{-0.3in}
    \caption{\digits DSL.}
    \label{fig:mnistdsl}
     \vspace{-0.1in}
    \end{minipage}
\end{wrapfigure}

As a third application domain, we consider \emph{visual arithmetic} tasks (e.g., adding or multiplying hand-written digits) from prior work on neurosymbolic programming~\cite{deepproblog,dolphin, mnisttasks1, mnisttasks2}. These tasks involves performing computations over a list of digit images using functional combinators such as \texttt{map}, \texttt{fold}, and \texttt{filter}, using the DSL in Figure ~\ref{fig:mnistdsl}.
% However, since prior work~\cite{deep-problog} only considers a very limited range of tasks (e.g, addition of two  hand-written digits), we generate benchmarks for this domain through a small user study in which we instruct 10 participants (who are not authors of this paper) to describe computations that they might want to perform for a given  list of digit images. 
The input space for this application domain consists of lists of digits from the SVHN dataset~\cite{svhn}, which is a more challenging version of the MNIST dataset ~\cite{mnist} comprised of images of digits in natural scene images. %\change{13.0\% of these inputs contain at least one digit misclassification.} 

Since prior work does not introduce an abstract domain for this  application, we consider a standard interval abstraction for each list element. Unlike the first two domains where all neural components perform binary classification, the neural constructs in this DSL perform digit classification; hence, the maximum size of prediction sets for this domain is 9. 

\vspace{1em}

To give the reader some intuition about the tasks considered in our evaluation, Table~\ref{tab:example} shows two representative tasks for each domain, along with their corresponding ground truth programs.

\subsection{Experimental Set-up and Methodology}
We evaluate \toolname (as well as all baselines and ablations) using the following methodology: First, we sample two inputs from the input space and supply the corresponding output, giving us a set $\inout$ of initial input-output examples.
Then, we compute  $\progs = \{\prog \ | \ \forall (\inp, \out) \in \inout. \ \out \in \confsem{\prog}(\inp) \}$ by enumerating all programs up to a fixed AST size of 20. 

%which is on par with program sizes considered in prior program synthesis work~\cite{STUFF}. 
We then use \toolname{} and all baselines to perform active learning on $\progs$, $\questions$, $\inps$ and $\inout$, using the ground truth programs and labels to emulate a user responding to queries. To account for variability from the random IO examples, we repeat all  experiments five times, each with a different seed, and report the mean and standard deviation of the outcomes. All experiments are run on a 2022 MacBook Pro with an 8-core m2 processor and 8 GB RAM, with a timeout limit of 600s. 

% \change{ {\bf \emph{Remark.}} Following standard practice in synthesis literature ~\cite{imageeye, feser2015synthesizing}, our experiments use enumerative synthesis to construct the initial hypothesis space. However, our active learning approach is compatible with any hypothesis space representation that (1) guarantees inclusion of the ground-truth program and (2) supports pruning of inconsistent candidates. For example, our approach can integrate with version-space algebras (VSAs) and tree automata.}

\subsubsection{Benchmark statistics.}  
 As summarized in Table~\ref{tab:details}, we evaluate \toolname on 112 tasks. On average, the initial hypothesis space contains around 1510 programs, the  question space contains around 800, and each question has approximately 9.8 possible answers. \change{Recall from Section ~\ref{sec:active_learning_problem} that the question space $\questions$ comprises (1) questions about the ground-truth output for each input, and (2) questions about the ground-truth label for each (neural component, input) pair. Hence, $\questions$ is fixed for a given task.} As expected, prediction set sizes are larger in the {\sc VisArith} domain, since the prediction space includes 9 labels rather than 2 for binary classification. However, the question space is significantly larger for {\sc ImageEdit} and {\sc ImageSearch} due to (1) the larger input space, and (2) the greater number of neural components in their DSLs. 

 \change{To quantify the quality of the neural components used in our experiments, we also report the average percentage of inputs with at least one misprediction. In the \imageedit and \imagesearch domains, slightly over half of  inputs contain a neural misprediction, on average. In these domains, a misprediction may be an undetected object, an erroneous object detection, or a misclassified attribute.  In {\sc VisArith}, 13.0\% of inputs contain a neural misprediction, on average. In this domain, a misprediction is a misclassified digit. These statistics demonstrate the systemic uncertainity of neural mispredictions across different domains and models. } 
 
Together, these three domains span diverse characteristics, forming a comprehensive test bed for evaluating our approach.

\begin{table}[t]
\caption{Representative tasks in the \imageedit, \imagesearch, and \digits domains.}
\vspace{-.1in}
    \centering
    \footnotesize
    \renewcommand{\arraystretch}{1.2} % Adjust row height slightly for better spacing
    \begin{tabular}{| m{0.11\textwidth} | m{0.22\textwidth} |  m{0.59\textwidth} |}
        \hline
        \textbf{Domain} & \textbf{Task Description} & \textbf{Ground Truth Program} \\
        \hline
         \imageedit & Crop out everyone who is not smiling or has their eyes closed. & 
         \vspace{-1cm}
        $\begin{aligned}
        \{ &\textsf{Intersect}(\textbf{\textsf{Is}}(\textsf{\textbf{Object}}( \texttt{face})),  \\
        & \textsf{Complement}(\textsf{Intersect}(\textbf{\textsf{Is}}(\textsf{\textbf{Smiling}}), \textbf{\textsf{Is}}(\textsf{\textbf{EyesOpen}})))) \rightarrow \textsf{CropOut}  \}
        \end{aligned}$ 
        \\
        \hline
        \vspace{.4cm}
         \imageedit & 
         \vspace{.3cm}
         Blur the faces of people who are not playing guitar. & 
         \vspace{.1cm}
        $\begin{aligned}
        \{ &\textsf{Intersect}(\textbf{\textsf{Is}}(\textsf{\textbf{Object}}( \texttt{face})),  \\
        & \textsf{Complement}(\textsf{Find}((\textbf{\textsf{Is}}(\textsf{\textbf{Object}}(\texttt{guitar})), \\ 
        & \ \ \ \ \ \ \ \ \ \ \ \ \ \ \ \ \ \ \ \ \ \ \ \ \ \ \ \ \ \ \ \ \ \ \ \textbf{\textsf{Is}}(\textsf{\textbf{Object}}(\texttt{face})), \textsf{GetAbove})))) \rightarrow \textsf{Blur}  \}
        \end{aligned}$ 
        \\
        \hline
        \vspace{.3cm}
        \imagesearch & \vspace{.3cm} Find all images that contain a cyclist wearing a helmet. & 
        % \vspace{.1cm}
        $\begin{aligned} 
        \exists x. \exists y. \exists z. &\textbf{\textsf{HasType}}(x, \texttt{bicycle}) \wedge \textbf{\textsf{HasType}}(y, \texttt{person}) \\ 
        \wedge &\textbf{\textsf{HasType}}(z, \texttt{helmet}) \wedge \textsf{HasRelation}(x, y, \texttt{below}) \\
        \wedge &\textsf{HasRelation}(y, z, \texttt{below})
        \end{aligned}$ 
        \\
        \hline
        % \vspace{.3cm}
        \imagesearch & 
        % \vspace{.3cm} 
        Find all images that do not contain any people. & 
        % \vspace{.1cm}
        $\begin{aligned} 
        \forall x. \neg \textsf{\textbf{HasType}}(x, \texttt{person})
        \end{aligned}$ 
        \\
        \hline
        \digits & Compute the sum of a list after doubling all items. & 
        $\begin{aligned} 
        \texttt{fold } \texttt{plus } 0 \ (\texttt{map} \ (\texttt{curry } \texttt{prod } 2) \ (\texttt{map } \textbf{\texttt{toDigit }} l))
        \end{aligned}$ 
        \\
        \hline
        \digits & Compute the maximum list item that is smaller than the head of the list. & 
        $\begin{aligned} 
        \texttt{fold } \texttt{max } 0 \ (\texttt{filter} \ &(\texttt{curry } \lambda.x < (\texttt{head} \ l)) (\texttt{map } \textbf{\texttt{toDigit}} \ l))
        \end{aligned}$ 
        \\
        \hline
    \end{tabular}
    \label{tab:example}
    \vspace{-.1in}
\end{table}

\begin{figure}
\vspace{-.2in}
    \begin{table}[H]
    \caption{Details about the (1) number of tasks, average sizes of (2) initial hypothesis space, (3) final hypothesis space, (4) input space, (5) question space, (6) answer space per question, (7) prediction set, and (8) \change{average percentage of inputs with mispredictions.}}
    \vspace{-0.1in}
\footnotesize
\begin{tabular}{|c|c|c|c|c|c|c|c|c|}
    \hline
     \makecell{\textbf{Domain}} & \# \textsf{Tasks} & $| \mathsf{Initial \ HS}|$ & $ | \mathsf{Final\ HS}|$ & $|\inps|$ & $|\questions|$ & $|\mathsf{Possible}(q)|$& $|\mathsf{Pred. \ Set}|$ & \makecell{\change{\textsf{\% Inputs}} \\ \change{\textsf{w/ Mispred.}}} \\ 
     \hline
     \textsc{ImageEdit} & 37 & 1731.0 & 5.7 & 271.7 & 1050.8 & 11.3 & 1.2 & \change{53.0\%} \\ 
     \hline 
     \textsc{ImageSearch} & 25 & 593.8 & 3.7 & 257.4 & 1135.1 & 12.7 & 1.1 & \change{59.0\% }\\ 
     \hline
     \makecell{\digits} & 50 & 1813.7 & 32.1 & 100.0 & 444.0 & 7.3 & 2.5 & \change{13.0\%} \\
     \hline
     \makecell{\textbf{Overall}} & 112 & 1514.1 & 17.1 & 191.9 & 797.4 & 9.8 & 1.8 & \change{30.0\%} \\
     \hline
\end{tabular}
\vspace{-.35in}
\label{tab:details}
\end{table}
\end{figure}

\subsection{Comparison Against Existing Active Learning Techniques}

To answer our first research question, we compare \toolname against two existing active learning techniques proposed in prior work for traditional program synthesis:
\begin{enumerate}[leftmargin=*]
\item \textsc{{\textbf{SampleSy}}} ~\cite{samplesy},  an active learning procedure whose question selection algorithm uses an objective function similar to that of \toolname. In each round of interaction, \textsc{SampleSy} selects the question whose worst answer eliminates the largest number of programs. \item \textsc{{\textbf{LearnSy}}}~\cite{learnsy}, an active learning procedure that selects the input maximizing the approximate number of pairs of programs whose outputs differ.
\end{enumerate}

Importantly, both of these baselines evaluate program consistency using the DSL’s standard semantics -- that is, they treat the neural prediction as the ground truth and produce a single output when executing a neurosymbolic program on a given input. \change{While prior work proposes synthesis techniques for similar application domains (e.g., PhotoScout ~\cite{barnaby2024photoscout} solves image search tasks using a neurosymbolic DSL), these techniques do not explicitly address the neurosymbolic active learning problem. Hence, such approaches are not baselines for evaluating the contributions of our work. }

%This assumption impacts both the termination condition of active learning and how each candidate question is scored during question selection.

% The existing implementations of these baselines utilize a VSA-based approach that is not tractable for our domains; hence, we used our own implementations in our experiments. 
%For each baseline, we perform active learning using the experimental setup described in Section ~\ref{sec:setup}. Since prior  techniques assume that the semantics of all DSL constructs are precise, their question spaces only involve questions about the target program being synthesized. Under the same assumption, these baselines can erroneously prune the ground truth program. 
%comprised of only queries about the ground truth program, and do not contain queries about the ground truth labels of neural components. 
%Further, the evaluate programs using the evaluation semantics rather than conformal semantics. Thus, the baselines may erroneously prune the ground truth program.

The results of this experiment are summarized in the ``\# Benchmarks Solved'' column of Table~\ref{tab:results1}. Here, a benchmark is considered solved if the synthesizer returns a program that is observationally equivalent to the ground truth program under the ground truth semantics. For \toolname, failure to return the ground truth program can occur for two reasons: (1) timing out, and (2) incompleteness of conformal prediction. Although conformal prediction includes the ground truth label in the prediction set with \emph{high probability}, there is a small but non-zero chance that it will be excluded. In the baselines, failure can also occur due to timing out or due to mispredictions by the neural classifiers. Unlike \toolname, the baselines do not use conformal semantics, making it more likely for the ground truth program to be pruned from the search space.
%Recall that conformal prediction may fail to include the ground truth in the prediction set with bounded probability; if this happens, \toolname may fail to return the ground truth program.

As shown in Table~\ref{tab:results1}, \toolname significantly outperforms the baselines in terms of its ability to converge to the ground truth program:  \toolname finds the ground truth program for over 98\% of the benchmarks, whereas \textsc{SampleSy} and \textsc{LearnSy} output the correct program for 65\% and 58\% of benchmarks, respectively. 
 %This gap highlights the benefit of combining conformal semantics with an active learning method that tolerates uncertainty in neural predictions.
%The differences in success rate are overwhelmingly due to using standard semantics vs.  our proposed constrained conformal semantics.  

\begin{wrapfigure}{r}{0.52\textwidth}
\vspace{-.3cm}
    \captionof{table}{Experimental results comparing the number of benchmarks solved by \toolname{} versus two baseline active learning techniques.}
    \label{tab:results1}
    \vspace{-.2cm}
\scriptsize
\begin{tabular}{|c?c|c|c|}
    \hline
     \makecell{ \\ \textbf{Domain}} 
     & \multicolumn{3}{c|}{\makecell{\textbf{\# Benchmarks Solved}}} \\ 
     &  \multicolumn{1}{c}{\textsf{Ours}} & \multicolumn{1}{c}{\textsc{SampleSy}} & \multicolumn{1}{c|}{\textsc{LearnSy}} \\ 
     \hline
     \makecell{\imageedit}  & $95.7\% \pm 2.4\%$ & $70.3\% \pm 5.1\%$ & $58.9\% \pm 14.6\%$ \\ 
     \hline
     \makecell{\imagesearch} & $99.2\% \pm 1.8\%$ & $72.8\% \pm 8.8\%$ & $55.2\% \pm 12.4\%$ \\
     \hline 
     \makecell{\digits} & $100\% \pm 0\%$ & $57.2\% \pm 5.4\%$ & $58.8\% \pm 2.3\%$  \\
     \hline
     \makecell{\textbf{Overall}} & $98.4 \pm 0.7\%$ & $65.0\% \pm 5.0\% $ & $58.0\% \pm 6.0\%$ \\ 
     \hline
     
\end{tabular}
\end{wrapfigure}

To gain insight about the failures of \toolname, we manually inspected the 9 cases out of 560 runs where \toolname fails to output the intended program.  Out of these 9 cases, 2 are due to timeouts. For the remaining 7 cases, failure is caused by the intended program's output set not containing the ground truth for a particular input. For all of these 7 benchmarks, we confirmed that our method does converge to the intended program if we manually increase the confidence threshold used for conformal prediction, thereby making conformal prediction more conservative (albeit at the cost of making user interaction times longer). We also performed a similar investigation for the failure cases of the two baselines: {\sc SampleSy} fails on 196 out of the 560 runs, whereas {\sc LearnSy} fails on 235. There are no timeouts for either of the baselines; hence, all failures are caused by using standard semantics instead of conformal semantics. These results underscore the need for conformal semantics in the context of neurosymbolic active learning.

\vspace{0.05in}
\idiotbox{RQ1}{\toolname's use of constrained conformal semantics leads to significantly higher success rates in finding the desired program compared to the baselines.}

% \toolname{} solves an average of \todo{X}\% of benchmarks, while the other two baselines fail to output the desired program for \todo{X}.}

\subsection{Evaluation of Rounds of User Interaction}\label{sec:rounds}

To answer our second research question, we evaluate the number of rounds of user interaction required to converge to the ground truth program. To further evaluate our question selection strategy, we also compare the worst-case pruning power objective against two  alternative question selection strategies:

\begin{itemize}[leftmargin=*]
\item {\bfseries\scshape Expected:}  This variant selects questions that maximize the \emph{expected} number of programs prunes, rather than the worst-case. Specifically, for each candidate question \( q = (f, I) \), we compute the prediction set \(\{a_1, \ldots, a_n\}\) for \(f(I)\), along with the associated confidence scores from the neural predictor. These scores are normalized to form a probability distribution \(\{p_1, \ldots, p_n\}\) over possible answers. The expected pruning power of \(q\) is then given by: $\text{EPP}(q) = \sum_{i=1}^n p_i \cdot \Pi(q, P \mid a_i)$. Here \(\Pi(q, P \mid a_i)\) denotes the fraction of programs pruned from the hypothesis space assuming the user's answer is \(a_i\). This strategy favors questions that eliminate many incorrect programs \emph{in expectation}.
\item {\bfseries\scshape Random:} Rather than optimizing a certain objective during question selection, this alternative  question selector randomly samples a question from the question space.
\end{itemize}

\begin{table}[H]
\vspace{-.2cm}
    \caption{Experimental results comparing the average number of rounds of user interaction and average time per round of interaction taken by \toolname versus two alternative question selection strategies. Our proposed question selection strategy, which uses a worst-case pruning power optimization objective, is denoted by the \textsc{WorstCase} column. The \textsc{Expected} column denotes a strategy that uses an \emph{expected} pruning power optimization objective, and the \textsc{Random} column denotes a strategy that samples a question.}
    \vspace{-0.1in}
\scriptsize
\begin{tabular}{|c?c|c|c?c|c|c|}
    \hline
     \makecell{ \\ \textbf{Domain}} 
     & \multicolumn{3}{c?}{\makecell{\textbf{Average \# Rounds} \\ \textbf{ of Interaction}}}  
     & \multicolumn{3}{c|}{\makecell
     {\textbf{Average Time per} \\ \textbf{Round of Interaction (s)}}}  \\ 
     & \multicolumn{1}{c}{\textsc{WorstCase}} & \multicolumn{1}{c}{\textsc{Expected}} & \multicolumn{1}{c?}{{\textsc{Random}}} & \multicolumn{1}{c}{\textsc{WorstCase}} & \multicolumn{1}{c}{\textsc{Expected}} & \multicolumn{1}{c|}{{\textsc{Random}}}  \\ 
     \hline
     \makecell{\imageedit}  & $5.0 \pm 1.7$ & $6.3 \pm 2.0$ & $113.4 \pm 29.5$ & $7.3 \pm 2.8$ & $48.5 \pm 14.1$ & $0.5 \pm 0.2$ \\ 
     \hline
     \makecell{\imagesearch} &  $6.4 \pm 1.6$ & $6.5 \pm 2.2$ & $145.8 \pm 54.6$ & $4.9 \pm 1.8$ & $46.9 \pm 12.8$ & $0.1 \pm 0.0$ \\
     \hline 
     \makecell{\digits} & $4.0 \pm 0.1$ & $9.8 \pm 0.2$ & $53.5 \pm 2.2$ & $4.4 \pm 0.2$ & $11.2 \pm 0.3$ & $0.3 \pm 0.0$ \\
     \hline
     \makecell{\textbf{Overall}} & $4.9 \pm 0.5$ & $7.9 \pm 0.7$ & $93.9 \pm 17.5$ & $5.5 \pm 1.4$ & $26.8 \pm 3.6$ & $0.3 \pm 0.0$ \\ 
     \hline
\end{tabular}
\vspace{-.1in}
\label{tab:results2}
\end{table}

We emphasize that both of these alternative strategies rely on the machinery introduced in this paper (e.g., constrained conformal evaluation) to guarantee convergence to the ground truth program. Thus, the differences in performance can be attributed entirely to the choice of question selection strategy.

The results of this evaluation are presented in Table~\ref{tab:results2}. When using our proposed \emph{worst-case pruning power} optimization objective from Definition~\ref{def:pp}, our method  takes an average of 4.9 rounds of user interaction. In other words, \toolname can disambiguate between the initial 1514 programs in the hypothesis space by asking fewer than 5 questions, on average. As expected, the \textsc{Random} strategy performs significantly worse, requiring an average of 93.9 rounds of user interaction -- nearly 20 times more than our \textsc{WorstCase} strategy. While the \textsc{Expected} strategy performs better than \textsc{Random}, it still requires 7.9 rounds on average. This result is somewhat surprising, given that \textsc{Expected} leverages prediction probabilities to guide question selection. In practice, however, we found that neural networks can sometimes be confidently wrong about their predictions, which can causes the expected pruning power to prioritize misleading questions.
%ultimately slowing convergence in several cases.

Additionally, Table~\ref{tab:results2} shows the average time per user interaction round for each strategy. The \textsc{Random} strategy is extremely fast, since it avoids any computation during question selection and simply samples a query uniformly at random. However, despite its low per-round overhead, \textsc{Random} still takes significantly longer overall. Even assuming an optimistic estimate of just 3 seconds for the user to answer each query, \textsc{Random} would take over 300 seconds to converge -- more than 7 times longer than the total active learning time required by the \textsc{WorstCase} strategy.

In contrast, the \textsc{Expected} strategy incurs a much higher cost per round compared to \textsc{WorstCase}, as it requires propagating the predicted probabilities of neural components during program evaluation. Overall, these results empirically demonstrate that optimizing for worst-case pruning power yields a favorable trade-off between the \emph{number} of rounds of user interaction and the \emph{time} to compute a question per round. \\

\vspace{0.05in}
\idiotbox{RQ2}{
The optimization objective based on worst-case pruning power yields the best trade-off between number of rounds and computation time compared to two alternative question selection strategies.
%that (1) randomly sample a question, and (2) optimize \emph{expected} pruning power.
%\toolname requires under five rounds of user interaction to discover the ground truth program on average. In contrast, the random question selection strategy requires over 90 rounds, and  the ablated version of \toolname  results in an average of 27 rounds of user interaction.
}

\subsection{Ablation Studies}

In this section, we aim to evaluate the impact of our two key algorithmic optimizations through ablation studies. For this evaluation, we consider the following  ablations of \toolname:
\begin{itemize}[leftmargin=*]
\item  {\bfseries\scshape No-AbsInt:}   This variant disables our proposed bidirection abstract interpretation method in the CCE algorithm. In particular,  this ablation does \emph{not} call {\sc ForwardAI} and {\sc BackwardAI} in the CCE procedure; it also omits the consistency check in lines 7--8 of the {\sc EvalConsistent} method.  
\item  {\bfseries\scshape No-BCE:}  This ablation does not utilize bounded conformal evaluation to derive an upper bound on the pruning power of each question. In more detail, it omits the initial for loop in lines 4--7 of the {\sc SelectQuestion} procedure; it also omits the pruning test at line 13.
\end{itemize}

\subsubsection{Runtime.}
Table~\ref{tab:results3} compares the runtime of \toolname against those of the two ablations. Since these ablations do not impact the number of rounds of user interaction, we only compare the time it takes to generate a question.  Across all domains, \toolname takes an average of 5.5 seconds per round, while {\sc No-AbsInt} and {\sc No-BCE} take roughly 2$\times$ and 6$\times$ longer, respectively. This difference highlights the impact of these key optimizations on our method's practicality.

\subsubsection{Scalability.}
% We further evaluate the impact of these algorithmic invariants on the \emph{scalability} of the approach. 
As discussed earlier in Section ~\ref{sec:synth},  our algorithm scales linearly with respect to  various parameters like input and question space size; however, its worst-case time complexity is exponential with respect to \emph{prediction set size}. Here, we evaluate the effectiveness of our proposed algorithmic optimizations (namely, \textsc{CCE} with bidirectional reasoning and \textsc{BCE} for question selection) in mitigating this worst-case blowup in practice. To perform this experiment, we increase prediction set sizes by varying the parameters of conformal prediction (see Section ~\ref{sec:impl} for more details). Intuitively, by varying this parameter, we make conformal prediction more conservative, leading to larger prediction set sizes.

\begin{wrapfigure}{r}{0.45\textwidth}
\vspace{-.5cm}
\centering
\captionof{table}{Experimental results comparing the average time per round of interaction taken by \toolname{} versus two ablations.}
% \vspace{-0.2cm}
\label{tab:results3}
\scriptsize
\begin{tabular}{|c?c|c|c|}
    \hline
    \makecell{ \\ \textbf{Domain}} 
    & \multicolumn{3}{c|}{\makecell{\textbf{Average Time per} \\ \textbf{Round of Interaction (s)}}}  \\ 
    & \textsf{Ours} & No-BCE & No-AbsInt \\ 
    \hline
    \imageedit  & $7.3 \pm 2.8$ & $42.1 \pm 8.5$ & $16.9 \pm 5.1$ \\
    \hline
    \imagesearch & $4.9 \pm 1.8$ & $41.9 \pm 8.1$ & $8.9 \pm 2.2$ \\
    \hline 
    \digits & $4.4 \pm 0.2$ & $21.5 \pm 0.5$ & $8.5 \pm 0.5$  \\
    \hline
    \textbf{Overall} & $5.5 \pm 1.4$ & $33.3 \pm 3.8$ & $11.2 \pm 2.3$  \\ 
    \hline
\end{tabular}
% \vspace{-0.15in}
\end{wrapfigure}

The result of this evaluation is presented in 
Figure ~\ref{fig:scalability}, which shows how the average user interaction time increases with respect to prediction set size in each domain. In these plots, the $x$-axis represents average prediction set size across the entire input space, and the $y$-axis represents average user interaction. The blue dots show the results for \toolname{}, while the purple triangles represent the results of an ablation that uses neither bidirectional reasoning nor \textsc{BCE} (i.e. it is a combination of both \textsc{No-AbsInt} and \textsc{No-BCE}). As demonstrated by the divergence between the two sets of dots, our algorithmic optimizations are indeed effective at preventing an exponential blowup. In particular, the blue line fits the \toolname{} results roughly linearly (best fit $ax + b$, $R^2 = 0.97$ for \digits, $R^2 = 0.94$ for \imageedit, and $R^2 = 0.98$ for \imagesearch), while the purple dashed line fits the ablation results exponentially (best fit $ab^x$, $R^2 = 0.99$ for \digits, $R^2 = 0.99$ for \imageedit, and $R^2 = 0.99$ for \imagesearch). \\  %These results clearly demonstrate that our proposed algorithmic ingredients are effective at preventing exponential blowup with respect to prediction set size.

\begin{figure}
    \centering
    % \vspace{-0.1in}
    \begin{minipage}[t]{0.33\textwidth}
        \centering
        \includegraphics[width=\textwidth]{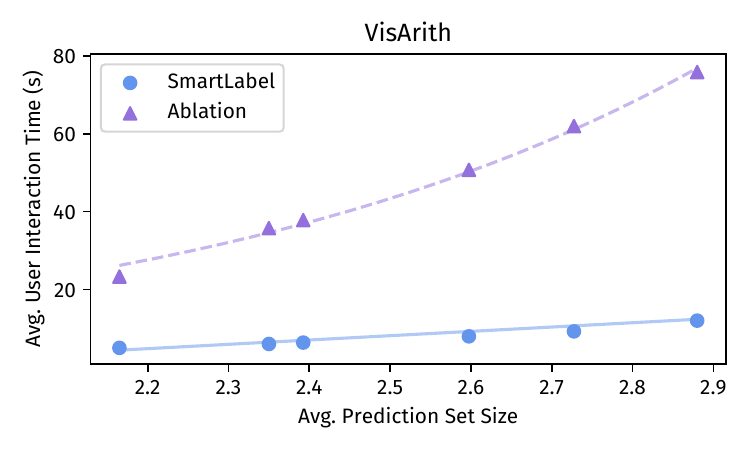}
        % \vspace{-0.3in}
        % \label{fig:bottleneck}
    \end{minipage}%
    \begin{minipage}[t]{0.33\textwidth}
        \centering
        \includegraphics[width=\textwidth]{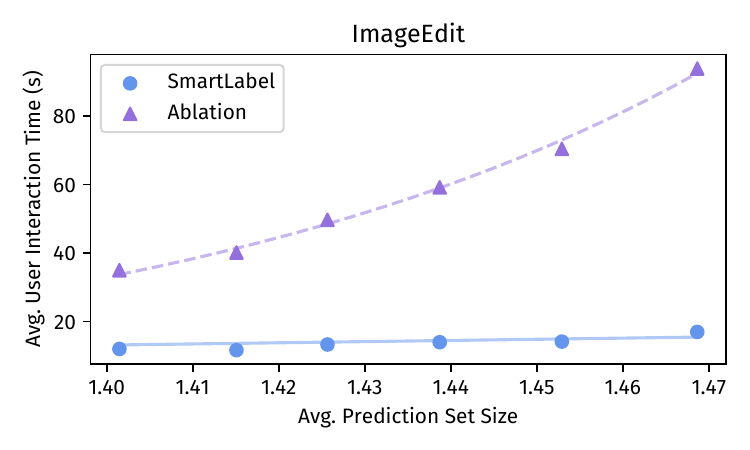}
        % \vspace{-0.3in}
    \end{minipage}
    \begin{minipage}[t]{0.33\textwidth}
        \centering
        \includegraphics[width=\textwidth]{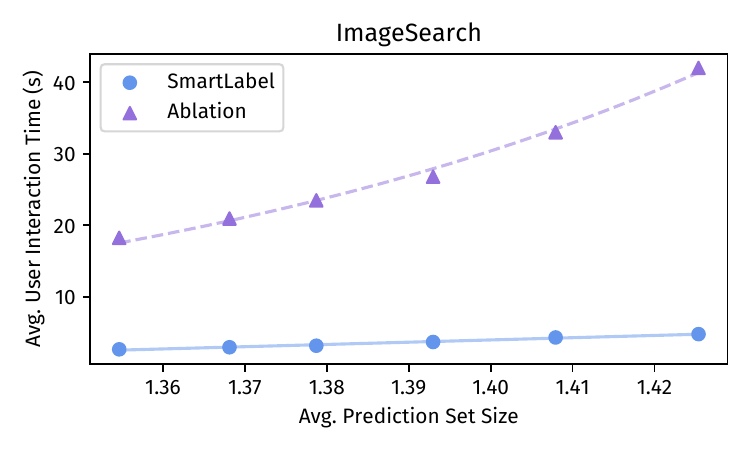}
        % \vspace{-0.3in}
    \end{minipage}
    \vspace{-.35in}
    \caption{User interaction time as prediction set size increases for each domain.}
    \label{fig:scalability}
    \vspace{-0.2in}
\end{figure}

\vspace{0.05in}
\idiotbox{RQ3}{Our key algorithmic ingredients have a significant positive impact on both average runtime per user interaction round as well as on the algorithm's scalability with respect to prediction set size. 
}

\subsection{Evaluation of the Impact of Active Learning Components on Runtime}\label{sec:runtime}

Recall that our active learning procedure involves three core components:  (1) \textsc{RefineHS}, which prunes inconsistent programs using CCE; (2) \textsc{Distinguish}, which checks whether further disambiguation is necessary; and (3) \textsc{SelectQuestion}, which identifies the next query by maximizing an objective.  
While Section~\ref{sec:synth} analyzes the theoretical worst-case complexity of these components, an interesting \emph{empirical} question is where \toolname spends the majority of its runtime in practice. Figure ~\ref{fig:bottleneck} shows how the three main components of our  algorithm contribute to user interaction time. As shown in this figure, most of the runtime (69.5\%) is consumed by the \textsc{SelectQuestion} procedure, followed by \textsc{RefineHS} (26.4\%), and only 4.1\% is spent on \textsc{Distinguish}. 
Despite its worst-case complexity, the practical efficiency of \textsc{Distinguish} can be explained by two factors. First, in early rounds of active learning, there are many distinguishable programs in the hypothesis space, and \textsc{Distinguish} only needs to find a single distinguishing pair to conclude that learning should continue -- allowing it to terminate quickly. Second, in later rounds, although distinguishing programs becomes harder, the hypothesis space is much smaller -- reducing the overall search burden. \\

\begin{wrapfigure}{r}{0.32\textwidth}
  \begin{center}
  \vspace{-.6in}
    \includegraphics[width=0.31\textwidth]{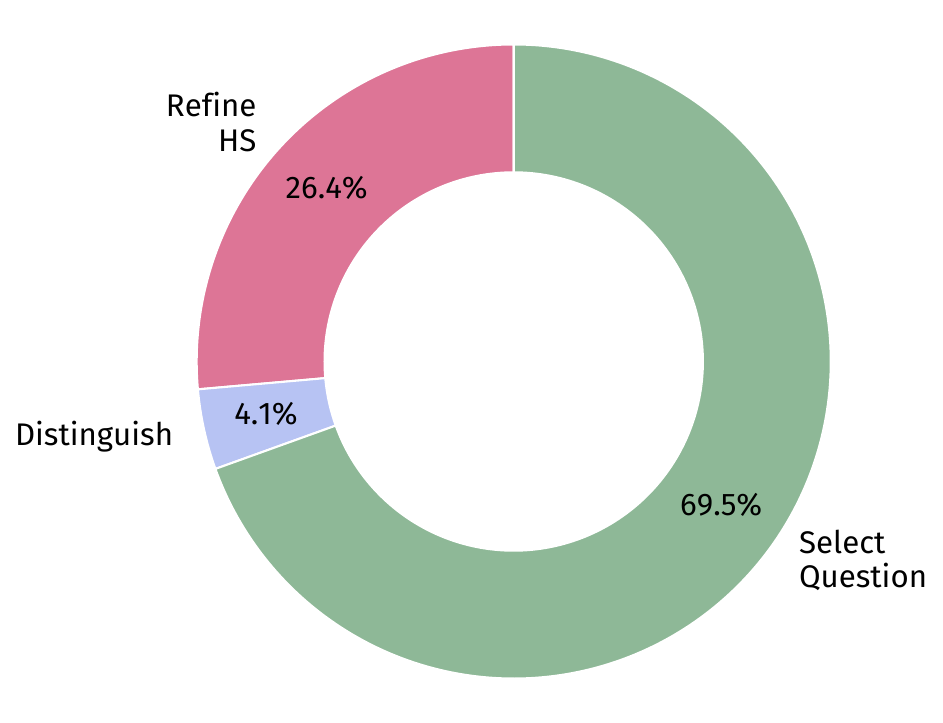}
  \end{center}
  % \vspace{-0.13in}
    \caption{Breakdown of active learning runtime.}\label{fig:bottleneck}
  \vspace{-0.5in}
\end{wrapfigure}

\vspace{0.05in}
\idiotbox{RQ4}{ The runtime of active learning is dominated by \textsc{SelectQuestion}, followed by \textsc{RefineHS}. In contrast, \textsc{Distinguish} is very fast in practice despites its worst-case theoretical complexity.  }

\section{Related Work}

\paragraph{\textbf{Active learning for program synthesis.}}  Active learning ~\cite{settles2009active, ren2021survey, bshouty1995exact, schohn2000less, dasgupta2004analysis, bshouty1994oracles} 
 is a type of machine learning approach where the algorithm selectively chooses the data from which it learns. Instead of using a large, randomly selected dataset, the model actively selects specific data points to label,   with the goal of  minimizing the number of queries to the human. 
In recent years, there has been significant interest in active learning techniques  for \emph{interactive program synthesis}. The goal of these techniques is to query the user  until there is no remaining ambiguity in the  specification. A variety of techniques ~\cite{susmit10, forest, Laich2020Guiding, flashprog, flashprofile} select a question that distinguishes at least two programs in the program space. However, these techniques do not address the \emph{optimal question selection} problem in the PBE setting, where the goal is to minimize the rounds of user interaction. Similarly to \toolname, \emph{SampleSy} and \emph{EpsSy} ~\cite{samplesy} both employ a greedy question selection algorithm that selects the question whose worst-case answer will result in the best outcome. The followup work, \emph{LearnSy} \cite{learnsy}, learns a model that predicts the likelihood that two programs will have the same output on a question. However, these prior algorithms are designed to synthesize programs in \emph{purely symbolic} languages, and do not take into account the discrepancy between the ground truth and evaluation semantics of neural constructs. As a result, they do not offer any correctness or completeness guarantees in the neurosymbolic setting.

%if the target language contains neural components, they offer no guarantee of correctness. 

% \emph{LearnSy} ~\cite{learnsy} is an interactive synthesis technique that, in each round of interaction, selects the query that maximizes the number of pairs of programs in the hypothesis space that provide different outputs on that query. To do so efficiently, \emph{LearnSy} learns a model that predicts the likelihood that two programs will have the same output on a question. \emph{LearnSy} has been shown to outperform \emph{SampleSy} and \emph{EpsSy} in the purely symbolic setting, and in theory its question selection algorithm could be adapted for the neurosymbolic setting. However, in practice this approach is impractical for two reasons. First, in the neurosymbolic setting, the question space is much larger since it includes questions about both the ground truth program and the ground truth labels of neural components. Hence, learning this model is extremely inefficient. Second, answering questions about neural components changes the output of programs on particular inputs, so the  model will become less accurate as the user answers more questions. 

\vspace{-0.05in}
\paragraph{\textbf{Neurosymbolic program synthesis.}} There is much recent work in applying program synthesis techniques to neurosymbolic DSLs. In the image processing domain, various works ~\cite{imageeye, barnaby2024photoscout, huang2020generating, ellis2018learning, johnson2017inferring, gupta2023visual, suris2023vipergpt} synthesize programs that combine symbolic operators with neural perception modules. In the data extraction and manipulation domain, several recent works ~\cite{flashgpt, chen2023data, chen2021web, jiang2021neuralizing, cheng2022binding} synthesize programs that utilize large language models to reason semantically about data. Other works ~\cite{gaunt2017differentiable, shah2020learning, valkov2018houdini} provide more general frameworks for synthesizing programs that compose neural and symbolic operators. A common limitation among these works is that the synthesized program is only guaranteed to match the user's specification under the assumption that the program's underlying neural components are always correct. Our work proposes a framework for neurosymbolic program synthesis that is robust to mispredictions of neural models.  

\vspace{-0.05in}
\paragraph{\textbf{Conformal prediction.}} In recent years, conformal prediction ~\cite{angelopoulos2023conformal, shafer2008tutorial} has  become a popular paradigm for uncertainty quantification of black-box learning models, finding use in a variety of learning tasks across different domains \cite{angelopoulos2020uncertainty}. In the traditional setting, one limiting assumption must be made about the test data in order to achieve the desired coverage guarantees -- namely that it is exchangeable with the calibration data. Much work in the area has thus focused on adapting conformal techniques to more general learning settings, such as distribution shift ~\cite{gibbs2021adaptive, tibshirani2019conformal, park2021pac}, time-series ~\cite{chernozhukov2018exact, chendynamictimeseries}, and online settings ~\cite{bastani2022practical, gibbs2022conformal}. 
A very recent (unpublished) work ~\cite{ramalingam2024uncertainty} aims to extend conformal prediction (parameterized by a user-specified confidence bound) from individual neural components to neurosymbolic programs, using a calibration set, a scoring function (for measuring non-conformity), and a user-specified confidence bound. Since the prediction sets obtained this way can be quite imprecise, they also use abstract interpretation to tighten them. Our work builds on this effort in defining conformal semantics for neurosymbolic languages, but allows the user to refine the conformal semantics through  user feedback. Additionally, the main focus of our paper is neurosymbolic active learning, which is not addressed in prior work. 

%Recent work ~\cite{todo} has developed a framework which uses conformal prediction with abstract interpretation to obtain prediction sets with probabilistic coverage guarantees on the outputs of neurosymbolic programs. Our work, building on top of this framework, considers a new setting of interactive program synthesis under conformal semantics and proposes a new efficient technique for refining the program conformal semantics with \emph{additional user feedback}. 

% - this method is agnostic to the specific conformal prediction algorithm used as well as to the DSL. Our work leverages this idea to construct our conformal semantics.  \todo{This description makes us look very incremental. We should say something more about what the contributions of this paper are over that work.}

\vspace{-0.05in}
\paragraph{\textbf{Abstract interpretation in program synthesis.}} Prior work has proposed pruning techniques for program synthesis that leverage abstract interpretation. Many such approaches employ abstract reasoning in only the forward direction ~\cite{feng2017component, singh2011synthesizing, so2017synthesizing, wang2018learning, wang2017synthesizing, mell2024optimal, guria2023absynthe}, or only the backward direction ~\cite{pailoor2021synthesizing}. Other works utilize abstraction refinement to generate a program that matches a set of IO examples under an abstract semantics, and then iteratively refine the semantics if the generated program is spurious ~\cite{wang2017program, wang2018learning, guo2019program, vechev2010abstraction}. Several works leverage both forward and backward reasoning in program synthesis. In particular, ~\cite{mariano2022automated} and ~\cite{mukherjee2020dataflow} use bidirectional abstract reasoning to efficiently prune infeasible partial programs in their respective domains (automated transpilation and LLVM superoptimization). Yoon et al. ~\cite{yoon2023inductive} propose a general framework for inductive synthesis that uses iterative forward-backward abstract interpretation to infer constraints on partial programs.  In contrast to all of these approaches, our method uses forward and backward abstract interpretation to speed up conformal evaluation, which is not addressed in prior work.
% To infer these constraints for a given input-output example, a forward analysis first computes invariants over the output of a partial program and the outputs of its intermediary expressions. A backward analysis then computes preconditions that must be satisfied by the program's holes in order for the a completion of the program to produce the desired output. This forward-backward reasoning repeats until the analyses converge. 

%While prior work uses bidirectional reasoning as a pruning technique, our approach uses bidirectional reasoning as an efficient evaluation strategy. In particular, our approach performs forward and backward abstract interpretation on complete programs in a hypothesis space, rather than partial programs generated during synthesis. In our setting, the goal of bidirectional reasoning is not to identify programs that are inconsistent with a set of input-output examples, but to constrain the prediction sets of outputs under conformal semantics. As a result, our approach improves the efficiency of conformal evaluation even for programs that are consistent with the IO examples.  

\vspace{-0.05in}
\paragraph{\textbf{Program synthesis with noisy data.}} Various works contend with the problem of synthesizing a program from noisy IO examples. In general, neural program synthesizers are more robust to noise than symbolic synthesizers. For instance, RobustFill ~\cite{robustfill} uses an RNN-based approach to generate string transformations programs from IO examples with noise (e.g. output strings with typos). However, purely neural synthesis approaches do not offer correctness guarantees. ~\citet{handa1} proposes a general framework for inductive synthesis over noisy data that finds a program that minimizes a user-defined cost function. This cost function quantifies how often the program differs from the IO examples. Its followup, ~\cite{handa2}, bounds the error rate of this framework with respect to a formalized noise source. Raychev et al.  similarly propose a synthesis technique for noisy data that outputs a ``best fit'' program for a cost function; this approach also bounds the error rate in the case that the amount of noise in the dataset is bounded~\cite{raychev2016learning}. In contrast to these works, \toolname guarantees that the ground truth program is synthesized under the assumption that prediction sets contain the ground truth label for neural components.

\vspace{-0.05in}
\section{Conclusion}

In this paper, we defined the neurosymbolic active learning problem and proposed the first technique for solving it. Our approach uses a new evaluation strategy called \emph{constrained conformal evaluation (CCE)} to account for inaccuracies in neural network predictions.  We have evaluated our approach -- implemented in a new tool called \toolname\ -- on 112 neurosymbolic program synthesis benchmarks across three domains. Our experimental results point to three key findings: First, \toolname is able to identify the desired program for over 98\% of the benchmarks, whereas prior techniques for active learning can only converge to the ground truth for at most 65\% of the benchmarks. Second, the questions identified by \toolname result in effective user interaction, requiring $19\times$ fewer rounds of user interaction compared to a random question selection strategy. Finally, our ablation studies show the importance of our key algorithmic optimizations, both in terms of reducing user interaction time and scaling with respect to prediction set sizes.

\section{Future Work}
\change{While our evaluation focuses on domain-specific languages, \toolname is not inherently limited to this setting. A promising direction for future work is to apply our approach to more expressive, general-purpose languages. The core components of our method (namely, active learning via constrained conformal evaluation) remain applicable so long as the user can provide an abstract interpreter for their target language. Prior work has implemented abstract transformers for a range of expressive languages ~\cite{fromherz2018static, arceri2021analyzing, keidel2019sound}.}

\change{General-purpose languages typically have larger program spaces, and may require different specification formats. One natural approach to scaling to such languages is to leverage large language models (LLMs). LLM-based techniques have recently shown promise in program synthesis tasks by generating candidate programs from natural language prompts~\cite{austin2021program, barnaby2024photoscout}. While LLMs are often noisy and imprecise, they can be used to produce a diverse candidate set that serves as the initial hypothesis space for \toolname. Our method is well-suited to this setting, as the CCE algorithm is designed to resolve uncertainty and eliminate incorrect candidates through user interaction.}

\change{Another avenue for future work is to explore more compact or symbolic hypothesis space representations, such as version-space algebras (VSAs), tree automata, or partial program sketches. While our experiments use enumerative synthesis to generate the initial hypothesis space, our active learning approach is compatible with any hypothesis space representation that (1) guarantees inclusion of the ground-truth program, and (2) supports pruning of inconsistent candidate programs. Alternative hypothesis space representations may enable more efficient synthesis in domains where enumerative search is infeasible.}

\change{More broadly, the problem our method addresses (i.e., ambiguity caused by uncertain neural predictions) applies to a wide range of neurosymbolic systems, not just the ones considered in this paper. While our evaluation focuses on systems where neural and symbolic components are tightly integrated (e.g., neural functions embedded in a DSL), similar challenges arise in more loosely coupled architectures. For example, some recent agentic systems use a large language model to orchestrate a set of symbolic tools or modules, calling different components based on a high-level plan \cite{wu2023autogen, chan2023chateval, wang2024survey}. These systems typically treat neural components as black boxes and lack systematic ways to detect or resolve uncertainty. We believe that extending our conformal evaluation and querying strategy to such agent-based architectures could improve their reliability -- especially in cases where neural predictions might silently lead to incorrect results.}
\section{Data-Availability Statement}

The artifact for this paper is available on Zenodo ~\cite{smartlabel-artifact}.
\begin{acks}
We would like to thank the members of the UTOPIA group, and the anonymous reviewers, for their help and feedback on this paper. This work was conducted in a research group supported by NSF awards CCF-1762299, CCF-1918889, CNS-1908304, CCF-1901376, CNS-2120696, CCF- 2210831, and CCF-2319471, CCF-2422130, CCF-2403211 as well as a DARPA award under agreement HR00112590133. 
\end{acks}

\bibliography{main}

%% Appendix
\pagebreak
\appendix
\section{Proofs}\label{sec:proofs}

\setcounter{theorem}{0}

\begin{lemma}\label{lemma:sufficient-dist}
Let $\prog_1, \prog_2$ be a pair of programs such that $\constrainedsem{\prog_1}(I) \neq \constrainedsem{\prog_2}(I)$ for some $\inp \in \inps$, where $\inps$ denotes the input space. Then, $\prog_1 \not \indist_{\fbf} \prog_2$. 
\end{lemma}

\begin{proof}
    Note that $\textsf{Possible}(\varnothing)$ contains a single formula,  $\textsf{true}$. From the assumption, this implies 
    \begin{align}
    \confsem{\prog_1}_{\assump  \land \textsf{true}}(\inp) \neq \confsem{\prog_2}_{\assump  \land \textsf{true}}(\inp)
    \end{align} 
    Hence, $\prog_1 \not \indist_{\fbf} \prog_2$
\end{proof}

\begin{lemma}
Let $\progs$ be a hypothesis space and let $\prog'$ be a program randomly sampled from $\prog$. For any user feedback $\fbf$, input space $\inps$ and question space $\questions$, then we have:
\[
\forall \prog_1, \prog_2 \in \progs. \ \prog_1 \indist_\fbf \prog_2
\] if and only if  {\sc Distinguish}($\prog', \progs \backslash \prog', \fbf, \inps, \questions$) returns false.
\end{lemma}

\begin{proof}
    $\Rightarrow$ Suppose $\forall \prog_1, \prog_2 \in \progs. \ \prog_1 \indist_\fbf \prog_2$. Then for any input $\inp$ and any set of questions $\questions' \subseteq \questions$, 
    \begin{align}
    \forall \fbf' \in \textsf{Possible}(\questions'). \ \confsem{\prog_1}_{\assump  \land \assump'}(\inp) = \confsem{\prog_2}_{\assump  \land \assump'}(\inp), 
    \end{align}
    and by extension,
    \begin{align}
    \forall \fbf' \in \textsf{Possible}(\questions'). \ \textsc{CCE}(\prog, \inp, \assump  \land \assump') = \textsc{CCE}(\prog_2, \inp, \assump \land \assump'). 
  \end{align}    
    
    Therefore, $\textsc{Distinguish}(\prog', \progs, \fbf, \inp, \questions)$ will return false for any $\prog' \sim \progs$.

    $\Leftarrow$ Suppose {\sc Distinguish}($\prog', \progs \backslash \prog', \fbf, \inps, \questions$) returns false. Towards contradiction, suppose that there exists $\prog_1, \prog_2 \in \progs$ such that $\prog_1 \not\indist_{\fbf} \prog_2$. Then there exists $\fbf' \in \textsf{Possible}(\questions)$ and some input $\inp$ such that $\confsem{\prog_1}_{\assump  \land \assump'}(\inp) = \confsem{\prog_2}_{\assump  \land \assump'}(\inp)$, and by extension $\textsc{CCE}(\prog_1, \inp, \assump  \land \assump') = \textsc{CCE}(\prog_2, \inp, \assump \land \assump')$. Since $\textsc{Distinguish}$ enumerates all possible sets of questions and all possible answers, the procedure must at some point check that 
    $\textsc{CCE}(\prog_1, \inp, \assump  \land \assump') = \textsc{CCE}(\prog', \inp, \assump \land \assump')$ and $\textsc{CCE}(\prog_2, \inp, \assump  \land \assump') = \textsc{CCE}(\prog', \inp, \assump \land \assump')$. 
    However, this implies $\textsc{CCE}(\prog_1, \inp, \assump  \land \assump') = \textsc{CCE}(\prog_2, \inp, \assump \land \assump')$, which is a contradiction. Therefore, $\forall \prog_1, \prog_2 \in \progs. \ \prog_1 \indist_\fbf \prog_2$. 
\end{proof}

\begin{lemma}
Let $\prog_1, \prog_2$ be a pair of programs such that $\prog_1 \indist_\assump \prog_2$. Then, assuming $\assump$ represents accurate user feedback, we have $\forall I \in \inps. \ \gtsem{\prog_1}(I) = \gtsem{\prog_2}(I)$ where $\inps$ denotes the input space.
\end{lemma}

\begin{proof}
    Towards contradiction, suppose there exist programs $\prog_1$, $\prog_2$ and an input $\inp^*$ such that 
    \begin{align*}
      \gtsem{\prog_1}(\inp^*) \neq \gtsem{\prog_2}(\inp^*). 
    \end{align*}

  Consider the set of labels $\fb = \{(f, \inp, \gtsem{f}(\inp)) \ | \ f \in \nns, \inp \in \inps \}$. This set provides the ground truth label for every neural function on every input. Constrained conformal semantics with the formula $\fbf \wedge \fbf_{\fb}$ is thus equivalent to the ground truth semantics. Then, since $\prog_1 \indist \prog_2$, their output under constrained conformal semantics with feedback $\fbf \wedge \fbf_{\fb}$ must be equal on all inputs. That is,
    \begin{align*}
 \gtsem{\prog_1}(\inp^*) = \confsem{\prog_1}_{\assump  \land \fbf_{\fb}}(\inp^*) = \confsem{\prog_2}_{\assump  \land \fbf_{\fb}}(\inp^*) = \gtsem{\prog_2}(\inp^*),
    \end{align*}
    which contradicts the initial assumption. 
    %Therefore, the lemma holds.

\end{proof}

\begin{theorem}
Let $\prog^*$ be the ground truth program, and let $\prog$ be the program returned by the {\sc ActiveLearning} procedure. Then, under the assumption that $\prog^* \in \progs$:
\[
\forall \inp \in \inps. \ \gtsem{\prog}(\inp) = \gtsem{\prog^*}(\inp)
\]
\end{theorem}
\begin{proof}
    Note that $\prog$ is a program selected from a refined hypothesis space $\progs' \subseteq \progs$ such that, for some user feedback $\fbf$ and some $\prog' \in \progs$, {\sc Distinguish}($\prog', \progs \backslash \prog', \fbf, \inps, \questions$) returns false. By Lemma ~\ref{lemma:dist1}, $\prog_1 \indist_{\fbf} \prog_2$ for all pairs $(\prog_1, \prog_2) \in \progs' \times \progs'$. Since the user always gives feedback consistent with the ground truth, $\prog^* \in \progs'$. Then $\prog \indist_{\fbf} \prog^*$. Therefore, by Lemma ~\ref{lemma:indis}, $\gtsem{\prog}(\inp) = \gtsem{\prog^*}(\inp)$ for all $\inp \in \inps$.
\end{proof}

\begin{theorem}\label{thm:qs1}
Let $\prog = (\textsf{def} \ \synthfunc(x) = S)$ be a program and let $\fbf$ denote user feedback. Then, $\prog \Vdash_k \fbf$ implies $\prog \models \fbf$.
\end{theorem}

\begin{proof}

    To prove $\prog \models \fbf_{\fb}$, we must show $\forall (\synthfunc, \inp, \out) \in \fb . \ \out \in \confsem{\prog}_{\fbf_\fb}(\inp)$. Since $\prog \Vdash_k \fbf_\fb$, we know $\forall (\synthfunc, \inp, \out) \in \fb . \ \out \in \confsem{\prog}_{\fbf_\fb}^k(\inp)$. Since BCE simply samples subsets of the prediction sets output by CCE, it must be that   $\confsem{\prog}_{\fbf_\fb}^k(\inp) \subseteq \confsem{\prog}_{\fbf_\fb}(\inp)$, and so $\out \in \confsem{\prog}_{\fbf_\fb}(\inp)$. Therefore, $\prog \models \fbf_{\fb}$

    % $\inp \in \inps$, and  $\fbf \models \synthfunc(\inp) = \out$. Then $[x \mapsto \inp], \fbf \vdash \prog \hookrightarrow \out'$ must have been inferred from the $\textsc{Prog-1}$ rule for constrained non-conformal semantics. By the definition of $\prog \Vdash \fbf$, $\out' \neq \bot$. Thus, it must be the case that $\out' = \out$, and $[x \mapsto \inp],\fbf \vdash S \hookrightarrow \out'$. Further, we may infer $[x \mapsto \inp] \vdash \prog \constrainedeval \confout$ from the $\textsc{Prog}$ rule of conformal constrained semantics, so it must be that $[x \mapsto \inp] \vdash S \constrainedeval \Delta$ and $\confout = \{\out \ | \ (\out, \varphi) \in \Delta\}$. By Lemma ~\ref{lemma:statement}, $(\out, \varphi) \in \Delta$ for some $\varphi$, so $\out \in \confout$. Then $\out \in \confsem{\prog}_{\fbf}(\inp)$. Therefore, $\prog \models \fbf$. 

\end{proof}

\begin{corollary}\label{cor:qs2}
Let $\progs$ be a set of programs and let $\fbf$ denote user feedback. Then, we have: \[
 (\hs{\progs}{\fbf}) \  \supseteq \  \{ \prog \in \progs \ | \ \prog \Vdash \fbf \} 
 \]
\end{corollary}

\begin{proof}
    Let $\prog' \in \{ \prog \in \progs \ | \ \prog \Vdash \fbf \}$. Since $\prog' \Vdash \fbf$, by Theorem ~\ref{thm:qs1}, $\prog' \models \fbf$. Then $\prog' \in \hs{\progs}{\fbf}$, and hence $(\hs{\progs}{\fbf}) \  \supseteq \  \{ \prog \in \progs \ | \ \prog \Vdash \fbf$.
\end{proof}

\begin{corollary}\label{cor:qs3}
Let $q = (f, I)$ be a question and $\progs$ be a program space consistent with $\fbf$. Then:
\begin{align*}
    \pp{q, \progs} \leq  \frac{|\progs| - \underset{i}{\text{max}}  \ | \{ \prog \in \progs \ | \ \prog \Vdash \fbf \land f(I) = a_i \}  |}{|\progs|}
\end{align*}
\end{corollary}

\begin{proof}
    By Corollary ~\ref{cor:qs2}, $\forall \inp \in \inps. \ | \{ \prog \in \progs \ | \ \prog \Vdash \fbf \wedge f(\inp) = \answer_i \}| \leq |(\hs{\progs}{\fbf \wedge f(\inp) = \answer_i} ) |$. Then $\text{max}_i \ | \{ \prog \in \progs \ | \ \prog \Vdash \fbf \wedge f(\inp) = \answer_i \}| \leq \text{max}_i \ |(\hs{\progs}{\fbf \wedge f(\inp) = \answer_i} )|$, and hence
\begin{align*}
    \frac{|\progs| - \underset{i}{\text{max}}  \ | (\hs{\progs}{\fbf \land f(I) = a_i} ) |}{|\progs|}\leq  \frac{|\progs| - \underset{i}{\text{max}}  \ | \{ \prog \in \progs \ | \ \prog \Vdash \fbf \land f(I) = a_i \}  |}{|\progs|}.
\end{align*}
    
\end{proof}

\begin{theorem}
The {\sc SelectQuestion} procedure returns a question $q^* \in \questions$ with the highest pruning power modulo hypothesis space $\progs$. That is,
\[
\forall q \in \questions. \ \pp{q, \progs} \leq \pp{q^*, \progs}
\]
\end{theorem}

\begin{proof}
    Note that {\sc SelectQuestion} returns a question $\question^* \in \questions$ such that, for all questions $\question \in \questions$, either we have computed the pruning power of $\question$ and checked that $\pp{q, \progs} \leq \pp{q^*, \progs}$, or we have checked that the pruning power of $\question^*$ is greater than the approximated pruning power of $\question$:
\begin{align*}
\frac{|\progs| - \underset{i}{\text{max}}  \ | \{ \prog \in \progs \ | \ \prog \Vdash \fbf \land f(I) = a_i \}  |}{|\progs|} \leq \pp{\question^*, \progs}.
\end{align*}
By Corollary ~\ref{cor:qs3}, 
\begin{align*}
         \pp{q, \progs} \leq  \frac{|\progs| - \underset{i}{\text{max}}  \ | \{ \prog \in \progs \ | \ \prog \Vdash \fbf \land f(I) = a_i \}  |}{|\progs|}.
\end{align*}
Therefore, 
\[
\forall q \in \questions. \ \pp{q, \progs} \leq \pp{q^*, \progs}.
\]

\end{proof}

\section{Hyperparameter Selection}
Recall that the bounded conformal evaluation strategy used in question selection samples subsets of the prediction sets output by a program and its intermediary operations. In the \digits domain, the BCE hyper-parameter $k$ was set to 1, and we found that increasing it beyond $k=1$ is not useful. In the \imageedit and \imagesearch domains, on the other hand, we found that proportionally sampling 5\% of the elements in the prediction set gives the best results until a $k$ limit is reached. In particular, if the prediction set contains $n$ elements, we choose $\emph{max}(n/20, k)$ elements, and we use a $k$ value of $10$ in our experiments.
\section{Details About Application Domains}\label{sec:domains2}

%We have implemented the proposed active learning technique as a new tool called \toolname{}. This tool 
In Section ~\ref{sec:eval}, we evaluate \toolname on three application domains: \imageedit, \imagesearch, and \digits. In this section, we provide more details about these domains and their abstractions. 

%Our proposed active learning algorithm may be applied to any neurosymbolic domain for which a suitable abstract domain and corresponding forward and backward transformers are available. In this section, we describe the two  domains used in our evaluation and their respective abstractions. 

\subsection{Intermediate Representation for Image Editing and Search }

While prior work proposes different DSLs targeting image editing versus image search, these domains both use the same neural networks  for image segmentation and classification. Hence, to avoid defining two different abstract domains, we translate both DSLs proposed in prior work to an intermediate representation (IR) shown in  Figure~\ref{fig:imageeyedsl} (henceforth referred to as \objextract)  used for retrieving objects from an image. The \imageedit DSL extracts one or more sets of objects using the DSL in Figure~\ref{fig:imageeyedsl} and applies the specified operation to each set. Meanwhile, the \imagesearch DSL returns a boolean depending on whether or not the IR in Figure~\ref{fig:imageeyedsl} returns the empty set.

\begin{figure}
\begin{minipage}{.39\textwidth}
    \centering
    \footnotesize
    \[
    \begin{array}{r l}
        \prog := & \lambda x. \ \text{let } y = \textsf{\textbf{Segment}}(x) \text{ in } E \\
         \extractor := & y \ | \ {\textsf{Filter}}(E, f) \ | \   {\sf Complement}(\extractor) \\ &| \ {\sf Union}_N(E_1, \cdots, E_N) \\ &| \ {\sf Intersect}_N(E_1, \cdots, E_N) \\ &| \ E \times E \\  
          f := & \imageeyeattr{a} \ | \ \textsf{ HasRelation}(\texttt{r})
    \end{array}
    \]
    \caption{The \objextract DSL. The neural components are in bold.}
    \label{fig:imageeyedsl}
\end{minipage}
\begin{minipage}{.6\textwidth}
    \centering
    \includegraphics[scale=0.18, trim={0 26cm 25cm 0},clip]{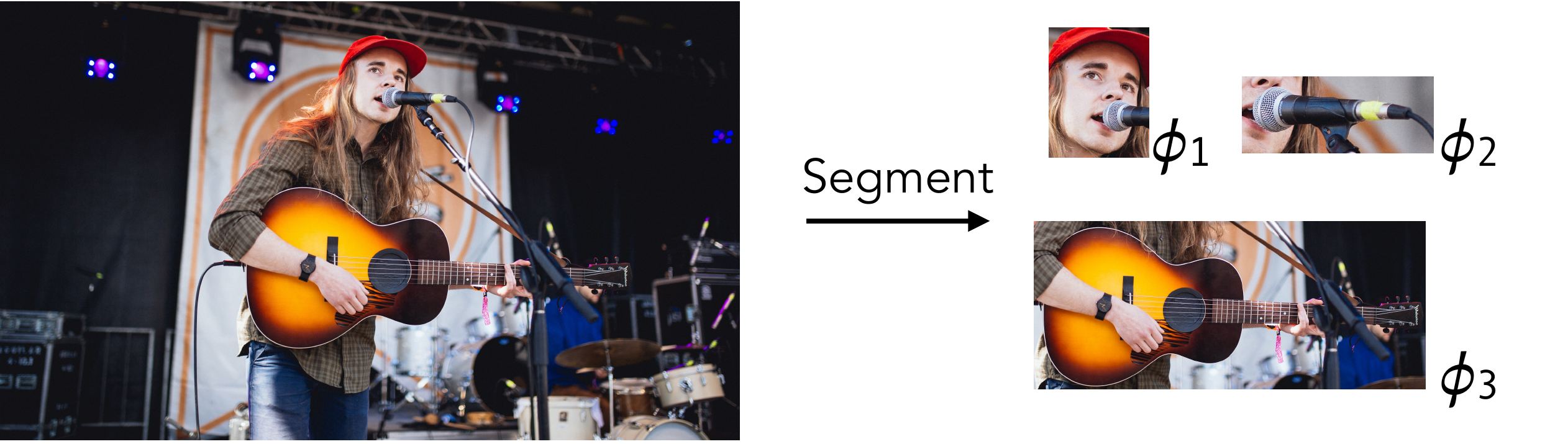}
    \caption{Image segments produced via object detection. }
    \label{fig:segment}
\end{minipage}
\vspace{-.3cm}
\end{figure}

 A program $P$ in this DSL is an object extraction function of the form $\lambda x. \ \text{let } y = \textsf{\textbf{Segment}}(x) \text{ in } E$, where $\textsf{\textbf{Segment}}$ is a neural component for performing \emph{segmentation} on the input image $x$. In particular, the $\textsf{\textbf{Segment}}$ operation outputs a set of \emph{image segments} $o = (\phi, \Delta)$, where $\phi$ denotes an image corresponding to a specific object and $\Delta$ is the location of that object within the original image. Following prior work~\cite{imageeye}, we represent the location of the object as a bounding box ($j_{\tt l}, j_{\tt r}, j_{t}, j_{b}$)  describing the left, right, top, and bottom pixels of the object. For instance,  Figure ~\ref{fig:segment} shows the result of running segmentation on an example image.

As shown in Figure~\ref{fig:imageeyedsl}, expressions $E$ in the \objextract DSL involve a combination of neural and symbolic operations. Specifically, symbolic expressions include standard set operations (complement, union, intersection, and Cartesian product) with well-known semantics. The \textsf{Filter} construct takes as input a set of image segments and a predicate $f$, which can be either neural or symbolic. In particular, $\imageeyeattr{a}$ is a neural component that checks whether a given image segment has a specific attribute \texttt{a}, such as \texttt{face} or \texttt{guitar}, indicating that this object is labeled as a face or guitar. On the other hand, $\textsf{HasRelation}(\texttt{r})$ is a symbolic binary predicate for checking whether a pair of image segments have a certain spatial relation (e.g., \texttt{below}). Since spatial relations can be determined based on  bounding boxes, evaluation of such predicates does not require the use of a neural network. 

\begin{example}
Consider the following expression in the \objextract DSL:
\begin{align*}
\small
    \textsf{Intersection}(\textsf{Filter}(y, \imageeyeattr{(\texttt{face})}),  \textsf{Complement}(\textsf{Filter}(y, \imageeyeattr{(\texttt{smiling})}))) 
\end{align*}
This expression yields all image segments containing human faces that are not smiling. For instance, applying this expression to the segments from Figure~\ref{fig:segment} would yield object $\phi_1$.
%As another example, consider the expression:
%\begin{align*}
%\small  
%    \textsf{Left}(\textsf{Filter}(\textsf{Filter}(y, \imageeyeattr{(\texttt{guitar}))} \times \textsf{Filter}(y, \imageeyeattr{(\texttt{face})}), \textsf{HasRelation}(\texttt{below})))
%\end{align*}
%This returns all guitars that are located below people in an image. The expressions $\textsf{Left}$ and $\textsf{Right}$ simply map a set of tuples to sets of their left and right elements, respectively. \jocelyn{I don't see Left and Right in the DSL. }$\textsf{HasRelation}(\texttt{r})$ is interpreted by comparing the bounding boxes of objects, while $\imageeyeattr{a}$ is interpreted using a neural attribute detection model. \jocelyn{I update the captions of Figure~\ref{fig:imageeyedsl}, so if we don't have much to say here, we can remove the last sentence.}
\end{example}

We conclude this section by  showing how to express a few representative tasks from the image editing and image search domains using our IR:

\begin{example}
    Consider the following program in the \imageedit DSL:
\begin{align*}
\small
         \{ \textsf{Intersect}(&\textbf{\textsf{Is}}(\textsf{\textbf{Object}}( \texttt{face})), \\
         \small &\textsf{Complement}(\textsf{Find}((\textbf{\textsf{Is}}(\textsf{\textbf{Object}}(\texttt{guitar})), \textbf{\textsf{Is}}(\textsf{\textbf{Object}}(\texttt{face})), \textsf{GetAbove})))) \rightarrow \textsf{Blur}  \}.
\end{align*}
This program blurs the faces of all people who are not holding guitars. We may instead express this task by applying a blur action to all objects extracted by the following \textsc{ObjExtract} expression:
\begin{align*}
         \textsf{Intersect}(&\textsf{Filter}(y, \textbf{\textsf{HasAttribute}}(\texttt{face})), \\ 
         &\textsf{Complement}(\textsf{Filter}(\textsf{Filter}(y, \textsf{\textbf{HasAttribute}}(\texttt{guitar})) \times \\ 
         & \ \ \ \ \ \ \ \ \ \ \ \ \ \ \ \ \ \ \ \ \ \ \ \ \ \ \ \ \ \ \ \ \ \ \ \textsf{Filter}(y, \textsf{\textbf{HasAttribute}}(\texttt{face})), \textsf{HasRelation}( \texttt{above})))) .
\end{align*}
\end{example}
\begin{example}
    Consider the following program in the \imagesearch DSL:
\begin{align*}
        \exists x. \exists y. \exists z. &\textbf{\textsf{HasType}}(x, \texttt{bicycle}) \wedge \textbf{\textsf{HasType}}(y, \texttt{person})  
        \wedge \textbf{\textsf{HasType}}(z, \texttt{helmet}) \wedge \\ &\textsf{HasRelation}(x, y, \texttt{below}) 
        \wedge \textsf{HasRelation}(y, z, \texttt{below}).
\end{align*}
This program returns all images in a corpus that contain cyclists wearing helmets. We may instead express this task by applying the following \textsc{ObjExtract} expression to all images in the corpus, and returning those whose extracted object sets are non-empty:
\begin{align*}
&\textsf{Filter}(\textsf{Filter}(\textsf{Filter}(y, \textsf{\textbf{HasAttribute}}(\texttt{bicycle})) \times \\ 
& \ \ \ \ \ \ \ \ \ \ \ \ \ \ \ \ \ \ \ \ \ \textsf{Filter}(y, \textsf{\textbf{HasAttribute}}(\texttt{person})), \textsf{HasRelation}( \texttt{above})) \times \\ 
& \ \ \ \ \ \ \ \ \ \ \textsf{Filter}(y, \textsf{\textbf{HasAttribute}}(\texttt{helmet})), \textsf{HasRelation}(\texttt{above})).
\end{align*}
\end{example}

\subsection{Abstract Domain for \objextract Intermediate Language}

Since programs in the \objextract IR  return a \emph{set} of objects under the regular evaluation semantics, they must return \emph{sets of sets} under the conformal semantics. We thus consider a \emph{set interval} abstract domain where each set $\confout$ of object sets is abstracted using a \emph{set interval} $[\out^-, \out^+]$ satisfying $\forall \out \in \confout. \  \out^- \subseteq \out \subseteq \out^+$. The abstraction and concretization functions for the set interval domain are defined as follows:
\begin{align*}
    \abstfunc(\confout) &= [\{\imgobj \ | \ \forall \out \in \confout. \ \imgobj \in \out \}  , \{ \imgobj \ | \ \exists \out \in \confout. \ \imgobj \in \out \}  ] \\ 
    \concfunc( [\out^-, \out^+]) &= \{ \out \ | \ \out^- \subseteq \out \subseteq \out^+ \}
\end{align*}
In other words, the abstraction function $\abstfunc$ takes a set $\confout$ of object sets and outputs an interval $[\out^-, \out^+]$ where $\out^-$ (resp. $\out^+$) is an under-approximation (resp. over-approximation of every object set $\out \in \confout$.
%contains all objects $\imgobj$ that are contained in \emph{every} object set $\out \in \confout$, while $\out^+$ contains any objects that are contained in \emph{some} object set in $\confout$. 
Conversely, the concretization function takes an abstract value $[\out^-, \out^+]$ and outputs the set of all object sets $\out$ that are a superset of $\out^-$ and a subset of $\out^+$.

\paragraph{\textbf{Forward abstract semantics.}}  Figure ~\ref{fig:imageeyeabs} presents the forward abstract semantics for this domain. The first two rules for  neural components apply the abstraction function to the conformal prediction result. The transformer for \textsf{Union} first computes the abstract values for the arguments and then takes the union of the lower and upper bounds. 
 The transformer for $\textsf{Intersect}$ is similar to that of $\textsf{Union}$, but it takes the intersections of the lower and upper bounds instead of their unions. The rule for $\textsf{Complement}$ first computes the abstract value $[O^-, O^+]$ for its argument $E$, and then computes the new lower and upper bounds  as $U^- \backslash O^+$  and $U^+ \backslash O^-$.
%The  rule for complement first computes the abstract value $[O^-, O^+]$ for its argument $E$ and then computes the new lower and upper bounds as $U^- \backslash O^+$  and $U^+ \backslash O^-$ where $U^-$ and $U^+$ represent under- and over-approximations of the ``universal" set (i.e., the set of all objects in the image).
In the rule for \textsf{Filter}, the under- (resp. over-) approximation for the whole expression includes an object $o$ if (1) it is  in the under- (resp. over-) approximation of the nested expression $E$ and (2) if $o$ must (resp. may) have attribute $a$ according to conformal prediction. The final rule shows the abstract transformer for filtering based on spatial relations between a pair of objects. This rule is similar to the previous one, but uses a symbolic predicate instead of a neural one.
%Due to space constraints, the rules for \textsf{intersect} and \textsf{complement} have been omitted; we refer the interested reader to Appendix ~\ref{sec:appb} for a full description of these rules and additional information about this domain.

 \begin{figure}
     \centering
     \small
\begin{align*}
    \abssem{\textsf{\textbf{Segment}}}(\inp) &= \abstfunc(\constrainedsem{\textsf{\textbf{Segment}}}(\inp)) \\
    \abssem{\imageeyeattr{a}}(\inp) &= \abstfunc(\constrainedsem{\imageeyeattr{a}}(\inp)) \\
    \abssem{\textsf{Union}(\extractor_1, \ldots, \extractor_N)}(\inp) &= \left [ \bigcup_i \out^-_i, \bigcup_i \out^+_i\right ]\text{, where } [ \out^-_i, \out^+_i] = \abssem{E_i}(\inp)  \\
    \abssem{\textsf{Intersect}(\extractor_1, \ldots, \extractor_N)}(\inp) &= \left [ \bigcap_i \out^-_i, \bigcap_i \out^+_i \right ] \text{, where } [\out^-_i, \out^+_i] = \abssem{E_i}(\inp)  \\
    \abssem{\textsf{Complement}(\extractor)}(\inp) &= [ U^-  \setminus \out^+, U^+ \setminus \out^- ] \\
    & \text{where } [\out^-, \out^+] = \abssem{E}(\inp) \text{ and } [U^-, U^+] = \abssem{\textsf{\textbf{Segment}}}(\inp)\\
    \abssem{\textsf{Filter}(\extractor, \imageeyeattr{a})}(\inp) &= [ \{\imgobj \ | \ \imgobj \in \out_1^- \wedge \imgobj \in \out_2^- \} , \{\imgobj \ | \ \imgobj \in \out_1^+ \wedge \imgobj \in \out_2^+ \ \} ]  \\ 
    &\text{ where } [ \out_1^-, \out_1^+] = \abssem{E}(\inp) \text{ and }  [\out_2^-, \out_2^+] = \abssem{\imageeyeattr{a}}(\inp) \\ 
    \abssem{\textsf{Filter}(\extractor_1 \times \extractor_2, \textsf{HasRelation}(\texttt{r}))}(\inp) &= [\{(\imgobj_1, \imgobj_2) \in \out_1^- \times \out_2^- \ | \  \stansem{\textsf{HasRelation}(\texttt{r})}(o_1, o_2) \}, \\
   % & \hspace{2.4cm} \wedge \not\exists \imgobj'. \ (\imgobj_1, \imgobj') \in  \stansem{\textsf{HasRelation}(\texttt{r})}(U^-)\}, \\
    & \hspace{.5cm} \{(\imgobj_1, \imgobj_2) \in \out_1^+ \times \out_2^+ \ | \ 
 \stansem{\textsf{HasRelation}(\texttt{r})}(o_1, o_2)  \}] \\ 
    & \text{ where } [\out^-_i, \out^+_i] = \abssem{E_i}(\inp)
     %\text{ and } [U^-, U^+] = 
% \abssem{\textsf{\textbf{Segment}}}(\inp) 
\end{align*}
\vspace{-0.25in}
     \caption{Abstract forward semantics for \objextract.}
     \label{fig:imageeyeabs}
     \vspace{-0.15in}
 \end{figure}

\paragraph{\textbf{ Backward abstract semantics.}} 
%An inverse abstract transformer $\backsem{E}{i}$ infers an abstract value for the $i$'th child of an expression $E$, given the output abstract value for $E$ as well as the abstract values for its children (computed through forward reasoning). 
Given an operator $f$ whose abstract output is $\alpha_o$ and whose abstract inputs are $\alpha_1, \ldots, \alpha_k$, the inverse abstract transformer $\backsem{f}{i}$ infers a new abstract value $\alpha'_i$ for the $i$'th child of $f$ such that:
\[
 \forall \overline{x},y.  \left ( \left ( \ f(x_1, \ldots, x_n) = y \land 
y \in \gamma(\alpha_o) \land \bigwedge_{j=1}^n x_j \in \gamma(\alpha_j) \right )  \ 
\Longrightarrow \ x_i \in \gamma(\backsem{f}{i}(\alpha_o, \alpha_1, \ldots, \alpha_n) \right )
\]
The idea is  to use the abstract backward semantics to tighten the abstract value of the $i$'th argument of $f$ as $\alpha_i \sqcap \backsem{f}{i}(\alpha_o, \alpha_1, \ldots, \alpha_n)$. 
Figure ~\ref{fig:imageeyebackwards} presents the backward abstract transformers for \objextract. To understand the rule for \textsf{Union}, suppose we have $A = A_1 \cup A_2$. Since $A \supseteq A_1 \supseteq A \backslash A_2$,  we can derive $A \backslash A_2$ as a lower bound for $A_1$ and $A$ as an upper bound. Generalizing this to more than two operands, the union rule computes the upper bound for the $i$'th operand of union  as $O^+$ and the lower bound as  $O^- \backslash \cup_{i \neq j} O_j^+$.  

To understand the rule for \textsf{Filter}, suppose $A = \textsf{Filter}(A_1, A_2)$. Then $A \subseteq A_1 \subseteq A \cup A_2$. Hence, we can use $A$ as a lower bound on $A_1$ and $A \cup A_2$ as an upper bound. Thus, the \textsf{Filter} rule computes the lower bound for the $i$'th operand as $\out^-$ and the upper bound as $\out^+ \cup (U^+ \setminus \out_j^-)$ (with $j \neq i$). The second filter rule for binary predicates is similar but yields the upper bound $U^+ \times U^+$, since such an operation can only be applied to pairs of objects. The rule for $\textsf{Intersect}$ is a generalization of the rule for \textsf{Filter}. The rule for \textsf{Complement} is the same as the abstract transformer in the forward direction, owing to the face that the set complement operation is involutive.

%Again, the rules for \textsf{Intersect} and \textsf{Complement} are left to Appendix ~\ref{sec:appb}.

%Given the goal output $(\out^-, \out^+)$, $\textsf{Union}$'s $i$th child must output any object $\imgobj \in \out^-$ that \emph{cannot} be output by any other child. Further, this child cannot output object that is not in $\out^+$. Similarly, $\textsf{Intersect}$'s $i$th child cannot output any object that must be output be every other child, but cannot be in $\out^+$. This child must output any object that is in $\out^-$. \todo{Explain filter}.

 \begin{figure}
     \centering
     \small
\begin{align*}
    % \backsem{\imageeyepred(\varphi)}{1}((\out^-, \out^+)) &= ((\out^-, \out^+)) \\
    \backsem{\textsf{Union}(\ldots)}{i}([\out^-, \out^+], [\out^-_1, \out^+_1], \ldots, [\out^-_n, \out^+_n]) &= \left [ \out^- \setminus \bigcup_{j\neq i} \out^+_j, \ \out^+ \right ]  \\
    \backsem{\textsf{Intersect}(\ldots)}{i}([ \out^-, \out^+], [\out^-_1, \out^+_1], \ldots, [\out^-_n, \out^+_n]) &= 
     \left [\out^-, \ O^+ \cup  \left ( U^+ \setminus (\bigcap_{j\neq i} \out^-_j )\right ) \right ]  \\
    \backsem{\textsf{Complement}(\extractor)}{1}([\out^-, \out^+], \ldots) &= [U^- \setminus \out^+ , U^+ \setminus \out^-]    \\
    \backsem{\textsf{Filter}(\extractor, \imageeyeattr{a})}{i}( [\out^-, \out^+], [\out^-_1, \out^+_1], [\out^-_2, \out^+_2]) &= 
    [\out^-, \out^+ \cup (U^+ \backslash O_j^-)]  \textrm{ where } j \neq i\\   
  %  (\out^-,  \{\imgobj \ | \ \imgobj \in\out^+ \vee \imgobj \not\in \out^+_j ) \})  \textrm{ where} j \neq i\\
  %  \backsem{\textsf{Filter}(\extractor_1 \times \extractor_2, {\textsf{HasRelation}(\texttt{r}))}}{1}([O^-, O^+]) &= [ \varnothing, U^+ \times U^+ ]
   \backsem{\textsf{Filter}(E, {\textsf{HasRelation}(\texttt{r}))}}{1}([O^-, O^+]) &= [ O^-, U^+ \times U^+ ] \\
      \backsem{E_1 \times E_2}{i}([O^-, O^+]) &= [ O^- \downarrow i, O^+ \downarrow i ]
\end{align*}
\vspace{-0.25in}
     \caption{Abstract backward semantics for \objextract where $[U^-, U^+] = \abssem{\textsf{\textbf{Segment}}}(\inp) $. We use the notation $O \downarrow j$ to denote the set $\{ o_j \ | \ (o_1, \ldots, o_n) \in O \}$.}
     \label{fig:imageeyebackwards}
     \vspace{-0.1in}
 \end{figure}

\subsection{Further Details About the \imageedit and \imagesearch Domains}

In this subsection, we provide further details about some implementation choices for the \imageedit and \imagesearch domains. 

\paragraph{\textbf{Calibration dataset.}}  As discussed in Section~\ref{sec:impl}, conformal prediction requires calibration datasets for which ground truth labels are available and that have a similar data distribution as the input space. Unlike the \digits domain, the \imageedit and \imagesearch domains lacks quality datasets for which ground truth labels for the objects and attributes are available. To deal with the lack of a labeled calibration set, we use the following methodology. For a given image $\inp$, we use the predictions of the underlying neural networks for segmentation and classification as the ground truth labels. We then distort $\inp$ using existing image perturbation techniques \cite{hendrycks2019benchmarking} to produce a new image $\inp'$ and generate predicted labels for $\inp$ by using the classification results on $\inp'$. We use the same strategy to obtain pseudo-ground truth labels for both the calibration and test set. Using pseudo-labels for both training and testing data is an established methodology in Computer Vision ~\cite{kendall2015posenet, valentin2016learning, sohn2020fixmatch, liu2021unbiased}.
%We use as predicted labels for $\inp$ the predictions made by ImageEye’s neural nets on $\inp$’. This strategy gives us predicted labels that are less accurate than the ground truth labels.

\paragraph{\textbf{Object detection and classification.}} The \imageedit and \imagesearch DSLs contain operators that detect and label object in an image (e.g. the $\textsf{\textbf{Is}}$ operator in the \imageedit DSL and the $\textsf{\textbf{HasAttribute}}$ operator in the \imagesearch DSL). As in ~\cite{imageeye}, we implement these operators using Amazon Rekognition. When a new batch of images is loaded, we perform a preprocessing step that computes the prediction sets of each neural attribute of each image, then memoizes these results to be used during active learning. Preprocessing takes an average of 1.6s per image. 

\subsection{Abstract Domain for Visual Arithmetic Application}

For visual arithmetic tasks, we consider the DSL shown in Figure ~\ref{fig:mnistdsl}. Programs  in this DSL take as input a list of images of handwritten digits, transform those images into integers using the neural function $\texttt{\textbf{toDigit}}$, and then perform list operations using the higher-order functions $\texttt{fold}$, $\texttt{map}$, and $\texttt{filter}$. A $\texttt{fold}$ operation can use the binary functions $\texttt{sum}$, $\texttt{max}$, $\texttt{product}$, and $\texttt{inc}$. A $\texttt{map}$ operation can use curried versions of any of these functions, as well as the neural function $\texttt{\textbf{toDigit}}$ that predicts the label of an image. A \texttt{filter} operation takes in a boolean predicate of the form $\lambda x. x \lhd c $ 
where $\lhd \in \{ \leq, =, \geq, \ldots \} $.

% Isil: this paragraph seems completely redundant
%Under conformal semantics, the neural function $\texttt{\textbf{toDigit}}$ outputs a set of integers containing the ground truth label of the input image. These prediction sets are then propagated throughout a program. As discussed in Section \ref{sec:top-level}, computing the output set of a program is expensive, especially when the prediction sets are large. To make conformal semantics more efficient, we use an evaluation technique that leverages forward and backward abstract interpretation. Hence, we need an abstract semantics for the MNIST domain.

 \paragraph{\textbf{Abstract domain.}} Our abstraction represents a  list of integers  as a list of tuples $(\alpha_I^i, z^i)$ where $\alpha_I^i$ is the interval abstraction of an integer and $z^i \in \{\textsf{true}, \textsf{false}, *\}$ is a three-valued logic element that denotes whether or not the corresponding element is \emph{definitely} in the list. Intuitively, if $z^i$ is $*$, then this element may or may not have been filtered from the list. If $z^i$ is $\textsf{false}$, then the element has definitely been removed, and if $z^i$ is $\textsf{true}$ then it is definitely still in the list.

\begin{figure}
     \centering
     \small
\begin{align*}
    \abssem{\mnistpred}(x) =& \ \abstfunc(\constrainedsem{\mnistpred}(x)) \\
    \abssem{f}(([a,b], z), ([a', b'], z')) =& \ ([\stansem{f}(a, a'), \stansem{f}(b, b')], z \wedge z') \\ 
    \abssem{\lambda x.x < c}(([a,b], z)) =& \begin{cases}
        ([a,b], z)& \text{ if } b < c \\ 
        ([a,b], \textsf{false})& \text{ if } a \geq c \\ 
        ([a,b], z \wedge *)& \text{ otherwise}
    \end{cases} \\
    % \abssem{\lambda x.x > c}(([a,b], z)) =& \begin{cases}
        % ([a,b], z)& \text{ if } b \leq c \\ 
        % (\bot, \textsf{false})& \text{ if } a > c \\ 
        % ([a,b], \textsf{false})& \text{ otherwise}
    % \end{cases} \\
    \abssem{\texttt{filter} \ h \ E} =&  \ \textsf{filter} \ \abssem{h} \ \abssem{E}
    \\
    \abssem{\texttt{map} \ g \ E } =&  \ \textsf{map} \ \abssem{g} \ \abssem{E} 
    \\
    \abssem{\texttt{fold} \ f \ c \ E} =& \  ([a_1 , b_2], \textsf{true}) \\ 
     \text{where } & \ ([a_1, b_1], z_1) = \textsf{fold} \ \abssem{f} \ \alpha(c)  \ (\textsf{filter} \ (\lambda(x, y). y = \textsf{true}) \ \abssem{E}) \ \\
     \text{ and } & \ ([a_2, b_2], z_2) = \textsf{fold} \ \abssem{f} \ \alpha(c) \ (\textsf{filter} \ (\lambda (x,y). \ y \neq \textsf{false}) \  \abssem{E}) 
\end{align*}
     \caption{Abstract forward semantics for the \digits DSL.}
     \label{fig:mnistabs}
 \end{figure}

 \paragraph{\textbf{Forward abstract semantics.}} Figure ~\ref{fig:mnistabs} presents the abstract semantics for this DSL.  In the following discussion, we explain each of the forward abstract transformers in more detail.
 
 \paragraph{\textbf{Binary operation}} A binary operation $f$ (e.g. addition) may be performed on two abstract values $([a,b], z)$, $([a',b'], z')$ by performing the operation on both the lower and upper bounds of the interval and taking the logical conjunction of $z$ and $z'$ under the 3-valued logic semantics~\cite{tvl}. 
 
 \paragraph{\textbf{Filter and map}} A filtering function $h$ may be performed on an abstract value $([a,b], z)$ by considering where \emph{any}, \emph{all}, or \emph{no} concrete values in $\concfunc(([a,b], z))$ would be filtered by $h$. For instance, consider the filtering function $\lambda x. x < c$. In the case that $a \geq c$, we know that no concrete values in $\concfunc(([a,b], z))$ are filtered, and so $h$ outputs the unchanged input $([a,b], z)$. In the case that $b < c$, we know that every concrete value in $\concfunc(([a,b], z))$
are filtered, and so $h$ outputs $([a,b], \textsf{false})$. In any other case, we do not know whether \emph{some} concrete value is filtered, so we output $([a,b], *)$ to account for this uncertainty. To evaluate $\texttt{filter} \ h \ E$ abstractly, we simply filter the abstract list $\abssem{E}$ using $\abssem{h}$.  
The evaluation of $\texttt{map} \ g \ E$ is similar to the evaluation of $\texttt{filter} \ h \ E$.
%where we map the abstract $\abssem{E}$ using $\abssem{g}$. 

\paragraph{\textbf{Fold operation}} Evaluating $\texttt{fold} \ c \ E$ abstractly is slightly more complex, as it requires computing an interval $[a', b']$ representing the minimum and maximum output of the $\texttt{fold}$ operation. Our abstract semantics relies on the fact that all list elements are natural numbers and that, for all functions $f$ in our DSL, we have 
$
\texttt{fold} \ f \ c \ x:xs \geq \texttt{fold} \ f \ c \ xs 
$.
Thus, we can compute the upper bound of the result as: 
\[
\abssem{f} \ \alpha(c)  \ \textsf{filter} \ (\lambda(x, y). y \neq \textsf{false}) \ \abssem{E}) 
\]

However, to compute a true lower bound, we should only consider the list elements that \emph{must} be in  the input list; hence, we compute the lower bound as 
\[
\abssem{f} \ \alpha(c)  \ \textsf{filter} (\lambda(x, y). y = \textsf{true}) \abssem{E}
\]

\paragraph{\textbf{Remark.}} Note that our forward abstract semantics preserve the length of an input list. In particular, if some operation could end up  removing an element from the list, the forward semantics retains that element in the list but sets it corresponding boolean to either $\mathsf{false}$ or $*$. Our backward semantics rely on this length-preservation property of the abstract domain. 

%To compute $b'$, we can perform a standard fold operation on the abstract list $\abssem{E}$ using $\abssem{c}$ and $\abssem{f}$, and take the upper bound of the output interval. To compute $a'$, we can perform a standard fold operation with $\abssem{c}$ and $\abssem{f}$ on $\textsf{Definite}(\abssem{E})$. Here $\textsf{Definite}$ is a function that filters out abstract values $([a,b],z)$ where $z$ is false. We may then take the lower found of the output interval.   

\paragraph{\textbf{Backward abstract semantics.}} 
%The constrained conformal evaluation strategy described in Section~\ref{sec:confeval} uses the forward abstract semantics to label every node in a program's AST with a corresponding abstract value for a given input, then uses the backward abstract semantics to refine these abstract values with respect to a ground truth output. 

%$\backsem{E}{i}(([a,b],z), ([a_1, b_1], z_1), \ldots, ([a_N, b_N], z_N))$ infers the abstract value that $E$'s $i$th child \emph{must} output in order for $E$ to output $([a,b],z)$, with the additional constraint that $E$'s children output $([a_1, b_1], z_1), \ldots ([a_N, b_N], z_N))$. We may consider $([a,b],z)$ the \emph{goal output} of $E$.

The backward abstract semantics for the \digits domain are given in Figure ~\ref{fig:mnistbackwards}. We explain each of the non-trivial rules in more detail below.

 \begin{figure}
     \centering
     \small
\begin{align*}
    \backsem{\mnistpred}{1}(([a,b], z)) &= ([a,b], z) \\
    % \abssem{f}(([a,b], z), ([a', b'], z')) &= ([\stansem{f}(a, a'), \stansem{f}(b, b')], z \wedge z') \\ 
    % \abssem{\lambda x.x < c}(([a,b], z)) &= \begin{cases}
        % ([a,b], z)& \text{ if } a \geq c \\ 
        % (\bot, \textsf{false})& \text{ if } b < c \\ 
        % ([a,b], \textsf{false})& \text{ otherwise}
    % \end{cases} \\
    % \abssem{\lambda x.x > c}(([a,b], z)) &= \begin{cases}
        % ([a,b], z)& \text{ if } b \leq c \\ 
        % (\bot, \textsf{false})& \text{ if } a > c \\ 
        % ([a,b], \textsf{false})& \text{ otherwise}
    % \end{cases} \\
    \backsem{\lambda x.x < c}{1}(([a,b],z)) &= \begin{cases}
        ([a, c - 1], z)& \text{ if } z = \textsf{true}\\ 
        ([a, b], z)& \text{ otherwise}
    \end{cases} \\ 
    % \backsem{\lambda x.x > c}{1}(([a,b],z)) &= \begin{cases}
    %     ([a, c], z)& \text{ if } z \\ 
    %     ([a, b], z)& \text{ otherwise}
    % \end{cases} \\ 
    \backsem{\texttt{filter} \ h \ E}{2}(l, \ldots) &= \textsf{map} \ \backsem{h}{1} \ l
    \\
    \backsem{\texttt{map} \ g \ E  }{2}(l, \ldots) &= \textsf{map} \ \backsem{g}{1}   \ l 
    \\
    \backsem{\texttt{fold} \ \texttt{sum} \ c \ E \ }{3}([a, b], [([a_1, b_1], z_1), \ldots, ([a_n, b_n], z_n)]) &= [([a_1', b_1'], z_1), \ldots, ([a_n', b_n'], z_n)] \\ 
    \text{where } [a_i', b_i'] &= [\textsf{max}(0, a - c- \sum_{j \neq i} b_j ), b - c- \sum_{j \neq i} a_j \cdot \mathbbm{1}(z_j) ]
\end{align*}
     \caption{Abstract backward semantics for the \digits DSL.}
     \label{fig:mnistbackwards}
 \end{figure}

\paragraph{Predicates.} Consider a predicate $\lambda x. x<c$ where $x$ is a list element. If $x$ is known to be \emph{definitely} in the list after applying this predicate to $x$, then $c$ constitutes a lower bound on $x$, as it was definitely \emph{not} removed from the list by applying this predicate.  Thus, as shown in Figure~\ref{fig:mnistbackwards}, we perform a case split on $z$, tightening the bound to $[c, b]$ when $z$ is true and leaving it as $[a, b]$ otherwise.

\paragraph{Filter and map.} The definition of $\backsem{\tt filter}{2}$  relies on the  backward semantics for predicates. For each element $x$ in $l$, the backward semantics computes  the input value of $x$ as $\backsem{h}{1}(x)$. Note that the correctness of the backward semantics relies on the fact that our list abstract domain is length-preserving, as mentioned earlier. The definition of $\backsem{\tt map}{2}$ is similar and relies on the abstract backward semantics $\backsem{g}{1}$. For example, $\backsem{{\tt curry} \ {\tt sum} \  c}{1}$ is defined as  $\lambda x. x-c$. Thus, $\backsem{{\tt map}{\tt{\ curry \ sum} \  c \ L }}{2}$ ends up subtracting $c$ from each element of the output list to compute the abstract value for input list $L$.

%iven a goal output list $l$, we evaluate $\texttt{map} \ g \ E$ by mapping $l$ with the inverse function $\backsem{g}{1}$. Intuitively, if the output of adding $2$ to every element in a list $l$ contains $[a,b]$, then $l$ must contain $[a-2, b-2]$. We similarly evaluate $\texttt{filter} \ h \ E$ with the inverse function $\backsem{h}{1}$. The backwards semantics of a comparison function $h$ refine the intervals of items that were not filtered. If every item less than $5$ was filtered from a list, and $[a,b]$ was definitely not filtered from the list, then $a$ must be at least $5$. 

\paragraph{Fold.}
The final rule in Figure~\ref{fig:mnistbackwards} defines the backward abstract semantics of $\texttt{fold}$. However, since precisely reasoning about {\tt fold} requires performing a case split on the function argument $f$, we only show the case where $f = {\tt sum}$ as a representative example. This rule shows how to compute the updated abstract value for the input list given that the output of the {\tt fold} operation is $[a, b]$ and the initial abstract value of the input list is $[ ([a_1, b_1], z_1), \ldots ([a_n, b_n], z_n)$. To understand this rule, suppose that the actual input list of fold has elements $[x_1, \ldots, x_n]$ and suppose that the actual output is $y$. Then, we have
\[
x_i = y - c - \sum_{j \neq i} x_j
\]
Thus, the lower bound $a_i'$ and upper bound $b_i'$ on $x_i$ can be computed as:
\[
\begin{array}{cc}
a_i' = a - c - \sum_{j \neq i} b_j & b_i' = b - c - \sum_{j \neq i} a_j \times \mathbbm{1}(z_i)
\end{array}
\]
Note that, for the upper bound computation, we define $\mathbbm{1}(z_i)$ to be $1$ if $z_i$ is true and $0$ otherwise, and we only subtract $a_i$ if $\mathbbm{1}(z_i)$ is $1$. This is needed to ensure that the value we compute is a true upper bound.

%reason about the value that every item output by E must have in order for the $\texttt{fold}$ operation to output an interval $[a,b]$. \todo{Finish explaining}. 

 % Due to space constraints, we do not discuss the forward and abstract transformers for this domain in the main paper and refer the interested reader to Appendix ~\ref{sec:appb}.
 %, which also includes additional details about the neurosymbolic image editing domain discussed in the previous subsection. 

%  \begin{figure}
%     \centering
%     \small
%     \begin{align*}
%         \prog := & \ \lambda \ l . \ E \ \ \ \ \ \ E := l \ | \ c \in \mathbb{N} \ | \ \texttt{fold} \ f \ c \ E  \ | \ \texttt{map} \ g \ E \ | \ \texttt{filter} \ h \ E \\ 
%         f :=&  \ \texttt{sum} \ | \ \texttt{max} \ | \ \texttt{product} \ | \ \texttt{inc} \ \ \ \ \ \ g := \texttt{curry} \ f \ c \ | \ \texttt{\textbf{toDigit}} \ \ \ \ \ \ h := \lambda x . x \lhd c 
%     \end{align*}
%     \vspace{-0.2in}
%     \caption{List processing DSL. }
%     \label{fig:mnistdsl}
%     \vspace{-0.2in}
% \end{figure}

\section{List of \imageedit Benchmarks}

\begin{enumerate}
    \item \{\textsf{Intersect}(\textsf{Is}(\textsf{Smiling}), \textsf{Is}(\textsf{EyesOpen})) $\rightarrow$\textsf{Brighten}\} \\
    Description: Brighten all faces that are smiling and have eyes open. \\ 
    Dataset: Wedding
    \item \{\textsf{Find}(\textsf{Is}(\textsf{Object}(\texttt{face}), \textsf{Is}(\textsf{Object}(\texttt{face})), \textsf{GetAbove}) $\rightarrow$ \textsf{Brighten}\} \\ 
    Description: Brighten all faces in back. \\ 
    Dataset: Wedding
    \item \{\textsf{Union}(\textsf{Is}(\textsf{Object}(\texttt{bride})), \textsf{Is}(\textsf{Object}(\texttt{groom})) $\rightarrow$ \textsf{Crop}\} \\ 
    Description: Crop image to feature just faces of bride and groom. \\ 
    Dataset: Wedding 
    \item \{\textsf{Intersect}(\textsf{Is}(\textsf{Object}(\texttt{face})), \newline \textsf{Complement}(\textsf{Is}(\textsf{Object}(\texttt{bride})) $\rightarrow$ \textsf{Blur}\} \\ 
    Description: Blur all faces except the bride's face. \\ 
    Dataset: Wedding 
    \item \{\textsf{Find}(\textsf{Find}(\textsf{Is}(\textsf{Object}(\texttt{face})), \textsf{Is}(\textsf{Object}(\texttt{face})), \textsf{GetRight}), 
\textsf{Is}(\textsf{Object}(\texttt{face}))), \textsf{GetRight}) $\rightarrow$ \textsf{Brighten}\} \\ 
Description: Brighten all faces except the leftmost two faces. \\ 
Dataset: Wedding
    \item \{\textsf{Intersect}(\textsf{Is}(\textsf{Object} (\texttt{face})), \textsf{Complement}(\textsf{Intersect}(\textsf{Is}(\textsf{Smiling}), \textsf{Is}(\textsf{EyesOpen})))) $\rightarrow$ \textsf{Blur} \} \\ 
    Description: Blur all faces that are not smiling and do not have their eyes open. \\ 
    Dataset: Wedding 
    \item \{\textsf{Intersect}(\textsf{Is}(\textsf{Smiling}), \textsf{Is}(\textsf{EyesOpen}), \textsf{Complement}(\textsf{Is}(\textsf{Object}( \texttt{groom}))))$\rightarrow$ \textsf{Blur}\} \\ 
    Description: Crop image to feature all faces that are smiling and have eyes open, except the groom's face. \\ 
    Dataset: Wedding
    \item \{\textsf{Union}(\textsf{Is}(\textsf{Object}(\texttt{bride})), \textsf{Intersect}(\textsf{Is}(\textsf{Smiling}), \textsf{Is}(\textsf{EyesOpen}))) $\rightarrow$ \textsf{Blur}\} \\ 
    Description: Crop image to feature the bride's face, plus faces that are smiling and have their eyes open. \\ 
    Dataset: Wedding
    \item \{\textsf{Intersect}(\textsf{Complement}(\textsf{Is}(\textsf{Smiling})), \textsf{Find}(\textsf{Is}(\textsf{Object}(\texttt{face})),
(\textsf{Is}(\textsf{Object}(\texttt{face})), \textsf{GetAbove})) $\rightarrow$ \textsf{Blur}\} \\ 
Description: Blur all faces in the back that are not smiling. \\ 
Dataset: Wedding 
\item \{\textsf{Union}(\textsf{Intersect}(\textsf{Is}(\textsf{Object}(\texttt{face})), \textsf{Complement}(\textsf{Is}(\textsf{Smiling}))), \textsf{Is}(\textsf{BelowAge}(18))) $\rightarrow$ \textsf{Blur}\} \\ 
Description: Blur all faces that are not smiling or are under 18. \\ 
Dataset: Wedding
\item \{\textsf{Union}(\textsf{Find}(\textsf{Is}(\textsf{Object}(\texttt{bride})), \textsf{Is}(\textsf{Object}(\texttt{face})), \textsf{GetRight}), \textsf{Is}(\textsf{Object}(\texttt{bride}))) $\rightarrow$ \textsf{Crop}\} \\ 
Crop image to feature just the bride's face and the face directly to her right. \\ 
Dataset: Wedding
\item \{\textsf{Union}(\textsf{Is}(\textsf{Object}(\texttt{bride})), \textsf{Find}(\textsf{Is}(\textsf{Object}(\texttt{bride})), \textsf{Is}(\textsf{Object}(\texttt{groom})), \textsf{GetAbove})) $\rightarrow$ \textsf{Crop}\} \\ 
Description: Crop image to feature just the bride and the groom when he is behind her.   \\
Dataset: Wedding
\item \{\textsf{Intersect}(\textsf{Find}(\textsf{Is}(\textsf{Object}(\texttt{face})), \textsf{Is}(\textsf{Object}(\texttt{face})), \textsf{GetRight}), \newline \textsf{Find}(\textsf{Is}(\textsf{Object}(\texttt{face})), \textsf{Is}(\textsf{Object}(\texttt{face})), \textsf{GetLeft})) $\rightarrow$ \textsf{Brighten}\} \\ 
Description: Brighten all faces except leftmost and rightmost face. \\ 
Dataset: Wedding 
\item \{\textsf{Find}(\textsf{Union}(\textsf{Is}(\textsf{Object}(\texttt{groom})), \textsf{Is}(\textsf{Smiling}), \textsf{Is}(\textsf{EyesOpen})),  \textsf{Is}(\textsf{Object}(\texttt{person})), \textsf{GetBelow}) $\rightarrow$ \textsf{Sharpen}\} \\ 
Description: Sharpen the groom, and all smiling people and people with their eyes open. \\
Dataset: Wedding 
\item \{\textsf{Intersect}(\textsf{Find}(\textsf{Is}(\textsf{Object}(\texttt{face})), \textsf{Is}(\textsf{Object}(\texttt{bride})), \textsf{GetRight}), \newline  \textsf{Find}(\textsf{Is}(\textsf{Object}(\texttt{face})), \textsf{Is}(\textsf{Object}(\texttt{bride})), \newline \textsf{GetLeft})) $\rightarrow$ \textsf{Crop}\} \\ 
Description: Crop image to feature just bride when someone is to her left and right. \\ 
Dataset: Wedding 
\item \{\textsf{Union}(\textsf{Find}(\textsf{Is}(\textsf{Object}(\texttt{bride})), \textsf{Is}(\textsf{Object}(\texttt{face})), \textsf{GetRight}), \textsf{Find}(\textsf{Is}(\textsf{Is}(\textsf{Object}(\texttt{bride}))), \newline \textsf{FaceObject}, \textsf{GetLeft}), \textsf{Is}(\textsf{Object}(\texttt{bride}))) \newline $\rightarrow$ \textsf{Crop}\} \\ 
Description: Crop image to feature just the bride and the people to her left and right. \\ 
Dataset: Wedding
\item \{\textsf{Complement}(\textsf{Is}(\textsf{Object}(\texttt{car}))) $\rightarrow$ \textsf{Blur}\} \\ 
Description: Blur all objects except cars. \\ 
Dataset: City Streets 
\item \{\textsf{Filter}(\textsf{Is}(\textsf{Object}(\texttt{car})), \textsf{Is}(\textsf{Object}(\texttt{face}))) $\rightarrow$ \textsf{Blur}\} \\
Description: Blur all faces in cars. \\ 
Dataset: City Streets
\item \{\textsf{Filter}(\textsf{Is}(\textsf{Object}(\texttt{car})), \textsf{Is}(\textsf{Object}(\texttt{text})) $\rightarrow$ \textsf{Blur}\} \\
Description: Blur all text on cars. \\ 
Dataset: City Streets
\item \{\textsf{Find}(\textsf{Is}(\textsf{Object}(\texttt{text})), \textsf{Is}(\textsf{Object}(\texttt{car})), \textsf{GetParents}) $\rightarrow$ \textsf{Blur}\} \\ 
Description: Blur all cars with text on them. \\
Dataset: City Streets
\item \{\textsf{Union}(\textsf{Is}(\textsf{Object}(\texttt{cat})), \textsf{Is}(\textsf{Object}(\texttt{face})) ) $\rightarrow$  \textsf{Brighten}\} \\ 
Description: Brighten all faces and all cats. \\ 
Dataset: Cats 
\item \{\textsf{Union}(\textsf{Is}(\textsf{Object}(\texttt{cat})), \textsf{Is}(\textsf{EyesOpen})) $\rightarrow$ \textsf{Brighten}\} \\ 
Description: Brighten all faces with eyes open and all cats. \\ 
Dataset: Cats 
\item \{\textsf{Find}(\textsf{Is}(\textsf{Object}(\texttt{guitar})), \textsf{Is}(\textsf{Object}(\texttt{face})), \textsf{GetAbove}) $\rightarrow$ \textsf{Sharpen}\} \\ 
Description: Sharpen faces of people playing guitar. \\ 
Dataset: Festival 
\item \{\textsf{Find}(\textsf{Intersect}(\textsf{Is}(\textsf{Smiling}), \textsf{Is}(\textsf{EyesOpen})), \textsf{Object}(\texttt{car}), \newline \textsf{GetParents}) $\rightarrow$ \textsf{Brighten}\} \\ 
Description: Brighten all people in cars who are smiling and have eyes open. \\
Dataset: City Streets
\item \{\textsf{Union}(\textsf{Is}(\textsf{Object}(\texttt{car})), \textsf{Is}(\textsf{Object}(\texttt{bicycle}))) $\rightarrow$ \textsf{Brighten}\} \\ 
Description: Brighten all cars and bicycles. \\
Dataset: City Streets 
\item \{\textsf{Find}(\textsf{Is}(\textsf{Object}(\texttt{person})), \textsf{Object}(\texttt{bicycle}), \textsf{GetBelow}) $\rightarrow$ \textsf{Brighten}\} \\ 
Description: Brighten all bicycles that are being ridden. \\ 
Dataset: City Streets 
\item \{\textsf{Find}(\textsf{Is}(\textsf{Object}(\texttt{bicycle})), \textsf{BelowAge}(\texttt{18}), \textsf{GetAbove}) $\rightarrow$ \textsf{Blur}\} \\ 
Description: Blur the faces of children riding bicycles. \\ 
Dataset: City Streets 
\item \{\textsf{Complement}(\textsf{Union}(\textsf{Is}(\textsf{Object}(\texttt{car})), \textsf{Is}(\textsf{Object}(\texttt{bicycle})))) $\rightarrow$ \textsf{Blackout}\} \\ 
Description: Blackout all objects except cars and bicycles. \\ 
Dataset: City Streets
\item \{\textsf{Intersect}(\textsf{Is}(\textsf{Object}(\texttt{text})),  \textsf{Complement}(\textsf{Filter}(\textsf{Is}(\textsf{Object}(\texttt{car})),  \textsf{Is}(\textsf{Object}(\texttt{car}))))) $\rightarrow$ \textsf{Blackout}\} \\ 
Description: Blackout all text not on a car. \\
Dataset: City Streets
\item \{\textsf{Union}(\textsf{Is}(\textsf{Object}(\texttt{bicycle})), \textsf{Is}(\textsf{Object}(\texttt{car})), \textsf{Is}(\textsf{Object}(\texttt{person}))) $\rightarrow$ \textsf{Brighten}\} \\ 
Description: Brighten all bicycles, cars, and people. \\
Dataset: City Streets
\item \{\textsf{Intersect}(\textsf{Is}(\textsf{Object}(\texttt{face})), \newline  \textsf{Complement}(\textsf{Find}(\textsf{Is}(\textsf{Object}(\texttt{bicycle})), \textsf{Is}(\textsf{Object}(\texttt{face})), \textsf{GetAbove}))) $\rightarrow$ \textsf{Blur}\} \\ 
Description: Blur faces of people not riding bicycles. \\ 
Dataset: City Streets
\item \{\textsf{Union}(\textsf{Is}(\textsf{Object}(\texttt{guitar})), \textsf{Find}(\textsf{Is}(\textsf{Object}(\texttt{guitar})), \textsf{Is}(\textsf{Object}(\texttt{face})), \textsf{GetAbove})) $\rightarrow \textsf{Brighten}$\} \\ 
Description: Brighten all guitars and people playing guitar. \\ 
Dataset: Festival 
\item \{\textsf{Intersect}(\textsf{Is}(\textsf{Object}(\texttt{face})),  \textsf{Complement}(\textsf{Find}(\textsf{Is}(\textsf{Object}(\texttt{guitar})), \textsf{Is}(\textsf{Object}(\texttt{face})), \textsf{GetAbove}))) $\rightarrow$ \textsf{Blur}\} \\ 
Description: Blur faces of people not playing guitar. \\ 
Dataset: Festival 
\item \{\textsf{Intersect}(\textsf{Is}(\textsf{Object}(\texttt{bicycle})), \newline  \textsf{Complement}(\textsf{Find}(\textsf{Is}(\textsf{Object}(\texttt{person})), \textsf{Object}(\texttt{bicycle}), \textsf{GetBelow}))) $\rightarrow$ \textsf{Sharpen}\} \\ 
Description: Sharpen bicycles that are not being ridden. \\ 
Dataset: City Streets 
\item \{\textsf{Intersect}(\textsf{Is}(\textsf{Object}(\texttt{bicycle})), \textsf{Complement}(\textsf{Find}(\textsf{Is}(\textsf{BelowAge}(\texttt{18})), \textsf{Is}(\textsf{Object}(\texttt{bicycle})), \textsf{GetBelow}))) $\rightarrow$ \textsf{Sharpen}\} \\ 
Description: Sharpen all bicycles that are not ridden by a child. \\ 
Dataset: City Streets
\item  \{\textsf{Intersect}(\textsf{Is}(\textsf{Object}(\texttt{cat})), 
 \newline \textsf{Complement}(\textsf{Find}(\textsf{Is}(\textsf{Object}(\texttt{cat})), \newline  \textsf{Object}(\texttt{cat}), \textsf{GetBelow}))) $\rightarrow$ \textsf{Crop} \} \\ 
 Description: Crop image to feature just topmost cat. \\ 
 Dataset: Cats 
 \item \{\textsf{Intersect}(\textsf{Find}(\textsf{Is}(\textsf{Object}(\texttt{cat})), \newline \textsf{Object}(\texttt{cat}), \textsf{GetRight}), \newline \textsf{Find}(\textsf{Is}(\textsf{Object}(\texttt{cat})), \textsf{Object}(\texttt{cat}), \newline \textsf{GetLeft})) $\rightarrow$ \textsf{Brighten}\} \\ 
 Description: Brighten cats that are between two other cats. \\ 
 Dataset: Cats
\end{enumerate}

\section{List of \imagesearch Benchmarks}
\begin{enumerate}
    \item $\exists x. \exists y. \textsf{HasAttribute}(x, \texttt{car}) \wedge \textsf{HasAttribute}(y, \texttt{bicycle})$ \\ 
    Description: Images that contain a car and a bicycle. \\ 
    Dataset: City Streets
    \item $\exists x. \exists y. \textsf{HasAttribute}(x, \texttt{guitar}) \wedge \textsf{HasAttribute}(y, \texttt{microphone})$ \\ 
    Description: Images that contain a guitar and a microphone. \\ 
    Dataset: Festival
    \item $\forall x. \neg \textsf{HasAttribute}(x, \texttt{person})$ \\ 
    Description: Images that do not contain any people. \\ 
    Dataset: Wedding
    \item $\exists x. \exists y. \textsf{HasAttribute}(x, \texttt{bride}) \wedge \textsf{HasAttribute}(y, \texttt{groom}) \wedge \textsf{HasRelation}(x, y, \texttt{left})$ \\ 
    Description: Images where the bride is to the left of the groom. \\ 
    Dataset: Wedding
    \item $\exists x. \forall y. \textsf{HasAttribute}(x, \texttt{bride}) \wedge \neg \textsf{HasAttribute}(y, \texttt{groom})$ \\ 
    Description: Images that contain the bride and not the groom. \\ 
    Dataset: Wedding
    \item $\exists x. \exists y. \textsf{HasAttribute}(x, \texttt{face}) \wedge \textsf{HasAttribute}(y, \texttt{glasses}) \wedge \textsf{HasRelation}(x, y, \texttt{contains})$ \\ 
    Description: Images with people wearing glasses.
    Dataset: Wedding
    \item $\exists x.\forall y. \textsf{HasAttribute}(x, \texttt{face}) \wedge (\textsf{HasAttribute}(y, \texttt{glasses}) \rightarrow \neg \textsf{HasRelation}(x, y, \texttt{contains}))$ 
    Description: Images with people not wearing glasses.
    Dataset: Wedding 
    \item $\exists x. \exists y. \text{HasAttribute}(x, \texttt{mouth\_open}) \wedge \text{HasAttribute}(y, \texttt{mouth\_open}) \wedge \textsf{HasRelation}(x, y, \texttt{next\_to})$ \\ 
    Description: Images with people talking. \\
    Dataset: Wedding 
    \item $\exists x. \exists y. \text{HasAttribute}(x, \texttt{mouth\_open}) \wedge \text{HasAttribute}(x, \texttt{smiling}) \wedge \text{HasAttribute}(y, \texttt{mouth\_open}) \wedge \text{HasAttribute}(y, \texttt{smiling}) \wedge \textsf{HasRelation}(x, y, \texttt{next\_to})$ \\ 
    Description: Images with people laughing. \\ 
    Dataset: Wedding
    \item $\exists x. \exists y. \textsf{HasAttribute}(x, \texttt{tie}) \wedge \textsf{HasAttribute}(y, \texttt{face}) \wedge \textsf{HasRelation}(x, y, \texttt{below})$ \\ 
    Description: Images with people wearing ties. \\ 
    Dataset: Wedding 
    \item $\exists x. \text{HasAttribute}(x, \texttt{suit}) \vee \textsf{HasAttribute}(x, \texttt{tie}) \vee \textsf{HasAttribute}(x, \texttt{gown})$ \\ 
    Description: Images with formal attire.  \\
    Dataset: Wedding 
    \item $\exists x. \textsf{HasAttribute}(x, \texttt{bride}) \wedge \textsf{HasAttribute}(x, \texttt{smiling})$ \\ 
    Description: Images where the bride is smiling. \\
    Dataset: Wedding 
    \item $\exists x. \exists y. \textsf{HasAttribute}(x, \texttt{guitar}) \wedge \textsf{HasAttribute}(y, \texttt{face}) \wedge \textsf{HasRelation}(x, y, \texttt{below})$ \\ 
    Description: Images with guitar players. \\
    Dataset: Festival
    \item $\exists x. \exists y. \textsf{HasAttribute}(x, \texttt{microphone}) \wedge \textsf{HasAttribute}(y, \texttt{face}) \wedge \textsf{HasRelation}(x, y, \texttt{below})$ \\ 
    Description: Images with singers. \\
    Dataset: Festival 
    \item $\exists x. \exists y. \exists z. \textsf{HasAttribute}(x, \texttt{microphone}) \wedge \textsf{HasAttribute}(y, \texttt{face}) \wedge \textsf{HasAttribute}(z, \texttt{guitar}) \wedge \textsf{HasRelation}(x, y, \texttt{below}) \wedge \textsf{HasRelation}(z, y, \texttt{below})$ \\ 
    Description: Images with people singing and playing guitar. \\ 
    Dataset: Festival
    \item $\exists x. \exists y. \textsf{HasAttribute}(x, \texttt{microphone}) \wedge \textsf{HasAttribute}(y, \texttt{face}) \wedge \textsf{HasRelation}(x, y, \texttt{below})$ \\ 
    Description: Images with smiling performers. \\
    Dataset: Festival
    \item $\exists x. \textsf{HasAttribute}(x, \texttt{speaker}) \vee \textsf{HasAttribute}(x, \texttt{microphone}) \vee \textsf{HasAttribute}(x, \texttt{guitar})$ \\
    Description: Images with stage equipment. \\ 
    Dataset: Festival 
    \item $\exists x. \exists y. \textsf{HasAttribute}(x, \texttt{bicycle}) \wedge \textsf{HasAttribute}(y, \texttt{face}) \wedge \textsf{HasRelation}(x, y \texttt{below})$ \\ 
    Description: Images with cyclists. \\ 
    Dataset: City Streets
    \item $\exists x. \forall y.  \textsf{HasAttribute}(x, \texttt{bicycle}) \wedge \neg \textsf{HasAttribute}(y, \texttt{car})$ \\ 
    Description: Images with bicycles but not cars. \\ 
    Dataset: City Streets
    \item $\exists x. \exists y. \textsf{HasAttribute}(x, \texttt{bicycle}) \wedge \textsf{HasAttribute}(y, \texttt{smiling}) \wedge \textsf{HasRelation}(x, y \texttt{below})$ \\ 
    Description: Images with happy cyclists. \\ 
    Dataset: City Streets
    \item $\exists x. \exists y. \exists z. \textsf{HasAttribute}(x, \texttt{bicycle}) \wedge \textsf{HasAttribute}(y, \texttt{face}) \wedge \textsf{HasAttribute}(y, \texttt{helmet}) \wedge \textsf{HasRelation}(x, y, \texttt{below}) \wedge \textsf{HasRelation}(y, z, \texttt{below})$ \\ 
    Description: Images with cyclists wearing helmets. \\ 
    Dataset: City Streets 
    \item $\exists x. \exists y. \textsf{HasAttribute}(x, \texttt{car}) \wedge \textsf{HasAttribute}(y, \texttt{face}) \wedge \textsf{HasRelation}(x, y, \texttt{contains})$ \\ 
    Description: Images with people driving cars. \\ 
    Dataset: City Streets 
    \item $\exists x. \forall y. \textsf{HasAttribute}(x, \texttt{face}) \wedge (\textsf{HasAttribute}(y, \texttt{car}) \rightarrow \neg \textsf{HasRelation}(y, x, \texttt{contains})$ \\ 
    Description: Images with people not driving cars. \\ 
    Dataset: City Streets
    \item $\exists x. \textsf{HasAttribute}(x, \texttt{bicycle}) \vee \textsf{HasAttribute}(x, \texttt{car})$ \\ 
    Description: Images with modes of transportation.
    Dataset: City Streets 
    \item $\exists x. \exists y. \textsf{HasAttribute}(x, \texttt{hat}) \wedge \textsf{HasAttribute}(y, \texttt{face}) \wedge \textsf{HasRelation}(x, y, \texttt{above})$ \\ 
    Description: Images with people wearing hats. \\
    Dataset: City Streets
\end{enumerate}

\section{List of \digits Benchmarks}\label{sec:mnistbenchmarks}

\begin{enumerate}
% Robert
\item $\lambda \ l$.\texttt{fold} \texttt{inc} 0 (\texttt{filter} (\texttt{curry} $\leq 5$) (\texttt{map} (\texttt{curry} \texttt{plus} $1$) 
(\texttt{map} \texttt{toDigit} $l$))) \\
Counts how many test scores are still low (below 5) after adding an extra credit point.

\item $\lambda \ l$.\texttt{fold} \texttt{max} 0 (\texttt{map} (\texttt{curry} \texttt{plus} (\texttt{toDigit} (\texttt{head} $l$)) (\texttt{filter} (\texttt{curry} $\leq 8$) (\texttt{map} \texttt{toDigit} (\texttt{tail} $l$))))) \\ 
Calculates the highest score after $x$ extra credit points are added to scores below $8$.

\item $\lambda l.$\texttt{fold} \texttt{plus} 0 (map (\texttt{curry} \texttt{plus} (\texttt{toDigit} (\texttt{head} $l$))) (\texttt{filter} (\texttt{curry} $\leq 6$) (\texttt{map} \texttt{toDigit} (\texttt{tail} $l$)))) \\ 
For cheap products being sold at a store, calculate the total revenue from selling the products after increasing the price by $x$ dollars.

\item \svhnprog \texttt{inc} 0 (\texttt{filter} \curry{\texttt{plus}}{\imgmap} (\texttt{map} \curry{\texttt{mult}}{2} \imgsmap)) \\
After doubling the amount of every product in the inventory, count how many have more than $x$ units in stock, where $x$ is the required minimum.

\item \svhnprog \texttt{max} 0 (\texttt{map} \curry{\texttt{plus}}{\imgmap} (\texttt{filter} \curry{$\leq$}{2} \imgsmap)) \\ 
Researchers at the North Pole took notes of readings from a thermometer that is inaccurate in extreme cold. For temperatures below 3 fahrenheit corrects the false temperature readings by increasing them by $x$. Then reports the highest temperature at which the sensor starts failing. 

% Shankara
\item \svhnprog \texttt{plus} 0 (\texttt{filter} \curry{$\geq$}{\imgmap} \imgsmap ) \\ 
Sum of digits greater than or equal to $x$.

\item \svhnprog \texttt{prod} 1 (\texttt{filter} \curry{$\geq$}{\imgmap} \imgsmap) \\ 
Product of digits less than $x$ after multiplying by 3. 

\item \svhnprog \texttt{inc} 0 (\texttt{filter} \curry{$\leq$}{8} (\texttt{filter} \curry{$\geq$}{8} \imgsmap ) ) \\ 
Computes the number of occurences of a single digit (8) in the list.

% Anirudh
\item \svhnprog \texttt{plus} \imgmap (\texttt{filter} \curry{$\geq$}{\imgmap} \imgsmap ) \\ 
Counts the total sum of donations exceeding a threshold $x$, plus an additional donation of $x$. 

\item \svhnprog \texttt{plus} 0 (\texttt{map} \curry{\texttt{prod}}{2} \imgsmap) \\ 
Calculates the total sum of donations that have been matched. 

\item \svhnprog \texttt{max} \imgmap (\texttt{filter} \curry{$\geq$}{5} \imgsmap) \\ 
Determines the maximum age of children in a day care over age 4. If there are no children over age 4, default to $x$.

\item \svhnprog \texttt{plus} 0 (\texttt{map} \curry{\texttt{plus}}{\imgmap} \imgsmap) \\ 
Calculates the inventory count after restocking.

\item \svhnprog \texttt{inc} 0 (\texttt{filter} \curry{$\leq$}{\imgmap} \imgsmap) \\ 
Counts the number of retail transactions below a limit $x$.

\item \svhnprog \texttt{prod} 1 (\texttt{map} \curry{\texttt{prod}}{7} \imgsmap) \\ 
Computes the product of a list of numbers after multiplying them by 7.

% Noah
\item \svhnprog \texttt{inc} 0 (\texttt{filter} \curry{$\geq$}{4} (\texttt{filter} \curry{$\leq$}{8} \imgsmap)) \\
Counts the number of participants in a study between the ages of 4 and 8.

\item \svhnprog \texttt{max} 0 (\texttt{filter} \curry{$\leq$}{2} \imgsmap) \\ 
Finds the maximum sub-2k dollar expenditure on a balance sheet.

\item \svhnprog \texttt{plus} 0 (\texttt{filter} \curry{$\geq$}{8} \imgsmap) \\ 
Finds the sum of expenditures over \$8k on a balance sheet.

\item \svhnprog \texttt{map} 0 (\texttt{map} \curry{\texttt{prod}}{2} \imgsmap) \\ 
Finds the size of the largest class if every class size were to double.

\item \svhnprog \texttt{max} 0 (\texttt{map} \curry{\texttt{prod}}{2} (\texttt{filter} \curry{$\geq$}{5} \imgsmap)) \\ 
Finds the size of the largest class, if every class with less than $5$ students dissolved and every other class size doubled.

% Zetten
\item \svhnprog \texttt{inc} 0 (\texttt{filter} \curry{$\geq$}{1} \imgsmap) \\ 
Counts the number number of non-zero digits in the list.

\item \svhnprog \texttt{prod} 1 \imgsmap \\ 
Calculates cumulative compound interest by years.

\item \svhnprog \texttt{max} 0 \imgsmap \\
Calculates maximum score on a midterm.

\item \svhnprog \texttt{inc} 0 (\texttt{filter} \curry{$\geq$}{\imgmap} \imgsmap) \\ 
Calculates number of people with a passing grade. 

% Ruijie
\item \svhnprog \texttt{inc} 0 (\texttt{filter} \curry{$\leq$}{0} \imgsmap) \\ 
Counts the number of 0's in the list.

\item \svhnprog \texttt{max} 0 (\texttt{filter} \curry{$\geq$} \imgsmap) \\ 
Finds the maximum among elements equaling at least 5, defaulting to 0 if there are no such elements.

\item \svhnprog \texttt{max} 0 (\texttt{filter} \curry{$\leq$}{\imgmap} \imgsmap) \\ 
Finds the highest failing grade on a test.
\item \svhnprog \texttt{plus} 0 (\texttt{map} \curry{\texttt{prod}}{6} \imgsmap) \\ 
Computes the total sales revenue by multiplying the quantity of items sold by unit price 6, and summing the results. 
\item \svhnprog \texttt{plus} 0 (\texttt{filter} \curry{$\geq$}{2} \imgsmap) \\ 
Counts the number of batteries in boxes containing at least 2 batteries.
\item \svhnprog \texttt{inc} 0 (\texttt{filter} \curry{$\geq$}{\imgmap} (\texttt{map} \curry{\texttt{plus}}{4} \imgsmap)) \\ 
Determines the number of students who would pass the course after adding 4 bonus points. 
\item \svhnprog \texttt{inc} 0 (\texttt{filter} \curry{$\geq$}{\imgmap} (\texttt{filter} \curry{$\geq$}{7} \imgsmap)) \\ 
Determines the number of objects within a size range suitable for a robot's gripper.
\item \svhnprog \texttt{inc} 0 (\texttt{filter} \curry{$\geq$}{2} \imgsmap) \\ 
Counts the number of conference attendees who have more than 1 dietary restriction.
\item \svhnprog \texttt{plus} 0 (\texttt{map} \curry{\texttt{prod}}{\imgmap} \imgsmap ) \\ 
Calculates the total weekly egg budget, where each person needs a specific number of eggs perday, and each egg costs $x$ dollars.
\item \svhnprog \texttt{max} 0 (\texttt{map} \curry{\texttt{plus}}{\imgmap} \imgsmap ) \\ 
Identifies the most expensive neutral oil in the grocery store, adding $x$ dollars of tax. 
\item \svhnprog \texttt{inc} 0 (\texttt{filter} \curry{$\geq$}{4} \imgsmap) \\ 
Counts the number of restaurants with ratings of at least 4 from a list of restaurant ratings.
\item \svhnprog \texttt{max} 0 (\texttt{map} \curry{\texttt{plus}}{5} \imgsmap) \\ 
Identifies the microwave recipe with the longest cooking time, adding 5 minutes for prep.
\item \svhnprog \texttt{inc} 0 (\texttt{filter} \curry{$\geq$}{5} \imgsmap) \\ 
Counts the total number of donations exceeding \$5.
\item \svhnprog \texttt{inc} 0 (\texttt{filter} \curry{$\geq$}{3} \imgsmap) \\ 
In a pool of PhD applications, counts how many have a GPA of at least 3.
\item \svhnprog \texttt{inc} 0 (\texttt{filter} \curry{$\leq$}{\imgmap} (\texttt{map} \curry{\texttt{plus}}{2} \imgsmap ) ) \\ 
Counts how many students cannot get a C grade (where the cutoff is $x$, after adding 2 bonus points.
\item \svhnprog \texttt{inc} 0 (\texttt{filter} \curry{$\leq$}{\imgmap} (\texttt{map} \curry{\texttt{plus}}{1} \imgsmap) ) \\ 
Counts how many reviewers have ratings below $x$, if each score were increased by 1.
\item \svhnprog \texttt{plus} 1 (\texttt{map} \curry{\texttt{max}}{8} (\texttt{filter} \curry{$\geq$}{4} \imgsmap ) ) \\ 
Calculates the sum of the prices of objects in a store, when all the products less than \$4 are removed, while remaining prices are increased to at least \$8.
\item \svhnprog \texttt{plus} 0 (\texttt{filter} \curry{$\leq$}{\imgmap} \imgsmap) \\
Calculates the total score among football teams scoring less than or equal to $x$ points.
\item \svhnprog \texttt{prod} 1 (\texttt{map} \curry{\texttt{max}}{1} \imgsmap) \\ 
Finds the product of a list, with 0's replaced with 1's. 
\item \svhnprog \texttt{prod} 1 (\texttt{filter} \curry{$\geq$}{2} (\texttt{filter} \curry{$\leq$}{4} \imgsmap)) \\ 
Finds the product of all numbers in a list between 1 and 5 exclusive.
\item \svhnprog \texttt{plus} 0 (\texttt{filter} \curry{\texttt{prod}}{3} \imgsmap) \\
Finds the total score in the last round of Family Feud, where points are worth triple.
\item \svhnprog \texttt{inc} 0 (\texttt{filter} \curry{$\geq$}{1} (\texttt{filter} \curry{$\leq$}{5} \imgsmap)) \\ 
Finds the number of non-zero race times less than or equal to 5. 
\item \svhnprog \texttt{plus} 0 (\texttt{filter} \curry{$\geq$}{3} \imgsmap) \\ 
Given a list of purchases at a cash register, gets the sum of transactions over \$2.
\item \svhnprog \texttt{inc} 0 (\texttt{filter} \curry{$\geq$}{\imgmap} (\texttt{map} \curry{\texttt{plus}{1}} \imgsmap)) \\ 
Count how many applicants have a GPA of at least $x - 1$. 
\item \svhnprog \texttt{plus} 0 (\texttt{map} \curry{\texttt{plus}}{2} (\texttt{map} \curry{\texttt{prod}}{\imgmap} \imgsmap)) \\ 
Computes the total cost of items in a box, where the cost per item is $x$, and the flat cost of the box is \$2. 
\item \svhnprog \texttt{max} 0 (\texttt{filter} \curry{$\leq$}{7} \imgsmap) \\ 
Given a list of people wanting various amounts of rice under 8kg, finds the one who requested the most rice. 
\item \svhnprog \texttt{max} 0 (\texttt{filter} \curry{$\leq$}{\imgmap} (\texttt{map} \curry{\texttt{prod}}{5} \imgsmap)) \\ 
Finds for the max weight of an item, such that you can put 5 of them on a lift with a maximum weight of $x$.
\end{enumerate}

\end{document}